\documentclass[11pt]{article}
\pdfoutput=1

\usepackage{amsmath}
\usepackage{graphicx}
\usepackage{enumerate}
\usepackage{natbib} 
\usepackage{url}  
\usepackage{hyperref}
\hypersetup{
	colorlinks=true,
	linkcolor=blue,
	filecolor=magenta,      
	urlcolor=cyan,
	citecolor=blue,
}
\usepackage{xcolor}
\usepackage{amssymb,bbm,authblk,graphicx}
\usepackage{amsthm}
\usepackage{algorithm}
\usepackage{algpseudocode}
\usepackage[margin=0.8in]{geometry}
\usepackage{multirow} 
\RequirePackage{amsfonts}
\usepackage{verbatim}
\usepackage{subfigure}
\usepackage{bm}
\usepackage{enumitem}
\usepackage{xr}
\usepackage{tablefootnote,authblk}
\newtheorem{theorem}{Theorem}
\newtheorem{lemma}{Lemma}

\newtheorem{assumption}{Assumption}
\newtheorem{definition}{Definition}

\newtheorem{remark}{Remark}

\newcommand{\bb}{\mathbb}
\newcommand{\tb}{\textbf}

\newcommand{\mni}{\frac{M_i}{N_i}}

\newcommand{\blind}{1}

\begin{document}

\def\spacingset#1{\renewcommand{\baselinestretch}%
{#1}\small\normalsize} \spacingset{1}

\if1\blind
{ \title{\bf \large Causal Inference under Interference: Regression Adjustment and Optimality}
\author[1]{Xinyuan Fan}
\author[2]{Chenlei Leng}  
\author[1]{Weichi Wu}

\affil[1]{Department of Statistics and Data Science, Tsinghua University}
\affil[2]{Department of Statistics, University of Warwick}
\date{}
  \maketitle

} \fi

\if0\blind
{
  \bigskip
  \bigskip
  \bigskip
  \begin{center}
    {\large Causal Inference under Interference: Regression Adjustment and Optimality}
\end{center}
  \medskip
} \fi

\begin{abstract}
In randomized controlled trials without interference, regression adjustment is widely used to enhance the efficiency of treatment effect estimation. This paper extends this efficiency principle to settings with network interference, where a unit's response may depend on the treatments assigned to its neighbors in a network. We make three key contributions: (1) we establish a central limit theorem for a linear regression-adjusted estimator and prove its optimality in achieving the smallest asymptotic variance within a class of linear adjustments; (2) we develop a novel, consistent estimator for the asymptotic variance of this linear estimator; and (3) we propose a nonparametric estimator that integrates kernel smoothing and trimming techniques, demonstrating its asymptotic normality and its optimality in minimizing asymptotic variance within a broader class of nonlinear adjustments. Extensive simulations validate the superior performance of our estimators, and a real-world data application illustrates their practical utility. Our findings underscore the power of regression-based methods and reveal the potential of kernel-and-trimming-based approaches for further enhancing efficiency under network interference.
\end{abstract}

\noindent%
{\it Keywords:}  Average Treatment Effect; Randomized Controlled Trials; Network Interference; Regression Adjustment; Central Limit Theorem
 
\spacingset{1.0}

\section{Introduction}
\label{sec:intro}
Consider a randomized controlled trial (RCT) where we observe data \(\{Y_i, W_i, \mathbf{z}_i\}_{i=1}^n\). For the \(i\)th unit, \(Y_i \in \mathbb{R}\) denotes the response variable, \(W_i \in \{0, 1\}\) is a binary treatment indicator with \(W_i \sim \text{Bernoulli}(\pi)\) for some \(\pi \in (0, 1)\), and \(\mathbf{z}_i \in \mathbb{R}^p\) is a \(p\)-dimensional vector of covariates. Our objective is on how to best estimate the average treatment effect (ATE), which quantifies the expected difference in responses had each unit been observed under both treatment and control conditions.

The challenge and relevance of this paper arise from the fact that the units in the study interfere with each other through a network structure. Consequently, a unit's response may be influenced not only by its own treatment assignment but also by the treatments assigned to its neighbors in the network. This setup deviates from the classical Stable Unit Treatment Value Assumption (SUTVA), which assumes that the potential outcomes for any unit are unaffected by the treatments assigned to other units. To account for this, we adopt the Neyman-Rubin causal model \citep{imbens2015causal}, representing the outcome for unit \(i\) as \(Y_i(W_1, \dots, W_n)\), emphasizing its dependence on both its own treatment and the treatments of others. 

For simplicity, we define the shorthand \(Y_i(x, W_{-i}) := Y_i(W_i = x, W_1, \dots, W_{i-1}, W_{i+1}, \dots, W_n)\), where \(W_{-i}\) denotes the treatment assignments of all units except for the \(i\)-th. The causal estimand of interest is the average treatment effect (ATE), also known as the direct effect of the treatment on the outcome \citep{savje2021average}, given by:
\begin{equation}\label{eq:ate}
	\tau_n = \frac{1}{n} \sum_{i=1}^n E[ Y_i(1, W_{-i}) - Y_i(0, W_{-i})],
\end{equation}
where the expectation is taken over the random treatment assignments. If the limit exists as \(n \to \infty\), and the network is treated as random, the ATE is defined as:
\[
\tau = \lim_{n \to \infty} \tau_n,
\]
which is the focus of the paper.

\subsection{Literature review}
A substantial body of literature addresses the estimation of the ATE under the no-interference assumption, where \(Y_i(W_1, \dots, W_n) \equiv Y_i(W_i)\). A widely employed method in this context is the regression-adjusted difference-in-means estimator:
\begin{equation}\label{eq:ra-dim}
	\hat{\tau}_n = \frac{1}{n_1} \sum_{W_i=1} \big[Y_i - \hat{\mu}_1(\mathbf{z}_i)\big] 
	- \frac{1}{n_0} \sum_{W_i=0} \big[Y_i - \hat{\mu}_0(\mathbf{z}_i)\big],
\end{equation}
where \(n_1 = \sum_{i=1}^n W_i\) and \(n_0 = \sum_{i=1}^n (1 - W_i)\) denote the numbers of treated and control units, respectively. The terms \(\hat{\mu}_k(\mathbf{z}_i)\), for \(k = 0, 1\), represent regression-adjusted  {residuals}  
obtained from separate regressions of \(Y_i\) on covariates \(\mathbf{z}_i\) within the treated (\(W_i = 1\)) and control (\(W_i = 0\)) groups. In the absence of regression adjustment (\(\hat{\mu}_0 = \hat{\mu}_1 = 0\)), this estimator simplifies to the standard difference-in-means estimator denoted  henceforth as $\hat\tau_{dim}$.

\cite{lin2013agnostic} first demonstrated that linear regression adjustment, conducted separately for each group, can yield significant variance reduction. \cite{li2017general} further established the optimality of this approach for causal estimators. Subsequent studies, including \citet{liu2020regression, su2021model, zhao2022reconciling,lu2024adjusting}, and \citet{liu2024randomization}, have extended and refined these methods, highlighting their robustness and efficacy in practical applications.

Given the success of regression adjustment in settings without interference, a natural question arises:
\begin{center}
	\textit{Can regression adjustment be effectively applied to estimate the ATE in the presence of network interference, and if so, how?}
\end{center}

This question has gained increasing importance alongside a rapidly growing body of literature on causal inference under interference, spanning fields such as economics, social science, and epidemiology \citep{duflo2013truth, cai2015social, paluck2016changing, sato2019peer, bragancca2022extension}. Several studies have developed methods for estimating the ATE without incorporating covariates \citep{aronow2017estimating, athey2018exact, savje2021average, chin2018central, gao2023causal, leung2022causal}, often under the assumption of deterministic networks with randomness arising solely from treatment assignments. 
In a related direction, \cite{li2022random} explored the interplay between  {treatment assignments}  
and a graphon  {network} model but did not address regression adjustment. Their work established that the difference-in-means estimator remains asymptotically normal under suitable conditions with network interference, providing a foundation for further methodological developments. 
Compared to these approaches, our contribution seeks to extend this line of work by integrating regression adjustment in the presence of network interference.  While some prior works have considered regression adjustment, their limitations are notable. For instance, \cite{aronow2017estimating} provided no theoretical guarantees, and \cite{gao2023causal} failed to demonstrate efficiency gains over the standard difference-in-means estimator despite employing a more complex framework. {Additionally, our Theorems~\ref{th: Consistency of hat tau}, \ref{th:main}, and \ref{th:5} apply to both dense and sparse graphs, whereas \citet{leung2022causal} and \citet{gao2023causal} are restricted to the sparse graph setting.}

\subsection{Our contributions}
The primary aim of this paper is to address a critical gap in the literature by establishing the optimality of two regression-adjusted causal estimators, within linear-regression and nonparametric regression adjustments respectively, directly responding to the question posed. Our contributions advance the study of treatment effect estimation in the presence of network interference, and are outlined as follows:

\begin{enumerate}
	\item \textbf{Optimal Variance Reduction in Linear Adjustments.}  
	Building on the framework of \cite{li2022random}, we model the underlying network using a graphon structure. Within this framework, we analyze linear regression adjustment, establish the asymptotic normality of our estimator, and demonstrate its optimality in terms of asymptotic variance reduction compared to the difference-in-means estimator within the class of linear regression adjustments. Our results extend the findings of \cite{lin2013agnostic} and \cite{li2017general}, bridging the gap from non-interfering settings to those with network interference. For further details, refer to Theorem~\ref{th:main}.
	\item \textbf{Justified Variance Estimation.}  
	We introduce a novel, consistent estimator for the asymptotic variance of our estimator, utilizing a low-rank graphon model for the network. This enhancement refines the heuristic variance estimator from \cite{li2022random}, providing a theoretically grounded alternative. See Theorem~\ref{th:3} for details. 
	\item \textbf{Advances in Nonparametric Adjustment.}  
	We extend the results established for linear regression adjustments to include nonparametric adjustments, establishing the optimality of our estimator's variance reduction within a broader class of adjustments, as formally defined in (\ref{def:class of conditional expectation}). See Theorems~\ref{th:6} and \ref{th:5}. Our nonparametric estimator innovatively combines kernel regression and trimming techniques, achieving $\sqrt{n}$-asymptotic normality for any fixed-dimensional covariates, thereby overcoming the ``curse of dimensionality".    
	\item   	\textbf{Novelty in Theory}  In this paper, we make significant theoretical contributions in three key areas.
	\begin{itemize}
		\item First, we introduce a plug-in estimator for low-rank graphon networks (see equation (\ref{def:hat b})) and demonstrate that the appropriately scaled weighted sum of node neighbors can uniformly approximate the corresponding graphon function (see Lemma~\ref{lem:11}). This provides a novel theoretical foundation for using adjacency matrices as plug-in estimators in low-rank graphon networks, marking a critical advancement in understanding their asymptotic behavior.
		\item 
		Second, to establish the asymptotic normality of the nonparametric adjustment (\ref{def:hat tau np}), we leverage U-statistics and their approximation techniques, addressing both the independent components and dependencies induced by the random network structure. This approach delivers the necessary theoretical guarantees and opens the door to extending these results to other types of random networks with minimal adjustments.
		\item 
		Third, in the analysis of the nonparametric adjustment estimator, we introduce an innovative projection technique to handle components associated with the exposure graph. This method complements the classical projection approach for U-statistics and serves as a crucial tool in the analysis. For a more detailed discussion, see Lemma~\ref{lem4th6:6} and Lemma~\ref{lem4th6:9}. Furthermore, this technique lays the groundwork for the further exploration of other nonparametric estimators.
	\end{itemize}
	\item \textbf{Assurance in Numerical Study.} We have developed explicit forms of the asymptotic variances of our regression-adjusted estimators.  Leveraging this, we conduct extensive numerical simulations to validate our approach. Remarkably, we show that our nonparametric estimator delivers highly satisfactory results, even for covariates with dimensions as large as $10$. This performance is particularly noteworthy, considering the inherent challenges of nonparametric smoothing at such dimensions.
\end{enumerate} 

\noindent
\textbf{Notations}.  For a real number \(a\), \(\lceil a \rceil\) denotes the smallest integer that is greater than or equal to \(a\). For two positive real numbers \(a\) and \(b\), let \(a \vee b = \max(a, b)\) and \(a \wedge b = \min(a, b)\).  {For a sequence of real vectors $x_1,\cdots,x_n$, let $\bar{x}=\sum_{i=1}^n x_i/n$.} For a \(p\)-dimensional vector \(\mathbf{x} = (x_1, \ldots, x_p)^\top\), its \(\ell_q\)-norm is defined as \(|\mathbf{x}|_q = \left( \sum_{i=1}^p |x_i|^q \right)^{1/q}\), and the maximum norm is \(|\mathbf{x}|_\infty = \max_i |x_i|\). {We denote $\mathbf{x}^{(-1)}=(x_2, \ldots, x_p)^\top$ as the vector leaving out the first element of $\mathbf{x}$.} For a nonzero scalar \(h\), we define \(\mathbf{x}/h = (x_1/h, \ldots, x_p/h)^\top\). For a matrix $\bb{X}$, write $\bb{X}\succeq 0$ if $\bb{X}$ is positive semidefinite. For a random variable \(x\), its \(\mathcal{L}_q\)-norm is \(\|x\|_{\mathcal{L}_q} = \left( E |x|^q \right)^{1/q}\). {For a random vector $\mathbf{z}$, its covariate matrix is given by $cov(\mathbf{z})=\bb{E}(\mathbf{z}\mathbf{z}^\top)-\bb{E}(\mathbf{z})\bb{E}(\mathbf{z}^\top)$.} For a differentiable bivariate function \(f_i(x, y)\), we use the shorthand notation \(f_i^{(0)}(x, y) = f_i(x, y)\) and \(f_i^{(k)}(x, y) = \frac{\partial^k}{\partial y^k} f_i(x, y)\) for \(k \geq 1\). For a multivariate function \(f(x_1, \ldots, x_p)\), its gradient is denoted by \(\nabla f = \left( \frac{\partial f}{\partial x_1}, \ldots, \frac{\partial f}{\partial x_p} \right)^\top\). For sequences \(a_n\) and \(b_n\), we write \(a_n \gtrsim b_n\) (or \(a_n \lesssim b_n\)) if there exist constants \(C\) and \(N\) such that \(a_n \geq C b_n\) (or \(a_n \leq C b_n\)) for all \(n \geq N\). Lastly, \(a_n \lesssim_p b_n\) means \(a_n = O_p(b_n)\).

\noindent
\textbf{Organization}. 
The remainder of this paper is organized as follows. Section~\ref{sec:Regression-based estimators} introduces the linear regression-adjusted estimator, establishing its asymptotic properties and demonstrating its variance optimality. In Section~\ref{sec:3.1}, we present a consistent estimator for its asymptotic variance. Section~\ref{sec:Local constant estimator} introduces a nonparametric regression-adjusted estimator and provides its properties. Section~\ref{sec:Numerical Evaluation2} presents numerical simulations, while  Section~\ref{sec:A naturalistic simulation with contact data} presents a naturalistic simulation mimicking a real-world network dataset to estimate vaccine effectiveness.  Section~\ref{sec:Conclusion} summarizes the paper. All proofs and technical details are provided in the  supplementary material.

\section{Regression-based Estimators}
\label{sec:Regression-based estimators}
We build on the random graph framework introduced by \cite{li2022random}, extending it to incorporate regression adjustment. Specifically, we assume that units interfere through a network, represented without loss of generality as an undirected random graph  \(G_n = (V_n, E_n)\). Here, \(V_n = \{1, \dots, n\}\) denotes the set of vertices (units), and \(E_n\) is the set of edges such that \(E_{n,ij} = E_{n,ji} = 1\) if units \(i\) and \(j\) are connected, and \(E_{n,ij} = 0\) otherwise. In general, under unrestricted interference, the argument \(Y_i(W_1, \dots, W_n)\) can take up to \(2^n\) possible values. Consequently, ATE as defined in \eqref{eq:ate} is not directly estimable due to the exponential growth in complexity. To overcome this challenge, we adopt a simplifying assumption commonly used in the literature (e.g., \citet{li2022random,park2024minimum, viviano2024policy}): a unit's potential outcome depends only on its own treatment status and the fraction of its treated neighbors \citep{hudgens2008toward}. This assumption significantly reduces the dimensionality of the problem while retaining a realistic representation of network interference. This leads to the following assumption:
\begin{assumption}[Anonymous Interference]  
	\label{assumption:anonymous interference}  
	The potential outcomes satisfy  
	\[
	Y_i = f_i(W_i, M_i/N_i),
	\]
	where \( N_i = \sum_{j=1}^n E_{ij} \) is the number of neighbors of unit \( i \), and \( M_i = \sum_{j=1}^n E_{ij}W_j \) is the number of treated neighbors. The function \( f_i \in \mathcal{F} \) may also depend on covariates \( \mathbf{z}_i \), and the pairs \( (f_i, \mathbf{z}_i) \) are assumed to be i.i.d. from a distribution on \( \mathcal{F} \times P_z \).  
\end{assumption}  

Allowing \( f_i \) to be random accounts for factors such as measurement errors. With this assumption,
\[
\tau_n = \frac{1}{n} \sum_{i=1}^n E[ f_i(1, M_i/N_i) - f_i(0, M_i/N_i) ].
\]
The second argument of \( f_i \) represents the fraction of treated neighbors, which, under suitable conditions, satisfies \( M_i/N_i \approx \pi \).
To ensure a well-defined limit as \( n \to \infty \), we model the network using a graphon framework. The graphon model provides a flexible and scalable representation of complex networks and is fundamental for characterizing exchangeable network structures \citep{lovasz2012large}.  

\begin{definition}[Graphon Model]
	\label{def:Graphon model}
	A graph \(G_n = (V_n, E_n)\) is said to follow a graphon model if the edges are generated as 
	\[
	E_{n,ij} \sim \text{Bernoulli}(\rho_n h(U_i, U_j) \wedge 1),
	\]
	where \(U_i \sim U(0,1)\) are independent and identically distributed latent variables, \(\rho_n > 0\) is a sparsity parameter, and \(h: [0,1]^2 \to [0, \infty)\) is a symmetric, measurable function referred to as the graphon.
\end{definition}

We impose the following assumption on the graphon model:

\begin{assumption}[Random Graph]
	\label{assump:random_graph}
	The graph \(G_n = (V_n, E_n)\) follows an undirected graphon model as defined in Definition~\ref{def:Graphon model}, with sparsity parameter \(\rho_n\) and graphon \(h(\cdot, \cdot)\). We assume the following:  
	\begin{enumerate}
		\item There exist constants \(c_l > 0\) and \(c_u > 0\) such that:  
		\[
		\inf_{x \in [0,1]} \int_0^1 h(x, y) \, dy \geq c_l \quad \text{and} \quad \sup_{x, y \in [0,1]} h(x, y) \leq c_u.
		\]
		\item The sparsity parameter satisfies \(\rho_n \sup_{x, y \in [0,1]} h(x, y) \leq 1\).  
		\item As \(n \to \infty\), the sparsity grows such that \(\sqrt{n} \rho_n \to \infty\).  
	\end{enumerate}
\end{assumption}

This assumption places structural constraints on the random graph. Notably, our framework accommodates both dense graphs (\(\rho_n \gtrsim 1\)) and sparse graphs (\(\rho_n \to 0\)). The condition \(\sqrt{n} \rho_n \to \infty\) aligns with that in \cite{li2022random}, ensuring that a sufficiently large number of subjects can influence the \(i\)-th subject, which is critical for meaningful inference in the presence of network interference. 
Finally, we impose a moment condition on \(f_i\) and \(\mathbf{z}_i\). Write $\mathbf{x}_i=(1, \mathbf{z}_i^\top)^\top$.

\begin{assumption}[Moment Conditions]
	\label{assumption:moment_conditions}
	The following conditions hold: $E(\mathbf{x}\mathbf{x}^\top) \succeq 0$; \(E|\mathbf{x}_1|_\infty^4 < \infty\); and 
	\[
	m_2 = \max_{
		k=0,1,2,3, \, w=0,1
	} \left\| \sup_{u \in [0,1]} |\mathbf{x}_1 f_1^{(k)}(w,u)|_\infty \right\|_{\mathcal{L}_2} < \infty.
	\]
\end{assumption}

This assumption imposes constraints on the moments of the outcome function \(f_i\), and is less stringent than the corresponding assumption in \cite{li2022random} that requires\\
\(
\sup_{k=0,1,2,3, \, w=0,1, \, u \in [0,1], \, f \in \mathcal{F}} |f^{(k)}(w, u)| \leq B
\)
for some finite constant \(B\). To see why, consider a scenario where \(f_i(w, u) = f_0(w, u) + \epsilon_i\), with \(f_0(w, u)\) being a bounded function and \(\epsilon_i\) representing an i.i.d. zero-mean error term that is unbounded. In this case, the boundedness condition in \cite{li2022random} would fail due to the unbounded nature of \(\epsilon_i\). In contrast, our moment conditions are satisfied, as they impose weaker constraints, accommodating unbounded noise or measurement errors.

Our linear regression-adjusted estimator is given by \eqref{eq:ra-dim}, where \(\hat\mu_k(\mathbf{z}_i), k = 0, 1\), is obtained via least-squares regression of \(Y_i\) on \(\mathbf{z}_i\), including an additional intercept for the treated (\(W_i = 1\)) and control (\(W_i = 0\)) groups, respectively.  
Formally, let \(\bb{X} = (\mathbf{x}_1, \dots, \mathbf{x}_n)^\top\) denote the matrix of regressors, \(\mathbf{Y} = (Y_1, \dots, Y_n)^\top\) the vector of outcomes, and \(\bb{W} = \text{diag}(W_1, \dots, W_n)\) the diagonal matrix of treatment indicators. The least-squares coefficients for the treated and control groups are given by
\[
\hat{\beta}_1 = (\bb{X}^\top \bb{W} \bb{X})^{-1} \bb{X}^\top \bb{W} \mathbf{Y}, \quad \hat{\beta}_0 = (\bb{X}^\top (\bb{I} - \bb{W}) \bb{X})^{-1} \bb{X}^\top (\bb{I} - \bb{W}) \mathbf{Y},
\]
where \(\bb{I}\) is the \(n \times n\) identity matrix. 
The fitted values \(\bb{X} \hat{\beta}_1\) and \(\bb{X} \hat{\beta}_0\) represent the expected outcomes \(E[Y_i(1, W_{-i})]\) and \(E[Y_i(0, W_{-i})]\), respectively. Our estimator is then defined as
\begin{equation}
	\label{def:regression estimator}
	\hat{\tau} = \frac{1}{n} \mathbbm{1}^\top \bb{X} (\hat{\beta}_1 - \hat{\beta}_0),
\end{equation} 
where \(\mathbbm{1}\) is the \(n\)-dimensional vector of ones. 
Equivalently, \(\hat{\tau}\) can be interpreted as the coefficient of \(W_i\) in an ordinary least squares regression of \(Y_i\) on \(W_i\),  {\(\mathbf{z}_i-\bar{\mathbf{z}}\), and \(W_i \times (\mathbf{z}_i-\bar{\mathbf{z}})\)}, as described in \cite{lin2013agnostic}. When no confusion arises, we will simply refer to $\hat{\tau}$ as regression-adjusted estimator.

We now present the results on the consistency and asymptotic normality of the regression estimator \(\hat{\tau}\). Let the population regression coefficients be defined as  
\[
\beta_1 = \bb{M}_{xx}^{-1} E(\mathbf{x}_1 f_1(1, \pi)) \quad \text{(treatment group)}, \quad \beta_0 = \bb{M}_{xx}^{-1} E(\mathbf{x}_1 f_1(0, \pi)) \quad \text{(control group)},
\]  
where \(\bb{M}_{xx} = E(\mathbf{x}_1 \mathbf{x}_1^\top)\) is the second moment matrix of the covariates.

\begin{theorem}[Consistency of \(\hat{\tau}\)]  
	\label{th: Consistency of hat tau}  
	Under Assumptions \ref{assumption:anonymous interference}, \ref{assump:random_graph}, and \ref{assumption:moment_conditions}, the following results hold as \(n \to \infty\):  
	\[
	\tau_n \to \tau = E(f_i(1, \pi) - f_i(0, \pi)), \quad \hat{\beta}_1 \overset{p}{\to} \beta_1, \quad \hat{\beta}_0 \overset{p}{\to} \beta_0, \quad \text{and} \quad \hat{\tau} \overset{p}{\to} \tau.
	\]  
\end{theorem}
\begin{theorem}[Asymptotic normality of \(\hat{\tau}\)]  
	\label{th:main}  
	Under the same assumptions as in Theorem \ref{th: Consistency of hat tau}, the following result holds:  
	\[
	\sqrt{n}(\hat{\tau} - \tau) \overset{d}{\to} N(0, V_{reg}),
	\]  
	where the asymptotic variance \(V_{reg}\) is given by  {
		\begin{align}  
			\label{def:Vreg}  
			V_{reg} &= \frac{1}{\pi} E\big(f_1(1, \pi) - \mathbf{x}_1^\top \beta_1\big)^2  
			+ \frac{1}{1-\pi} E\big(f_1(0, \pi) - \mathbf{x}_1^\top \beta_0\big)^2 \notag \\  
			&\quad + (\beta_1^{(-1)} - \beta_0^{(-1)})^\top 
			cov(\tb{z}_1)
			(\beta_1^{(-1)} - \beta_0^{(-1)})  
			+ b \pi(1-\pi)\big(E f_1^{(1)}(1, \pi) - E f_1^{(1)}(0, \pi)\big)^2,  
		\end{align}  
	}
	with  
	\begin{align}  
		\label{def:b}  
		b = E\left[E\left(\frac{h(U_i, U_j)}{E\big(h(U_i, U_j) \mid U_j\big)} \bigg| U_i\right)\right]^2.  
	\end{align}  
	
	Let \(V_{dim}\) denote the asymptotic variance of the difference-in-means estimator (obtained by taking \(\mathbf{x}_i = 1\) as the regressor). The relationship between \(V_{dim}\) and \(V_{reg}\) is  
	\begin{align}  
		\label{Vdim Vreg}  
		V_{dim} - V_{reg} = \frac{1}{\pi(1-\pi)}\big(E(\mathbf{z}_1 \tilde{f}) - E\mathbf{z}_1 E\tilde{f}\big)^\top  
		\big(E(\mathbf{z}_1 \mathbf{z}_1^\top)\big)^{-1}  
		\big(E(\mathbf{z}_1 \tilde{f}) - E\mathbf{z}_1 E\tilde{f}\big) \ge  0,  
	\end{align}  
	where \(\tilde{f} = (1-\pi)f_1(1, \pi) + \pi f_1(0, \pi)\).  
\end{theorem}
Theorem \ref{th:main} highlights several key insights. First, regardless of the choice of \(\mathbf{z}_i\), the regression estimator \(\hat{\tau}\) is always asymptotically normal with rate \(\sqrt{n}\). Furthermore, when there are no covariates, i.e., \(\mathbf{x}_i = 1\), this result is consistent with Theorem 4 in \cite{li2022random}. 
Second, \(V_{reg}\) is always less than or equal to \(V_{dim}\), with the inequality being strict unless the covariance between \(\mathbf{z}_1\) and \((1-\pi)f_1(1,\pi) + \pi f_1(0,\pi)\) is zero. This suggests that the regression adjustment leads to variance reduction whenever the covariates are informative. Similar conclusions have been drawn in other contexts (\cite{lin2013agnostic, gao2023causal}), but, to our knowledge, this is the first such result under the random graph asymptotic setup. 
Third, when \(\mathbf{x}_i\) is variable, the third term in equation (\ref{def:Vreg}), namely \(b \pi (1-\pi) (E f_1^{(1)} (1,\pi) - E f_1^{(1)} (0,\pi))^2\), remains constant and equals zero when there is no network interference. This term can be viewed as an ``irreducible component'' arising from network interference, while the first two terms contribute to variance reduction. 
In Section \ref{sec:Local constant estimator}, we extend this analysis to the local constant estimator, which exhibits a similar structure.

Our regression estimator \(\hat{\tau}\) is optimal for estimating the average treatment effect in the asymptotic sense, as it minimizes the variance among linear regression adjustment estimators. To demonstrate this, consider the class of regression adjustment estimators defined as:
\begin{equation}
	\label{def:class of regression adjustment}
	\hat{\tau}(\alpha_1, \alpha_0) = \sum_{i=1}^n \left( \frac{W_i(Y_i - \alpha_1^\top (\mathbf{z}_i - \bar{z}))}{\sum_{j=1}^n W_j} - \frac{(1-W_i)(Y_i - \alpha_0^\top (\mathbf{z}_i - \bar{z}))}{\sum_{j=1}^n (1-W_j)} \right),
\end{equation}
where \(\alpha_1, \alpha_0\) are two vectors {and $\bar{z}=\sum_{i=1}^n \tb{z}_i/n$.} Clearly, the estimator \(\hat{\tau}\) corresponds to \(\hat{\tau}(\hat{\beta}_1^{(-1)}, \hat{\beta}_0^{(-1)})\).  

We now present the following result:

\begin{theorem}
	\label{th:4}
	Under the assumptions in Theorem \ref{th: Consistency of hat tau}, for any fixed vectors \(\alpha_1\) and  \(\alpha_0\), we have:
	\[
	\sqrt{n}(\hat{\tau}(\alpha_1, \alpha_0) - \tau) \overset{d}{\to} N\left(0, \tilde{V}(\alpha_1, \alpha_0)\right),
	\]
	where {
		\begin{align}
			\label{eqn:th3 2}
			\tilde{V}(\alpha_1, \alpha_0) - V_{reg} = \frac{1}{\pi(1-\pi)} u(\alpha_1, \alpha_0)^\top 
			cov(\tb{z}_1)
			u(\alpha_1, \alpha_0) \ge  0,
		\end{align}
		and \(u(\alpha_1, \alpha_0) = (1 - \pi)(\alpha_1  - \beta_1^{(-1)}) + \pi(\alpha_0 - \beta_0^{(-1)})\).
	}
\end{theorem}
This reveals several important results. For any given $\alpha_1$ and $\alpha_0$, the regression adjustment estimator $\hat{\tau}(\alpha_1, \alpha_0)$ is asymptotically normal, with an explicit expression for its asymptotic variance. Furthermore, when $u(\alpha_1, \alpha_0) = 0$, a typical case being $\alpha_1 = \beta_1^{(-1)}$ and $\alpha_0 = \beta_0^{(-1)}$, $\tilde{V}(\alpha_1, \alpha_0)$ reaches its minimum, which equals $V_{reg}$, the asymptotic variance of the regression-adjusted estimator defined in \eqref{def:Vreg} . It is worth noting that $V_{reg}$ is the asymptotic variance of $\hat{\tau}(\hat{\beta}_1^{(-1)}, \hat{\beta}_0^{(-1)})$, which itself does not belong to the class $\hat{\tau}(\alpha_1, \alpha_0)$ because $\hat{\beta}_1$ and $\hat{\beta}_0$ depend on $Y$, $\mathbf{z}$, and $W$.  Theorem~\ref{th:4} shows that we can treat $\hat{\beta}_1$ and $\hat{\beta}_0$ as if they were the true $\beta_1$ and $\beta_0$ without affecting the asymptotic variance. Our results are consistent with similar findings in the absence of network interference \citep{li2017general}.

\section{Variance Estimation}
\label{sec:3.1}

In the absence of covariates, \citet{li2022random} proposed a variance estimator for their ATE estimator, designed to provide an upper bound on the true variance. In this section, we address this conservative strategy by utilizing the widely-adopted graphon model framework to facilitate variance estimation in the more complex setting of network interference with regression adjustment. A similar approach was previously employed by \citet{li2022random} in the context of their PC balancing estimator, though their focus was on estimating a different causal estimand, namely the indirect effect. 

The asymptotic variance \( V_{reg} \) of \( \hat{\tau} \), as derived in Theorem~\ref{th:main}, can be decomposed into four components, which jointly capture information from both the outcome function \( f_i \) and the graphon structure. Key parameters such as \( \beta_0 \), and \( \beta_1 \) can be consistently estimated as 
\( \hat{\beta}_0 \), and \( \hat{\beta}_1 \), respectively. As a result, most terms in the variance decomposition can be consistently estimated using their empirical counterparts. The main exceptions are \( b \), \( E[f_1^{(1)}(1, \pi)] \), and \( E[f_1^{(1)}(0, \pi)] \), which require additional techniques for accurate estimation.
{\begin{remark}
		In the absence of network interference, \citet{lin2013agnostic} introduced a robust sandwich standard error estimator. In our setting, the most challenging component of the asymptotic variance is the final term in Equation~\eqref{def:Vreg}, which reflects the influence of the network structure. While a sandwich estimator could be used to approximate the first three terms in \eqref{def:Vreg}, we choose to omit this approach for simplicity and leave its exploration for future research.
	\end{remark}
}

Assume that the graphon \( h(x, y) \) has a low rank \( r \), such that it can be expressed as  
\begin{equation}
	\label{eqn:graphon rank}
	h(x, y) = \sum_{k=1}^r \lambda_k \psi_k(x) \psi_k(y),
\end{equation}
where \( \lambda_k \) are eigenvalues satisfying \(|\lambda_1| \geq |\lambda_2| \geq \cdots \geq |\lambda_r| > 0\), and \( \psi_k \) are the corresponding eigenfunctions. The eigenfunctions \(\psi_k\) satisfy the orthonormality conditions:  
\[
\int_0^1 \psi_k^2(x) \, dx = 1, \quad \text{and} \quad \int_0^1 \int_0^1 \psi_k(x) \psi_l(y) \, dy = 0 \quad \text{for } k \neq l.
\]
The low-rank assumption is a widely accepted and well-justified modeling approach for random networks, with the stochastic block model being a notable example. This assumption is also made in \citet{li2022random}. 
To approximate the eigenfunctions \( \psi_k \) for \( k = 1, \dots, r \), we perform a singular value decomposition (SVD) on the adjacency matrix \( E_n \). Denote the leading \( r \) eigenpairs of \( E_n \) as \( (\hat{\lambda}_k, \hat{\psi}_k) \), ordered such that \( |\hat{\lambda}_1| \geq |\hat{\lambda}_2| \geq \cdots \geq |\hat{\lambda}_r| \).

To estimate \( b \) as defined in (\ref{def:b}), we treat the adjacency matrix \( E_n \) as a perturbed version of the connection probability matrix and use the following plug-in estimator:  
\begin{equation}
	\label{def:hat b}
	\hat{b} = \frac{1}{n} \sum_{i=1}^n \left( \sum_{j=1}^n \frac{E_{ij}}{\sum_{k=1}^n E_{jk}} \right)^2.
\end{equation}
In Lemma~\ref{lem:1}, we prove that under suitable assumptions, \(\hat{b} \overset{p}{\to} b\).  This proof leverages a novel technical result, specifically Lemma~\ref{lem:11}, which addresses low-rank graphon networks. In particular, we demonstrate that an appropriately scaled weighted sum of a node's neighbors can uniformly approximate the corresponding graphon function. For instance, Lemma~\ref{lem:11} establishes that  
\[
\sup_{i=1,\cdots,n}\left| 
\frac{N_i}{n\rho_n} - \int_0^1 h(U_i, y) \, dy
\right| = o_p(1).
\]

To estimate \( E f_1^{(1)} (1,\pi) \) and \( E f_1^{(1)} (0,\pi) \), we adapt the PC balancing estimator. Specifically, we estimate \( E f_1^{(1)} (1,\pi) \) as: 
\begin{align}
	\label{eqn:dev1}
	\hat{E} f_1^{(1)} (1,\pi) = \frac{1}{n \pi} \sum_{i=1}^n W_i Y_i \left( \frac{M_i}{ {\pi}} - \frac{N_i - M_i}{1 -  {\pi}} + \sum_{k=1}^r \hat{a}_k \hat{\psi}_{k i} \right),
\end{align}
where  \( \hat{a}_k \) are coefficients determined by the constraint: 
\begin{align*}
	\sum_{i=1}^n \hat{\psi}_{l i} \left( \frac{M_i}{ {\pi}} - \frac{N_i - M_i}{1 -  {\pi}} + \sum_{k=1}^r \hat{a}_k \hat{\psi}_{k i} \right) = 0,
\end{align*}
for all \( l = 1, \dots, r \).  
Similarly, the estimator for \( E f_1^{(1)} (0,\pi) \) is given by: 
\begin{align}
	\label{eqn:dev2}
	\hat{E} f_1^{(1)} (0,\pi) = \frac{1}{n (1 -  {\pi})} \sum_{i=1}^n (1 - W_i) Y_i \left( \frac{M_i}{ {\pi}} - \frac{N_i - M_i}{1 -  {\pi}} + \sum_{k=1}^r \hat{a}_k \hat{\psi}_{k i} \right).
\end{align}

Putting everything together, the asymptotic variance of \(\hat{\tau}\) is consistently estimated by  
\begin{align}
	\label{eqn:4}
	\hat{V}_{reg} = & \frac{1}{ \pi} \frac{\sum_{i=1}^n W_i (Y_i - \mathbf{x}_i^\top \hat{\beta}_1)^2}{\sum_{i=1}^n W_i} 
	+ \frac{1}{1-\hat\pi} \frac{\sum_{i=1}^n (1 - W_i)(Y_i - \mathbf{x}_i^\top \hat{\beta}_0)^2}{\sum_{i=1}^n (1 - W_i)} \notag\\
	&{ + (\hat{\beta}_1^{(-1)} - \hat{\beta}_0^{(-1)})^\top \left( \frac{1}{n} \sum_{i=1}^n \mathbf{z}_i \mathbf{z}_i^\top - \left(\frac{1}{n} \sum_{i=1}^n \mathbf{z}_i \right) \left( \frac{1}{n} \sum_{i=1}^n \mathbf{z}_i^\top \right) \right) (\hat{\beta}_1^{(-1)} - \hat{\beta}_0^{(-1)})
	} \notag\\
	& + \hat{b}  \pi (1 -  \pi) \left( \hat{E} f_1^{(1)}(1,  \pi) - \hat{E} f_1^{(1)}(0,  \pi) \right)^2.
\end{align}

The following assumption is required to ensure the consistency of the proposed estimators.

\begin{assumption}
	\label{ass:main2}
	The sparsity parameter \(\rho_n\) satisfies \(\rho_n \to 0\) as \(n \to \infty\). The graphon \(h(x, y)\) is represented in the form (\ref{eqn:graphon rank}) with rank \(r\). Additionally, there exists a constant \(M_1 > 0\) such that 
	\[
	\sup_{x \in [0,1]} \max_{k = 1, \dots, r} |\psi_k(x)| \leq M_1.
	\]
\end{assumption}

The sparsity condition above aligns with the structure of most real-world networks, which are typically sparse. With this assumption, we establish the following result:
\begin{theorem}
	\label{th:3}
	Under Assumptions \ref{assumption:anonymous interference}, \ref{assump:random_graph}, \ref{assumption:moment_conditions}, and \ref{ass:main2}, the estimator \(\hat{V}_{reg}\) satisfies \(\hat{V}_{reg} \overset{p}{\to} V_{reg}\).
\end{theorem}

By combining the above theorem and Theorem~\ref{th:main}, we derive a \(1 - \alpha\) asymptotic confidence interval for \(\tau\) as
\[
\left[\hat{\tau} - z_{\alpha/2} \sqrt{\hat{V}_{reg}/n}, \hat{\tau} + z_{\alpha/2} \sqrt{\hat{V}_{reg}/n}\right],
\]
where \(z_{\alpha/2}\) is the \(\alpha/2\)-quantile of the standard normal distribution. This estimator should be compared with the results in \citet{li2022random}, which suggest using \(8\hat{\tau}^2\) as a conservative bound for \(b(E[f_1^{(1)}(1, \pi)] - E[f_1^{(1)}(0, \pi)])^2\). However, our calculations and simulation below show that the coverage probability of the confidence interval derived using this conservative bound may fall short of the nominal level, meaning it may not be as conservative as originally suggested.

\subsection{Empirical validation}
We assess the accuracy of our asymptotic variance estimates from Theorems~\ref{th:main} and \ref{th:3} through simulations and compare our results with the conservative confidence interval method of \cite{li2022random}. In our setup, we define the graphon \(h(x,y) = x^2 + y^2 + xy + 0.1\) and set \(\rho_n = n^{-0.25}\). The outcome function is specified as
\[
f_i\left(W_i, \frac{M_i}{N_i}\right) = W_i \left( -2 \left( 1 - \frac{M_i}{N_i} \right)^2 -2 \mathbf{z}_i \left(\frac{M_i}{N_i}\right)^2 + \frac{1}{2} \xi_i \right) + \mathbf{z}_i^2,
\]
where \(\xi_i \overset{i.i.d.}{\sim} N(0,1)\) and \(\mathbf{z}_i \overset{i.i.d.}{\sim} U(-2,1)\). Note that \(h\) is a rank-3 graphon, as \(h(x,y) = 10(x^2 + 0.1)(y^2 + 0.1) - 10x^2y^2 + xy\). 

To evaluate the performance of our methods, we run simulations for all combinations of \( n = 100, 300, 500 \) and \( \pi = 0.5, 0.6, 0.7\). Each simulation is repeated 1000 times, and we compute the frequency at which the confidence intervals capture the true value of \( \tau \). The results, summarized in Table~\ref{tab:sim1}, show that our confidence intervals are generally close to the nominal level. In contrast, the method of \cite{li2022random} does not provide reliable coverage guarantees.
\begin{table}[!htbp]
	\centering
	\small
	\caption{Coverage rates of the nominal 95\% confidence intervals.}
	\label{tab:sim1}
	\begin{tabular}{c|ccccccc}
		\hline
		& \multicolumn{3}{c}{This paper} & & \multicolumn{3}{c}{\cite{li2022random}} \\ 
		\cline{2-4} \cline{6-8} $n \backslash \pi$ & 0.5 & 0.6 & 0.7& & 0.5 & 0.6 & 0.7 \\ 
		\hline 
		100 & 0.901 & 0.935 & 0.964 & &0.837 & 0.853 & 0.880 \\
		300 & 0.914 & 0.941 & 0.960 && 0.848 & 0.856 & 0.892 \\ 
		500 & 0.925 & 0.943 & 0.963 && 0.832 & 0.852 & 0.897 \\ \hline \end{tabular}
\end{table}

\section{Nonparametric Adjustment}
\label{sec:Local constant estimator}
Thus far, we have explored the estimation of \( \tau \) using linear-regression adjustment. Given the formulation in \eqref{eq:ra-dim}, a natural question arises: Is it truly necessary to restrict \( \mu_1 \) and \( \mu_0 \) to linear forms?  
This prompts a deeper inquiry:  
\begin{center}
	\textit{Does an ``optimal" regression adjustment estimator exist?}  
\end{center}  
To investigate this, we expand our scope beyond the linear-regression adjustment of \( \hat{\tau}(\alpha_1, \alpha_0) \) in Section~\ref{sec:Regression-based estimators} and consider a broader class of estimators:
\begin{align}
	\label{def:class of conditional expectation}
	\hat{\tau}(g_1, g_0) = \sum_{i=1}^n \left( 
	\frac{W_i \left( Y_i - g_1(\mathbf{z}_i) + \bar{g}_1 \right)}{\sum_{j=1}^n W_j} 
	- 
	\frac{(1-W_i) \left( Y_i - g_0(\mathbf{z}_i) + \bar{g}_0 \right)}{\sum_{j=1}^n (1-W_j)}
	\right),
\end{align}
where $g_1$ and $g_0$ are arbitrary functions, $\bar{g}_1 = \frac{1}{n} \sum_{i=1}^n g_1(\mathbf{z}_i)$, and $\bar{g}_0 = \frac{1}{n} \sum_{i=1}^n g_0(\mathbf{z}_i)$. Notably, for any bounded functions $g_1$ and $g_0$, $\hat{\tau}(g_1, g_0)$ is asymptotically normal, with a well-defined expression for its asymptotic variance. 
We formalize this result in the following theorem.
\begin{theorem}
	\label{th:6}
	Let Assumptions \ref{assumption:anonymous interference}, \ref{assump:random_graph}, and \ref{assumption:moment_conditions} hold, and $g_1$ and $g_0$ be  bounded functions. Then, the estimator $\hat{\tau}(g_1, g_0)$ satisfies:
	\[
	\sqrt{n} \left( \hat{\tau}(g_1, g_0) - \tau \right) \overset{d}{\to} N \left( 0, \tilde{V}(g_1, g_0) \right),
	\]
	where the asymptotic variance $\tilde{V}(g_1, g_0)$ is given by:
	\begin{align}
		\label{def:V g1 g0}
		\tilde{V}(g_1, g_0) &= \mathrm{Var}(f_i(1, \pi) - f_i(0, \pi)) \notag \\
		&\quad + \frac{1}{\pi(1-\pi)} \mathbb{E} \left[ 
		(1-\pi) \left( f_i(1, \pi) - \mathbb{E}[f_i(1, \pi) \mid \mathbf{z}_i] \right) + 
		\pi \left( f_i(0, \pi) - \mathbb{E}[f_i(0, \pi) \mid \mathbf{z}_i] \right) 
		\right]^2 \notag \\
		&\quad + b \pi(1-\pi) \left( \mathbb{E}[f_1^{(1)}(1, \pi)] - \mathbb{E}[f_1^{(1)}(0, \pi)] \right)^2 \notag \\
		&\quad + \frac{1}{\pi(1-\pi)} \mathbb{E} \left\{
		(1-\pi) \left( g_1(\mathbf{z}_i) - \mathbb{E}[g_1(\mathbf{z}_i)] - \mathbb{E}[f_i(1, \pi) \mid \mathbf{z}_i] + \mathbb{E}[f_i(1, \pi)] \right) 
		\right. \notag \\
		&\quad \left. + \pi \left( g_0(\mathbf{z}_i) - \mathbb{E}[g_0(\mathbf{z}_i)] - \mathbb{E}[f_i(0, \pi) \mid \mathbf{z}_i] + \mathbb{E}[f_i(0, \pi)] \right)
		\right\}^2,
	\end{align}
	with $b = \mathbb{E} \left[ \mathbb{E} \left( \frac{h(U_i, U_j)}{\mathbb{E}[h(U_i, U_j) \mid U_j]} \bigg| U_i \right) \right]^2$.
\end{theorem}
In this decomposition, the first three terms remain unchanged regardless of \( g_1 \) and \( g_0 \) and cannot be further reduced. However, the final term disappears when we choose  
\[
g_1(\mathbf{z}_i) = E(f_i(1,\pi) | \mathbf{z}_i), \quad g_0(\mathbf{z}_i) = E(f_i(0,\pi) | \mathbf{z}_i),
\]  
that is, when \( g_1 \) and \( g_0 \) are set as the conditional expectations of \( Y_i \) given \( \mathbf{z}_i \) in each group. This insight naturally leads to our next discussion on nonparametric regression adjustment.  
\begin{remark}
	When the outcome function is linear in the covariates \( \mathbf{z}_i \), the regression-adjusted estimator \( \hat{\tau} \) achieves optimality within the class defined in (\ref{def:class of conditional expectation}). Specifically, consider \( \tilde{f} \) as defined in (\ref{Vdim Vreg}). The variance of the regression-adjusted estimator satisfies  
	\[
	V_{reg} = V_{dim} - \frac{1}{\pi(1-\pi)} \left[ E(\mathbf{x}_1 \tilde{f})^\top \mathbf{M}_{xx}^{-1} E(\mathbf{x}_1 \tilde{f}) - (E\tilde{f})^2 \right] 
	\geq V_{dim} - \frac{1}{\pi(1-\pi)} \text{Var}(E(\tilde{f}|\mathbf{z}_1)),
	\]
	where the equality follows from (\ref{eqn:pf1}) in the proof of Theorem~\ref{th:main}, and the inequality results from Lemma 2 in \cite{lavergne2008cauchy}. Equality holds when  
	\[
	E\left( (1-\pi)f_1(1,\pi) + \pi f_1(0,\pi) \mid \mathbf{z}_1 \right)
	\]
	is a linear function of \( \mathbf{z}_1 \). Moreover, straightforward calculations show that \( V_{dim} - \frac{1}{\pi(1-\pi)} \text{Var}(E(\tilde{f}|\mathbf{z}_1)) \) is equal to the first three terms on the right-hand side of (\ref{def:V g1 g0}), confirming the result.
\end{remark}
Existing literature (e.g., \cite{beemer2018ensemble, shi2019adapting, naimi2023challenges}) has successfully explored machine learning approaches for adjustments beyond linear regression, demonstrating their effectiveness in empirical studies. Additionally, related theoretical advancements (e.g., \cite{chernozhukov2017double, fan2022estimation}) have addressed these adjustments in settings without network interference. However, to the best of our knowledge, the application of these methods in the context of network interference remains unexplored.

Recall, with some abuse of notation, a classic setup where one has i.i.d. samples $(\mathbf{z}_1, Y_1), \ldots, (\mathbf{z}_n, Y_n)$. A simple estimator of the mean function $m(\mathbf{z}) = E(Y|\mathbf{z})$ is the familiar local constant estimator defined as:
\[
\hat{m}(\mathbf{z}) = 
\frac{
	\sum_{j=1}^n K\left(\frac{\mathbf{z}_j - \mathbf{z}}{\tilde{h}}\right) Y_j
}{
	\sum_{j=1}^n K\left(\frac{\mathbf{z}_j - \mathbf{z}}{\tilde{h}}\right)
},
\]
where $K(\cdot)$ is a kernel function, and $\tilde{h}$ is a bandwidth parameter. Under standard regularity conditions, it holds that $\sup_{\mathbf{z}} |m(\mathbf{z}) - \hat{m}(\mathbf{z})| = o_p(1)$, ensuring that the local constant estimator can consistently estimate $m(\mathbf{z})$. This fact motivates the estimation of  $E(f_i(1, M_i/N_i) | \mathbf{z}_i)$ and $E(f_i(0, M_i/N_i) | \mathbf{z}_i)$,  two functions closely related to the last term in \eqref{def:V g1 g0}. 
Formally, define the local constant estimators for units in the treatment and control groups, respectively, as: 
\[
\hat m_1(\mathbf{z}) = \frac{
	\sum_{j=1}^n K\left(\frac{\mathbf{z}_j - \mathbf{z}}{\tilde h}\right)Y_j W_j
}{
	\sum_{j=1}^n K\left(\frac{\mathbf{z}_j - \mathbf{z}}{\tilde h}\right)W_j
}, \quad
\hat m_0(\mathbf{z}) = \frac{
	\sum_{j=1}^n K\left(\frac{\mathbf{z}_j - \mathbf{z}}{\tilde h}\right)Y_j (1 - W_j)
}{
	\sum_{j=1}^n K\left(\frac{\mathbf{z}_j - \mathbf{z}}{\tilde h}\right)(1 - W_j)
}.
\]
These estimators allow us to approximate \(f_i(1, M_i/N_i)\) and \(f_i(0, M_i/N_i)\), with one of them representing a counterfactual, using \(\hat m_1(\mathbf{z}_i)\) and \(\hat m_0(\mathbf{z}_i)\), respectively. The ATE \(\tau\) is estimated as:  
\begin{align}
	\label{def:hat tau np}
	\hat \tau_{np} = \frac{1}{n} \sum_{i=1}^n \left(
	\hat m_1(\mathbf{z}_i) - \hat m_0(\mathbf{z}_i)
	\right) I\left(\tilde p_1(\mathbf{z}_i) > \tilde b, \tilde p_2(\mathbf{z}_i) > \tilde b, \hat p(\mathbf{z}_i) > 1.01\tilde b \right),
\end{align}
where:  
\[
\hat p(\mathbf{z}) = \frac{1}{n \tilde h^p} \sum_{j=1}^n K\left(\frac{\mathbf{z}_j - \mathbf{z}}{\tilde h}\right), \quad 
\tilde p_1(\mathbf{z}) = \frac{1}{n \tilde h^p \hat \pi} \sum_{j=1}^n K\left(\frac{\mathbf{z} - \mathbf{z}_j}{\tilde h}\right)W_j, \]
\[	\tilde p_2(\mathbf{z}) = \frac{1}{n \tilde h^p (1 - \hat \pi)} \sum_{j=1}^n K\left(\frac{\mathbf{z} - \mathbf{z}_j}{\tilde h}\right)(1 - W_j),\]
represent the density estimates of the covariates, and \(\tilde b \to 0\) is a trimming parameter employed to ensure that the denominators in the density estimates are not too small, thereby avoiding numerical instability.  Similar trimming techniques can be found in the literature, such as \cite{hardle1989investigating} and \cite{banerjee2007method}. The constant \(1.01\) used in the indicator function is purely a technical choice and can be replaced by any constant greater than \(1\). 

The nonparametric estimator \( \hat{\tau}_{np} \) is essentially a simple average. Remarkably, even though the convergence of \( \hat{m}_1(\mathbf{z}) \) and \( \hat{m}_0(\mathbf{z}) \) slows as the dimensionality \( p \) increases, \( \hat{\tau}_{np} \) still achieves the standard parametric convergence rate of \( \sqrt{n} \). This demonstrates that our nonparametric estimator effectively \textit{overcomes the curse of dimensionality}.

While achieving the \( n^{-1/2} \) rate for estimating $\tau$ is expected in standard kernel smoothing with i.i.d. data, the presence of complicated dependence due to network interference in our setup makes theoretical development significantly more challenging. Specifically, beyond the standard terms such as \( \delta_{11} \) and \( \delta_{21} \) in Lemma~\ref{lem4th6:5} and Lemma~\ref{lem4th6:8}, which correspond to the leading order terms in Taylor expansions, additional non-negligible terms, including \( \delta_{12} \) in Lemma~\ref{lem4th6:6} and \( \delta_{22} \) in Lemma~\ref{lem4th6:9}, emerge as counterparts to the first-order terms in Taylor expansions. These additional terms arise due to the inherent randomness introduced by the network structure. To obtain a precise approximation of such terms, we introduce a novel method inspired by the projection framework of U-statistics. We exclude the randomness induced by the network structure during the projection step, and subsequently leverage the graphon formulation of the network to control residual terms (such as $\tilde r_{n1},\tilde r_{n2},\tilde r_{n3}$ in the proof of Lemma~\ref{lem4th6:6}). This result is not only a key contribution to our work but can also be of independent interest.

The above claim can be formalized when the dimensionality of the covariates is fixed.	In this context, we outline the regularity assumptions, which are standard in nonparametric methods.
\begin{assumption}
	\label{ass:main3}
	\begin{enumerate}
		\item The kernel function \(K(u)\) satisfies \(K(u) = 0\) for \(u \in \{u : |u|_\infty \geq 1\}\), \(\sup_{u \in \mathbb{R}^p} |K(u)| < \infty\), \(\int_{\mathbb{R}^p} K(u) \, du = 1\) and is symmetric \(q\)-th order. 
		That is, 
		\begin{align*}
			\int_{\mathbb{R}^p} u_1^{l_1} u_2^{l_2} \cdots u_p^{l_p} K(u) \, du &= 0, \quad \forall \, l_1 + l_2 + \cdots + l_p < q, \\
			\int_{\mathbb{R}^p} u_1^{l_1} u_2^{l_2} \cdots u_p^{l_p} K(u) \, du &\neq 0, \quad \exists \, l_1 + l_2 + \cdots + l_p = q.
		\end{align*}
		Furthermore, for some constant \(L < \infty\), \(|K(u) - K(u')| < L |u - u'|_\infty\) for all \(u, u' \in \mathbb{R}^p\).
		
		\item The covariate \(\mathbf{z}_i\) has a marginal density \(p(z)\) such that \(\sup_z p(z) < \infty\), \(\sup_z |z|_\infty^p p(z) < \infty\), and \(E|z|_\infty^{2p} < \infty\). Additionally, within the support of \(\mathbf{z}\), the \(q\)-th derivative of \(p(\mathbf{z})\) is uniformly continuous, and \(\nabla p(\mathbf{z})\) is bounded.
		
		\item Both \(E(f_1(1, \pi) | \mathbf{z}_1 = u)\) and \(E(f_1(0, \pi) | \mathbf{z}_1 = u)\) are continuous with respect to \(u\).
		
		\item The bandwidth \(\tilde h\) satisfies \(\tilde h = o(\tilde b)\), \(\tilde b = o(1)\), \(n \tilde h^{2p} \tilde b^4 / \log^2 n \to \infty\), and \(n \tilde h^{4q} / \tilde b^4 \to 0\).
		\item It holds
		\begin{align}
			\max_{w = 0, 1, \, k = 0, 1, 2} \sup_{\mathbf{z}_1} E\big(\sup_{y \in [0, 1]} |f_1^{(k)}(w, y)|^2 \big| w, \mathbf{z}_1\big) < \infty.
		\end{align}
		\item  
		$  E|S_1(\mathbf{z})| I(p(\mathbf{z}) < 1.02\tilde{b}) = o(n^{-1/2}),$  
		where \(S_1(u) = E(f_1(1, \pi) - f_1(0, \pi) | \mathbf{z}_1 = u)\).
	\end{enumerate}
\end{assumption}
For the kernel \(K(\cdot)\), we adopt a product kernel defined as \(K(u_1, \cdots, u_p) = \prod_{j=1}^p \tilde K(u_j)\), where \(\tilde K(\cdot)\) is a one-dimensional kernel function. In our numerical study, we utilize kernels of different orders (i.e., \(q = 2, 4, 6\)) to address varying dimensions, as outlined in (\ref{eqn:choice of kernel}). Condition 4 defines a feasible region, enabling a rule-of-thumb parameter selection as described in Remark~\ref{rem:Parameter choice}. 
Finally, note that the moment condition is less restrictive than boundedness. For instance, consider \(f_i(w, u) = f_0(w, u) + \epsilon_i\), where \(f_0(w, u)\) is a bounded function and \(\epsilon_i\) is an i.i.d. unbounded error. In this case, \(f_i(w, u)\) is unbounded, yet our condition remains satisfied.
\begin{remark}
	In Condition 6 of Assumption \ref{ass:main3}, we impose certain constraints on the relationship between \(S_1(\mathbf{z})\), \(p(\mathbf{z})\), and \(\tilde{b}\). Similar conditions have been encountered in the literature, such as Assumption 8 in \cite{hardle1989investigating}. These conditions are generally mild. For instance, if \(p(\mathbf{z})\) is bounded away from zero over the support of \(\mathbf{z}\), the condition is automatically satisfied as \(\tilde{b} \to 0\). Furthermore, if \(z\) follows a univariate normal distribution and \(S_1(z) = \exp(-cz^2)\) for some \(c > 0\), a sufficient condition for Assumption~\ref{ass:main3}(6) is \(\tilde b = O(n^{-1/(2+2c)})\). Lastly, the constant 1.02 is primarily technical. Specifically, the constants 1.01 in (\ref{def:hat tau np}) and 1.02 in Assumption~\ref{ass:main3}(6) can be replaced by any real numbers \(c_1\) and \(c_2\), provided that \(1 < c_1 < c_2\).
\end{remark}

\begin{remark}[Rule-of-thumb parameter choice]
	\label{rem:Parameter choice}
	Based on Assumption~\ref{ass:main3}, we first choose \(q\) such that \(2q - 6 < p < 2q\), and then select the bandwidths \(\tilde{h} = C_1 n^{-1/a_1}\) and \(\tilde{b} = C_2 n^{-1/a_2}\), where \(a_1, a_2 > 0\) satisfy the conditions \(p + 2q < a_1 < 4q\) and \(a_2 > \frac{4a_1}{4q - a_1}\), with constants \(C_1\) and \(C_2\).
	
	In the simulation studies presented in Section~\ref{sec:Numerical Evaluation2}, we set \(a_1 = 0.5p + 3q\) and \(a_2 = \frac{3p + 18q}{q - 0.5p}\). For the constants \(C_1\) and \(C_2\), we choose \(C_1 = 1 + 0.5p\), and \(C_2\) as the \(\alpha\)-quantile of \( \hat{p}(\mathbf{z}_i)\), where \(\alpha\) is set to 0.01 or 0.05 depending on the scenario.
	
	It is worth noting that \(\hat{\tau}_{np}\) is robust to  
	\(\tilde{h}\), since \(\hat{\tau}_{np} = \hat{\tau}_{dim}\) when \(\tilde{h} = \infty\) and $\tilde h=0$. Additionally, \(\hat{\tau}_{np}\) is robust to the choice of \(\alpha\).
\end{remark}

The main result of this section is stated in Theorem~\ref{th:5}.
\begin{theorem}[Asymptotic normality of \(\hat{\tau}_{np}\)]
	\label{th:5}
	Under Assumptions \ref{assumption:anonymous interference}, \ref{assump:random_graph}, \ref{assumption:moment_conditions}, and Assumption \ref{ass:main3}, we have
	\[
	\sqrt{n}(\hat{\tau}_{np} - \tau) \overset{d}{\to} N(0, V_{np}),
	\]
	where
	\begin{align*}
		V_{np} &= \text{Var}(f_i(1, \pi) - f_i(0, \pi)) \\
		&\quad + \frac{1}{\pi(1 - \pi)} \mathbb{E}\left[ \left( (1 - \pi)\left(f_i(1, \pi) - \mathbb{E}(f_i(1, \pi) | \mathbf{z}_i)\right) + \pi \left(f_i(0, \pi) - \mathbb{E}(f_i(0, \pi) | \mathbf{z}_i)\right) \right)^2 \right] \\
		&\quad + b\pi(1 - \pi)\left( \mathbb{E}(f_1^{(1)}(1, \pi)) - \mathbb{E}(f_1^{(1)}(0, \pi)) \right)^2,
	\end{align*}
	and \(b = \mathbb{E}\left[\mathbb{E}\left(\frac{h(U_i, U_j)}{\mathbb{E}(h(U_i, U_j) | U_j)} \mid U_i \right)\right]^2\).
\end{theorem}

Theorem~\ref{th:5} establishes that \(\hat{\tau}_{np}\) is asymptotically normal, with a rate of convergence of \(\sqrt{n}\). More importantly, its asymptotic variance consists of three terms, all of which are identical to those in \eqref{def:V g1 g0} and are independent of \(g_1\) and \(g_0\). This result demonstrates that \(\hat{\tau}_{np}\) is asymptotically optimal in terms of variance. In particular, \(\hat{\tau}_{np}\) outperforms \(\hat{\tau}\) by having a smaller asymptotic variance.

To consistently estimate \( V_{np} \) in practice, we propose an iterative approach based on polynomial regression. As demonstrated in Theorem~\ref{th:6}, the asymptotic variances of polynomial regressions are always at least as large as \( V_{np} \). This allows us to progressively increase the degree of the polynomial regression to estimate \( V_{np} \) more accurately. We begin with a constant regressor \( \mathbf{x}_i = 1 \), which corresponds to the difference-in-means estimator, and then iteratively add higher-order terms from the basis function set \( \{\phi_i\}_{i=1}^\infty \) until the estimated variance stabilizes. A common choice for the basis functions is \( \phi_i(z) = z^i \), which performs well in our numerical experiments. Alternatively, Legendre polynomials can be used for enhanced computational stability.

\subsection{Numerical validation}
\label{sec:Numerical Evaluation2}
We conduct numerical simulations to assess the performance of \( \hat{\tau}_{np} \) as predicted by Theorem~\ref{th:5}. In particular, we explore the sensitivity of the results to the tuning parameters \( \tilde{h} \) and \( \alpha \), as discussed in Remark~\ref{rem:Parameter choice}. Additionally, we evaluate how \( \hat{\tau}_{np} \) performs with varying dimensionalities of \( \mathbf{z} \), and investigate the accuracy of the estimation of \( V_{np} \).

To evaluate the results presented in Theorem~\ref{th:5}, we consider the following scenario. Let
\[
h(x, y) = x^2 + y^2 + xy + 0.1, \quad \rho_n = n^{-0.25}, \quad \pi = 0.7,
\]
and define \( f_i(W_i, \theta_i) \) as
\[
f_i\left(W_i, \theta_i\right) = W_i\left(\theta_i - \frac{1}{2} + \frac{\sum_{j=1}^p z_{i,j}}{\sqrt{3p - 4 + 2^{2 - p}}} + \frac{1}{2} \xi_i \right) + \frac{1}{2\sqrt{p}} \sum_{j=1}^p \exp(z_{i,j}),
\]
where
\[
\theta_i = \begin{cases}
	\pi, & \text{without network interference}, \\
	\frac{M_i}{N_i}, & \text{with network interference}.
\end{cases}
\]
Furthermore, \( \xi_i \sim \mathcal{N}(0, 1) \) i.i.d., and \( \mathbf{z}_i = (z_{i,1}, \dots, z_{i,p}) \sim \mathcal{N}(0, \Sigma) \) i.i.d., with \( \Sigma_{i,j} = 0.5^{|i-j|} \) for \( 1 \leq i,j \leq p \). For this scenario, straightforward calculations show that for any dimension \( p \), we have \( \tau = 0.2 \) and
\[
V_{np} \approx \begin{cases}
	1.357, & \text{without network interference}\\
	1.616, & \text{with network interference}
\end{cases}.
\]
In this section, we employ Epanechnikov kernels of different orders depending on \( p \). Specifically, we use the following kernel functions:
\begin{align}
	\label{eqn:choice of kernel}
	K(z_1, \dots, z_p) = \begin{cases}
		\prod_{j=1}^p \frac{3}{4}(1 - z_j^2)I(|z_j| \leq 1), & 1 \leq p \leq 3 \\
		\prod_{j=1}^p \frac{45}{32}(1 - z_j^2)\left( 1 - \frac{7}{3} z_j^2 \right)I(|z_j| \leq 1),& 4 \leq p \leq 7 \\
		\prod_{j=1}^p \frac{525}{256}(1 - z_j^2)\left( 1 - 6 z_j^2 + \frac{33}{5} z_j^4 \right)I(|z_j| \leq 1), & 8 \leq p \leq 10
	\end{cases}.
\end{align}
For each configuration described below, we repeat the simulation 1000 times.
\subsubsection{Finite sample performance}
We set \( n = 1000 \) and \( p = 1\) or \(5 \). For the choice of \( \tilde{h} \), we follow Remark~\ref{rem:Parameter choice} and select \( \tilde{h} = 0.518\) for \( p = 1 \) and \( \tilde{h} = 2.173 \) for \( p = 5 \). Additionally, we set \( \alpha = 0.01 \) which is discussed in Remark \ref{rem:Parameter choice}. 
A sensitivity analysis is conducted in the next section, demonstrating the robustness of the results to changes in the parameter values. 

Figure~\ref{fig:sim2} shows the distributions of \( \hat{\tau}_{np} \) for \( p = 1 \) and \( p = 5 \). We observe that the sample distributions closely match the theoretical distribution derived in Theorem~\ref{th:5} (shown in red). {Notably, our \(\hat\tau_{np}\) exhibits good performance regardless of whether network interference is present or absent.} Moreover, the sample distribution for \( p = 5 \) fits well with the limiting distribution, indicating that our estimator, \( \hat{\tau}_{np} \), can mitigate the ``curse of dimensionality" to some extent when estimating the average treatment effect \( \tau \). Further numerical validations are provided in Section~\ref{sec:Performance under multi-dimensional covariates}.

\begin{figure}[ht]
	\centering
	\includegraphics[width=14cm]{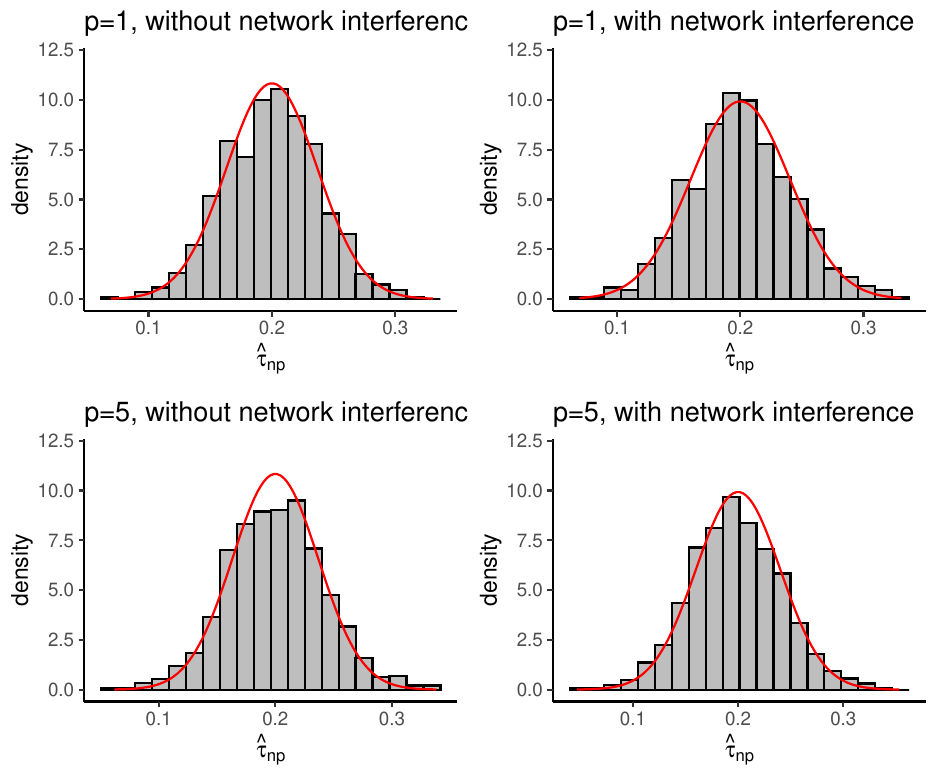}
	\caption{The histograms of $\hat\tau_{np}$ for $p=1$ and $p=5$, with the red curves representing the limiting Gaussian distributions according to Theorem \ref{th:5}.}
	\label{fig:sim2}
\end{figure}

To demonstrate the performance gain of  $\hat\tau_{np}$ over $\hat\tau$, we plot in Figure \ref{fig:sim3} their respective MSEs with network interference, when $p$ is set as $5$ and $\alpha$ as $0.05$. We see that clearly nonparametric adjustment improves substantially over linear regression-adjusted estimator for all the sample sizes we have examined.

\begin{figure}[tbp]
	\centering 
	\includegraphics[width=0.5\textwidth]{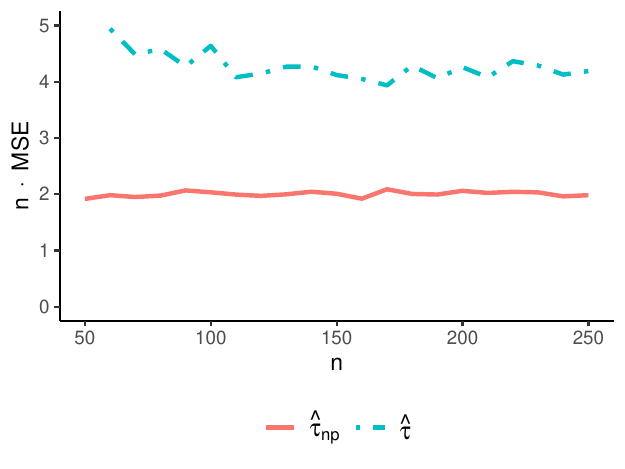}
		\caption{ 
			A comparison of MSEs between $\hat\tau_{np}$ (the estimator after nonparametric regression adjustment) and $\hat\tau$ (the estimator after linear regression adjustment).
		}
		\label{fig:sim3}
	\end{figure}

	\subsection{Robustness to parameter choices}
	
	In this section, we demonstrate that the results are robust to the choice of parameters. Specifically, we compute the sample MSE of \( \hat{\tau}_{np} \) for various combinations of \( (\tilde{h}, \alpha) \). For \( p = 1 \), we consider \( \tilde{h} = 0.2, 0.4, \dots, 1.0 \) and \( \alpha = 0.01, 0.05 \). For \( p = 5 \), we use \( \tilde{h} = 1.8, 2.0, \dots, 2.6 \) and \( \alpha = 0.01, 0.05 \). We set $n=1000$. Since this paper primarily focuses on scenarios with network interference, we only present simulations for this case.
	
	Table~\ref{tab:sim3} summarizes the values of the sample MSE scaled by \( n \). The results indicate that the sample MSE is not significantly influenced by the choice of \( \tilde{h} \) and \( \alpha \). More importantly, these MSEs are strictly better than the MSE of \( \hat{\tau}  \) which is $3.95$ approximately. 
	\begin{table}[ht]
		\centering
		\small
		\caption{Values of the sample MSE multiplied by $n$ for different choices of $(\tilde h,\alpha)$.}
		\label{tab:sim3}
		\begin{tabular}{cc|ccccc|cccc}
			\hline
			\multicolumn{2}{c|}{\multirow{2}{*}{$p=1$}} & \multicolumn{2}{c}{$\alpha$} &  & \multicolumn{2}{c|}{\multirow{2}{*}{$p=5$}} & \multicolumn{2}{c}{$\alpha$} \\ \cline{3-4} \cline{8-9} 
			\multicolumn{2}{c|}{} & 0.01 & 0.05 &  & \multicolumn{2}{c|}{} & 0.01 & 0.05 \\ \cline{1-4} \cline{6-9} 
			\multirow{5}{*}{$\tilde h$} & 0.2 & 1.906 & 2.165 &  & \multirow{5}{*}{$\tilde h$} & 1.8 & 1.949 & 2.207 \\
			& 0.4 & 1.896 & 1.845 &  & & 2.0 & 1.858 & 1.973 \\
			& 0.6 & 1.740 & 1.826 &  & & 2.2 & 2.057 & 1.907 \\
			& 0.8 & 1.755 & 1.974 &  & & 2.4 & 2.015 & 2.204 \\
			& 1.0 & 2.095 & 1.875 &  & & 2.6 & 2.035 & 2.009 \\ \hline
		\end{tabular}
		
	\end{table}
	
	\subsubsection{Performance under multi-dimensional covariates}
	\label{sec:Performance under multi-dimensional covariates}
	
	In this section, we evaluate the performance of \( \hat{\tau}_{np} \) for scenarios with network interference as the dimension of the covariates, \( p \), varies over a broad range. Specifically, we set \( p = 1, 2, \dots, 10 \), fix \( n = 1000 \) and \( \alpha = 0.01 \). The results are summarized in Table~\ref{tab:sim4}. 
	
	As expected, the sample MSE increases with \( p \). However, the convergence rates consistently remain at \( \sqrt{n} \), in contrast to the slower rates typically associated with the ``curse of dimensionality.'' This suggests that the local constant estimator \( \hat{\tau}_{np} \) is robust to increases in dimensionality and demonstrates superior performance compared to the regression estimator \( \hat{\tau} \), even for relatively large \( p \).
	
	\begin{table}[htbp]
		\centering
		\small
		\caption{Performance as the dimension of covariates increases.}
		\label{tab:sim4}
		\begin{tabular}{cccccc c ccccc}
			\hline
			\( p \) & \( \tilde{h} \) & Mean & Variance & MSE & & \( p \) & \( \tilde{h} \) & Mean & Variance & MSE \\ \hline
			1   & 0.518      & 0.200 & \( 1.735/n \) & \( 1.735/n \) & & 6   & 2.523      & 0.199 & \( 2.103/n \) & \( 2.104/n \) \\
			2   & 0.745      & 0.200 & \( 1.794/n \) & \( 1.794/n \) & & 7   & 2.881      & 0.204 & \( 2.248/n \) & \( 2.267/n \) \\
			3   & 0.995      & 0.200 & \( 1.957/n \) & \( 1.958/n \) & & 8   & 3.652      & 0.191 & \( 1.902/n \) & \( 1.982/n \) \\
			4   & 1.831      & 0.199 & \( 1.922/n \) & \( 1.923/n \) & & 9   & 4.046      & 0.192 & \( 2.092/n \) & \( 2.144/n \) \\
			5   & 2.173      & 0.198 & \( 2.016/n \) & \( 2.019/n \) & & 10  & 4.443      & 0.194 & \( 2.120/n \) & \( 2.147/n \) \\ \hline
		\end{tabular}
	\end{table}

	\subsubsection{Variance estimation performance}
	
	To assess the variance estimation  performance, we consider sample sizes \(n = 100, 300, 500\) and dimensions \(p = 1, 5\), evaluating whether the confidence intervals adequately cover the true value of \(\tau\). For \(p = 1\) and a nominal coverage of 95\%, the coverage rates are 96.5\%, 97.1\%, and 97.1\% for \(n = 100\), 300, and 500, respectively. When \(p = 5\), the coverage rates are 92.1\%, 97.3\%, and 98.3\% for the same sample sizes. {The results align with the findings in Table~\ref{tab:sim1}. Additionally, the coverage probabilities are slightly conservative, primarily due to the cautious nature of our variance estimation method outlined in Section~\ref{sec:Local constant estimator}.}

	\section{A naturalistic simulation with contact data}
	\label{sec:A naturalistic simulation with contact data}
	In this section, we use real networks to illustrate how our results can be applied and to demonstrate how the structure of the network impacts statistical inference. The data were collected by the SocioPatterns project on {\url{http://www.sociopatterns.org}} using active RFID devices (\cite{gemmetto2014mitigation, stehle2011high}). Specifically, on October 1st, 2009, from 8:40 to 17:18, contact data were recorded for 236 individuals, including 10 classes of students and one group of teachers. 
	
	Previous analyses of this dataset revealed variations in the contact patterns of students, with prominent changes occurring around 12:00 and 14:00 (see, for example, Figure 6 in \cite{stehle2011high}). Based on these findings, we divide the data into two time periods: morning (before 12:00) and midday (12:00--14:00). For each period, we construct a contact network. Specifically, an undirected edge \((i, j)\) is drawn in the graph \(G = (V, E)\) if individuals \(i\) and \(j\) were recorded by the RFID device to have contacted each other at least three times (i.e., totaling at least one minute) during the respective time period. 
	
	We visualize the two resulting graphs in Figures~\ref{fig:natural1}, where the nodes are arranged to highlight interactions within and between classes. From these visualizations, we observe distinct patterns: in the morning, interactions are predominantly within classes, whereas at midday, interactions between classes become more prominent. These two networks represent typical real-world scenarios--social interactions in structured spaces (e.g., classrooms) and open spaces (e.g., playgrounds). Therefore, conducting experiments on these networks provides valuable insights.
	\begin{figure}[htbp]
		\centering 
		\subfigure{
			\includegraphics[width=8cm]{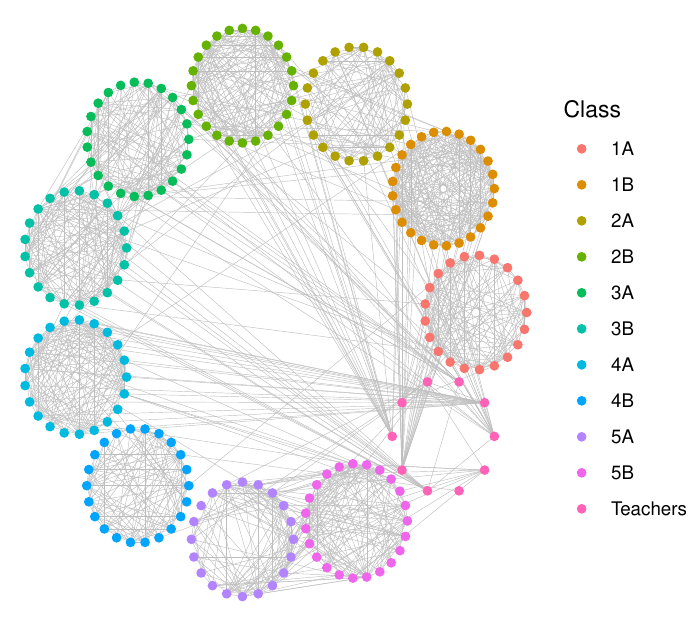}
		}\subfigure{
			\includegraphics[width=8cm]{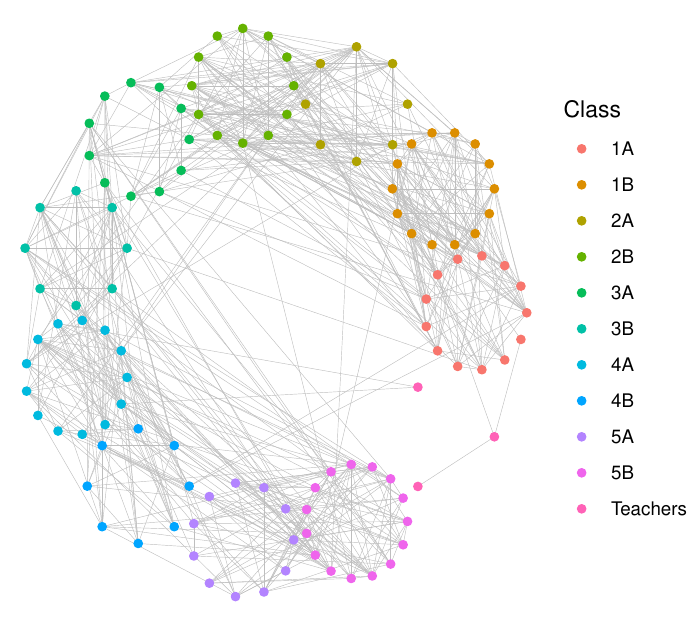}
		} 
		\caption{Aggregated contact networks for students during the morning (left) and midday (right) periods. The morning network reflects interactions primarily within structured environments, such as classrooms, while the midday network captures interactions occurring in more open and informal spaces, such as playgrounds.}
		\label{fig:natural1}
	\end{figure}
	In this section, we simulate experiments using these networks to estimate the effectiveness of epidemic vaccines, considering the impact of network effects. Individual infection rates are likely confounded by these network effects, as they depend not only on an individual's vaccination status but also on whether the individuals they interact with have been vaccinated. Many studies have modeled epidemics on networks \citep{donnat2024one, pare2020modeling, pare2018analysis}.
	
	We extract three key factors for the simulation: (i) the individual's vaccination status (i.e., treatment assignment), (ii) the proportion of vaccinated contacts (i.e., network interference), and (iii) personal vulnerability (i.e., covariate). Specifically, we simulate experiments where the treatment assignment \(W_i\) follows a Bernoulli distribution with mean \(\pi = 0.2\), independently and identically for each individual. The outcome function for individual \(i\) is defined as
	\[
	Y_i = f_i(W_i, M_i/N_i) = \frac{2}{1 + \exp(-z_i^*)} \left(1 - \frac{2}{5} W_i \right) \left(1 - \sqrt\frac{M_i}{N_i}\right)
	\]
	
	where \(z_i^* \overset{i.i.d.}{\sim} N(0, 2)\), \(N_i = \sum_{j=1}^n E_{ij}\), and \(M_i = \sum_{j=1}^n E_{ij} W_j\). 
	
	A few comments are in order:
	\begin{itemize}
		\item The outcome \(Y_i\) can represent a specific medical indicator. In real-world applications, an individual is typically considered infected if their value exceeds a given threshold.
		\item Since \(Y_i\) should be negatively dependent on the proportion of vaccinated neighbors, we use a simple relationship to reflect this qualitative dependence.
		\item Personal vulnerability is influenced by factors such as age \citep{marmor2023assessing}, but for simplicity, we model it as a single random variable, \(z_i^*\). In practice, personal vulnerability \(z_i^*\) may need to be estimated via specific procedures. We assume the available covariate \(z_i\) is a perturbed version of \(z_i^*\). Specifically, we set \(z_i = z_i^* v_i\), where \(v_i \overset{i.i.d.}{\sim} \text{Uniform}(0.9, 1.1)\).
	\end{itemize}
	
	We apply three estimators to estimate the average treatment effect: the difference-in-means estimator, the regression estimator \(\hat\tau\), and the nonparametric estimator \(\hat\tau_{np}\). In this experiment, the true value of the treatment effect is \(\tau \approx -0.221\). For the parameter choices, guided by Remark~\ref{rem:Parameter choice}, we set \(\tilde h = 0.43\) and \(\tilde b = 0.81\) for the morning network, and \(\tilde h = 0.48\) and \(\tilde b = 0.83\) for the midday network. Additionally, we choose \(r = 10\) to account for the ten classes and set \(\alpha = 0.01\).
	
	The results, along with the corresponding confidence intervals, are presented in Table~\ref{tab:natural1}. From the table, we observe that all methods yield similar point estimates, and none of the confidence intervals includes zero. This indicates a significant reduction in infection risk due to vaccination in this experiment. Among the methods, the confidence interval for the difference-in-means estimator is the widest, while those for the regression estimator and the local constant estimator are comparatively narrower.
	
	\begin{table}[htbp]
		\centering
		\small
		\caption{Estimation of average treatment effect and 95\% confidence intervals.}
		\label{tab:natural1}
		\begin{tabular}{ccccc}
			\hline
			Network                  & Method                 & Estimate & Standard error & 95\% confidence interval \\ \hline
			\multirow{3}{*}{Morning} & 
			Difference-in-means & $-0.217$   & 0.0498         & $[-0.315,-0.119]$      \\
			& Regression   & $-0.243$   & 0.0338          & $[-0.310,-0.177]$      \\
			& Nonparametric        & $-0.255$   & 0.0334         & $[-0.320,-0.189]$      \\ \hline
			\multirow{3}{*}{Midday}  &
			Difference-in-means & $-0.234$   & 0.0721        & $[-0.375,-0.092]$      \\
			& Regression             & $-0.256$  & 0.0536         & $[-0.361,-0.151]$      \\
			& Nonparametric         & $-0.243$   & 0.0528        & $[-0.347,-0.140]$      \\ \hline
		\end{tabular}
		
	\end{table}

	In the subsequent analysis, based on 1000 repetitions, we demonstrate that the regression estimator leads to a variance reduction of approximately  28.3\% for the morning network and 22.0\% for the midday network compared to the difference-in-means estimator.
	
	To assess the coverage rate of the confidence intervals, we repeat the experiment 1000 times and compute the average of the point estimates, standard errors, and the coverage rates of the confidence intervals. The results are summarized in Table \ref{tab:natural2}. As shown, the coverage rates of the confidence intervals are close to $0.95$, and the mean point estimates are near the true value of $-0.221$. This demonstrates that, under real-world network structures, our methods can be reliably used for statistical inference on the average treatment effects under network interference. Furthermore, it is seen that the two regression estimators leads to variance reduction, sometimes substantially.
	
	Additionally, we observe that the confidence intervals for $\tau$ differ in length between the morning and midday networks due to their distinct structural characteristics. We also present the coverage rates in Table~\ref{tab:natural2} when the network structure is ignored, {i.e., by setting the term $b\pi(1-\pi)(E f_1^{(1)} (1,\pi)-E f_1^{(1)} (0,\pi))^2$ to $0$.}  The results clearly indicate that neglecting the specific structure of the network when analyzing real-world data can lead to substantial errors in statistical inference.
	
	\begin{table}[htbp]
		\centering
		\small
		\caption{Average point estimates, standard errors (SE), and coverage rates of confidence intervals (CI). CI-II refers to the confidence interval calculated without accounting for network interference.}
		\label{tab:natural2}
		\begin{tabular}{cccccc}
			\hline
			Network                  & Method        & Mean & SE & CI & CI-II \\ \hline
			\multirow{3}{*}{Morning} & Difference-in-means  & $-0.231$ & 0.0584 & 0.944 & 0.912 \\
			& Linear regression & $-0.232$ & 0.0416 & 0.923 & 0.853 \\
			& Nonparametric & $-0.225$ & 0.0403 & 0.945 & 0.862 \\ \hline
			\multirow{3}{*}{Midday}  & Difference-in-means  & $-0.229$ & 0.0941 & 0.962 & 0.917\\
			& Regression & $-0.228$ & 0.0734 & 0.954 & 0.839 \\
			& Nonparametric & $-0.216$ & 0.0718 & 0.971  & 0.839 \\ \hline
		\end{tabular}
	\end{table}

	\section{Conclusion}
	\label{sec:Conclusion}
	
	In this paper, we propose two estimators---(linear) regression-based and nonparametric adjustment methods---to estimate the average treatment effect in the context of graphon-type network interference. Both estimators exhibit optimality properties, ensuring robust and efficient treatment effect estimation. Additionally, we introduce a variance estimation method that accounts for both the randomness of treatment assignments and the network structure, providing a more precise measure of uncertainty.
	
	Several promising avenues for future research remain. First, investigating other forms of network dependence, extending beyond the fraction of treated neighbors, would provide valuable insights. Second, exploring the properties of the proposed estimators, \(\hat{\tau}\) and \(\hat{\tau}_{np}\), in the context of random networks beyond the graphon model could deepen our understanding. Third, examining the behavior of these estimators in dynamic network settings, where network structures evolve over time, could offer new challenges and opportunities for method development. Lastly, it would be valuable to explore the performance of the estimators under different types of network interference, such as temporal or spatial correlations between nodes. 
	
	\if1\blind
	\section*{Acknowledgments}
	We gratefully acknowledge the suggestions from Hanzhong Liu and Yuhao Wang of  Tsinghua University, which helped improve this document. Weichi Wu is supported by NSFC No.12271287. 
	\fi
	
	\bibliographystyle{apalike}
	\bibliography{Bibliography-MM-MC}

\newpage
\begin{center}
{\large\bf SUPPLEMENTARY MATERIAL}
\end{center}
\section*{Proofs}
\label{sec:Proofs}

\begin{proof}[Proof of Theorem~\ref{th: Consistency of hat tau}]
	Recall that $\textbf{x}_i=(1,\mathbf{z}_i^\top)^\top$ are i.i.d. random vectors, and $W_i$'s are i.i.d. Bernoulli random variables that are independent of $\mathbf{x}_i$'s. We have 
 \begin{align}
 \label{pf1:eqn1}
     \frac{1}{n}\sum_{i=1}^n W_i\mathbf{x}_i\mathbf{x}_i^\top \overset{a.s.}{\to}\pi \bb{M}_{xx}.
 \end{align} Let $g(W_i,\mathbf{z}_i)=E(f_i(W_i,\pi)|\mathbf{z}_i,W_i)$ and $g(\bb{W},\mathbf{z})=(g(W_1,\mathbf{z}_1),\cdots,g(W_n,\mathbf{z}_n))^\top$. By Assumption~\ref{assumption:moment_conditions}, we have $E(|\mathbf{x}_1g(1,\mathbf{z}_1)|)<\infty$. Then we have 
 \begin{align}
     \label{pf1:eqn2}
     \frac{1}{n}\sum_{i=1}^n W_i\mathbf{x}_ig(W_i,\mathbf{z}_i)\overset{a.s.}{\to} \pi E(\mathbf{x}_1g(1,\mathbf{z}_1)).
 \end{align}  Consider the term $\frac{1}{n}\bb{X}^\top \bb{W}(\mathbf{Y}-g(\bb{W},\mathbf{z}))$, its $j${th} $(1\le j\le p+1)$ dimension is
	\begin{align}
		\label{eqn:1}
		I_j&=\frac{1}{n}\sum_{i=1}^nW_ix_{ij}\left(f_i\left(1,\frac{M_i}{N_i}\right)-E(f_i(1,\pi)|\mathbf{z}_i)\right)\notag\\
		&=\frac{1}{n}\sum_{i=1}^nW_ix_{ij}\left(f_i\left(1,\pi\right)-E(f_i(1,\pi)|\mathbf{z}_i)\right)+o_p(1)
		=o_p(1)
	\end{align}
where the second equation follows from the Taylor expansion, the independence of \( W_i \) from \((f_i, \mathbf{z}_i)\), the moment condition in Assumption~\ref{assumption:moment_conditions}, and the bound \(\min_i N_i \gtrsim_p n\rho_n\) which is established in Lemma 15 of \cite{li2022random}. The third equation then follows from the law of large numbers. Therefore, \begin{align}
	    \label{pf1:eqn3}
     \frac{1}{n}\bb{X}^\top \bb{W}\mathbf{Y}=
      \frac{1}{n}\bb{X}^\top \bb{W}g(\bb{W},\mathbf{z})+o_p(1).
	\end{align}
 By (\ref{pf1:eqn1}), (\ref{pf1:eqn2}), and (\ref{pf1:eqn3}), we have
	\begin{equation}
		\hat \beta_1=(\bb{X}^\top \bb{W}\bb{X})^{-1}\bb{X}^\top \bb{W}\mathbf{Y}
		=\bb{M}_{xx}^{-1}E(\mathbf{x}_1g(1,\mathbf{z}_1))+o_p(1)=\beta_1+o_p(1).
	\end{equation}
	Similarly, we have 
	\begin{equation}
		\hat \beta_0=\beta_0+o_p(1).
	\end{equation}
	Since 
 \begin{align}
     \label{pf2:aeqn:6}
     \frac{1}{n}\mathbf{1}^\top\bb{X}\overset{a.s.}{\to}  {E}\mathbf{x}_1^\top,
 \end{align} we have \begin{equation}
		\hat\tau=(E\mathbf{x}_1)^\top \bb{M}_{xx}^{-1}E\left(\mathbf{x}_1(f_1(1,\pi)-f_1(0,\pi))\right)+o_p(1).
	\end{equation}
	Recall that $\mathbf{x}_1=(1,\mathbf{z}_1^\top)^\top$, $\bb{M}_{xx}=\left(\begin{matrix}
		1 & E\mathbf{z}_1^\top\\
		E\mathbf{z}_1 &  E(\mathbf{z}_1 \mathbf{z}_1^\top)
	\end{matrix}\right)$, we then have
	\begin{equation}
		\label{eqn:3}
		\bb{M}_{xx}^{-1}=\left(\begin{matrix}
			1+E\mathbf{z}_1^\top\left(E(\mathbf{z}_1\mathbf{z}_1^\top)\right)^{-1}E\mathbf{z}_1 & -E\mathbf{z}_1^\top \left(E(\mathbf{z}_1\mathbf{z}_1^\top)\right)^{-1}\\
			-\left(E(\mathbf{z}_1\mathbf{z}_1^\top)\right)^{-1}E\mathbf{z}_1 &  \left(E(\mathbf{z}_1\mathbf{z}_1^\top)\right)^{-1}
		\end{matrix}\right).
	\end{equation}
	As a result, we have 
 \begin{align}
     \label{pf2:aeqn:7}
     E\mathbf{x}_1^\top \bb{M}_{xx}^{-1}=(1,0,\cdots,0).
 \end{align}
 Under our Assumption~\ref{assumption:anonymous interference}, \ref{assump:random_graph} and \ref{assumption:moment_conditions}, we have that $\tau=E(f_1(1,\pi)-f_1(0,\pi))$, therefore, $\hat \tau\overset{p}{\to}\tau$. 
\end{proof}

\begin{proof}[Proof of Theorem~\ref{th:main}]
	Let $g(W_i,\mathbf{z}_i)=E(f_i(W_i,\pi)|\mathbf{z}_i,W_i)$, and $g(\bb{W},\mathbf{z})=(g(W_1,\mathbf{z}_1),\cdots,g(W_n,\mathbf{z}_n))^\top$. 
	
\textbf{Step 1 (asymptotic normality of $\hat\tau$).} In this step, we use the Cramer-Wold device to show that $\hat\tau$ defined in (\ref{def:regression estimator}) is asymptotic normal. We first provide a sketch of this step. Recall that by the related definitions and the proof of Theorem~\ref{th: Consistency of hat tau},
\begin{align*}
    \hat\tau=\frac{1}{n}\mathbf{1}^\top \bb{X} (
\hat\beta_1
-
\hat\beta_0
 ), 
\tau=(E\mathbf{x}_1)^\top ( \beta_1
- \beta_0
 ).
\end{align*}
 
Then
	\begin{align}
		\label{eqn:2}
		\sqrt{n}(\hat\tau-\tau)=\sqrt{n}\left(\begin{matrix}
			\frac{1}{n}\sum_{i=1}^n \mathbf{x}_i^\top & -\frac{1}{n}\sum_{i=1}^n \mathbf{x}_i^\top & 
			(\beta_1-\beta_0)^\top
		\end{matrix}\right)
		\left(\begin{matrix}
			\hat\beta_1-\beta_1 \\
			\hat\beta_0-\beta_0\\
			\frac{1}{n}\sum_{i=1}^n (\mathbf{x}_i-E\mathbf{x}_i)
		\end{matrix}\right).
	\end{align}
	 In \tb{Step 1(a)}, we prove through the Cramer-Wold device  that the vector 
 \begin{equation*}
 \frac{1}{\sqrt{n}}
     \begin{pmatrix}
         \bb{X}^\top \bb{W}(\mathbf{Y}- \bb{X}\beta_1)\\
          \bb{X}^\top (\bb{I}-\bb{W})(\mathbf{Y}- \bb{X}\beta_0)\\
           \sum_{i=1}^n (\mathbf{x}_i-E\mathbf{x}_i)
     \end{pmatrix}
 \end{equation*}
  is asymptotically normal. In \tb{Step 1(b)}, we prove that  \begin{align*}
  \sqrt{n}
      \left(\begin{matrix}
			\hat\beta_1-\beta_1 \\
			\hat\beta_0-\beta_0\\
			\frac{1}{n}\sum_{i=1}^n (\mathbf{x}_i-E\mathbf{x}_i)
		\end{matrix}\right)
  \end{align*} is asymptotically normal by combining the result in \tb{Step 1(a)} and the Slutsky's theorem. In \tb{Step 1(c)}, we give the asymptotic normality of $\hat\tau$. 

 \tb{Step 1(a).} We observe that 
 \begin{equation*}\frac{1}{n}
     \begin{pmatrix}
         \bb{X}^\top \bb{W}(\mathbf{Y}- \bb{X}\beta_1)\\
          \bb{X}^\top (\bb{I}-\bb{W})(\mathbf{Y}- \bb{X}\beta_0)\\
          \sum_{i=1}^n (\mathbf{x}_i-E\mathbf{x}_i)
     \end{pmatrix}=
     \left(\begin{matrix}
			\frac{1}{n}\sum_{i=1}^n W_i\mathbf{x}_i(f_i(1,\mni)-\mathbf{x}_i^\top\beta_1)\\
			\frac{1}{n}\sum_{i=1}^n (1-W_i)\mathbf{x}_i(f_i(0,\mni)-\mathbf{x}_i^\top\beta_0)\\
			\frac{1}{n}\sum_{i=1}^n (\mathbf{x}_i-E\mathbf{x}_i)
		\end{matrix}\right).
 \end{equation*}

 For arbitrary but fixed real vectors $\mathbf{a}_1,\mathbf{a}_2,\mathbf{a}_3$,  consider the term 
	\begin{align*}
		J=
		\left(\begin{matrix}
			\mathbf{a}_1^\top & -\mathbf{a}_2^\top & 
			\mathbf{a}_3^\top
		\end{matrix}\right)
		\left(\begin{matrix}
			\frac{1}{n}\sum_{i=1}^n W_i\mathbf{x}_i(f_i(1,\mni)-\mathbf{x}_i^\top\beta_1)\\
			\frac{1}{n}\sum_{i=1}^n (1-W_i)\mathbf{x}_i(f_i(0,\mni)-\mathbf{x}_i^\top\beta_0)\\
			\frac{1}{n}\sum_{i=1}^n (\mathbf{x}_i-E\mathbf{x}_i)
		\end{matrix}\right).
	\end{align*}
	Straightforward calculations show that \begin{align*}
		J=\frac{1}{n}\sum_{i=1}^n \left[
		\frac{W_i}{\pi}q_i\left(1,\mni\right)-\frac{1-W_i}{1-\pi}q_i\left(0,\mni\right)-a_3^\top E\mathbf{x}_1
		\right],
	\end{align*}
	where
	\begin{align*}
		q_i\left(W_i,t\right)&=\left[
		\pi a_1^\top \mathbf{x}_iW_i+(1-\pi)a_2^\top \mathbf{x}_i(1-W_i)
		\right](f_i(W_i,t)-\mathbf{x}_i^\top \beta_{W_i})\\
		&+(\pi W_i-(1-\pi)(1-W_i))a_3^\top \mathbf{x}_i.
	\end{align*}
Note that 
 \begin{align}
 \label{pf2:aeqn:1}
     E(\mathbf{x}_i(g(1,\mathbf{z}_i)-\mathbf{x}_i^\top \beta_1))=&
      E(\mathbf{x}_i(E(f_i(1,\pi)|\mathbf{z}_i)-\mathbf{x}_i^\top \bb{M}_{xx}^{-1}E(\mathbf{x}_if_i(1,\pi))))\notag\\
      =&E(\mathbf{x}_if_i(1,\pi))-E(\mathbf{x}_if_i(1,\pi))=\mathbf{0}.
 \end{align}
 Similarly, 
 \begin{align}
 \label{pf2:aeqn:2}
     E(\mathbf{x}_i(g(0,\mathbf{z}_i)-\mathbf{x}_i^\top \beta_0))=\mathbf{0}.
 \end{align}
	Then we derive that 
	\begin{align}
 \label{pf2:aeqn:3}
		&\frac{1}{n}\sum_{i=1}^n E\left[
		q_i\left(1,\mni \right)-q_i\left(0,\mni \right)
		\right] \notag\\
		=& 
		\frac{1}{n} \sum_{i=1}^n 
		E\left(\pi a_1^\top \mathbf{x}_i(g(1,\mathbf{z}_i)-\mathbf{x}_i^\top \beta_1+f_i\left(1,\mni\right)-g(1,\mathbf{z}_i))\right)\notag\\
		&-	\frac{1}{n} \sum_{i=1}^n 
		E\left((1-\pi) a_2^\top \mathbf{x}_i(g(0,\mathbf{z}_i)-\mathbf{x}_i^\top \beta_0+f_i\left(0,\mni\right)-g(0,\mathbf{z}_i))\right)
		+
		a_3^\top E\mathbf{x}_1
		\notag\\
		=&	\frac{1}{n} \sum_{i=1}^n 
		E\left(\pi a_1^\top \mathbf{x}_i\left(f_i\left(1,\mni \right)-E\left(f_i(1,\pi)|\mathbf{z}_i\right)\right)\right)\notag\\
		&-\frac{1}{n} \sum_{i=1}^n 
		E\left((1-\pi) a_2^\top \mathbf{x}_i\left(f_i\left(0,\mni \right)-E\left(f_i(0,\pi)|\mathbf{z}_i\right)\right)\right)+
		a_3^\top E\mathbf{x}_1\notag\\
		=& a_3^\top E\mathbf{x}_1+O\left(\frac{1}{n\rho_n}\right),
	\end{align}
	where the second equality is due to results in (\ref{pf2:aeqn:1}), (\ref{pf2:aeqn:2}), and we derive the last equality using a similar method as in (\ref{eqn:1}). Moreover, under Assumption~\ref{assumption:anonymous interference}, \ref{assump:random_graph}, \ref{assumption:moment_conditions}, and (\ref{pf2:aeqn:3}), it is straightforward to verify that Theorem 4 (particularly their equation (24)) in \citet{li2022random} still holds (note that our conditions for $q_i$ here are slightly weaker than their conditions for $f_i$). Then we have that
	\begin{align}
		\label{eqn:23}
		\sqrt{n}J\overset{d}{\to} N(0,\sigma_0^2+\pi(1-\pi)E(R_i+Q_i)^2),
	\end{align}
	where 
	\begin{align*}
 \sigma_0^2&=Var(q_i(1,\pi)-q_i(0,\pi)),\\
		R_i&=q_i(1,\pi)/\pi+q_i(0,\pi)/(1-\pi),\\
		Q_i&= E\left(\frac{h(U_i,U_j)}{E(h(U_i,U_j)|U_j)}\bigg|U_i
		\right)
  E\left(\frac{\partial}{\partial t}q_i(1,t)\bigg|_{t=\pi}
   -\frac{\partial}{\partial t}q_i(0,t)\bigg|_{t=\pi}\right).
	\end{align*}
	Let $
		P_{i1}=f_i(1,\pi)-\mathbf{x}_i^\top\beta_1, 
		P_{i2}=f_i(0,\pi)-\mathbf{x}_i^\top\beta_0, 
		P_{i1}'=f_i^{(1)}(1,\pi),
		P_{i2}'=f_i^{(1)}(0,\pi)$. 
	Then we have
	\begin{align*}
		q_i(1,\pi)&=\pi a_1^\top \mathbf{x}_iP_{i1}+\pi a_3^\top \mathbf{x}_i, 
		q_i(0,\pi)=(1-\pi) a_2^\top \mathbf{x}_iP_{i2}-(1-\pi) a_3^\top \mathbf{x}_i,\\	
		\frac{\partial}{\partial t}q_i(1,t)\bigg|_{t=\pi}&=\pi a_1^\top \mathbf{x}_iP_{i1}', 
		\frac{\partial}{\partial t}q_i(0,t)\bigg|_{t=\pi}=(1-\pi) a_2^\top \mathbf{x}_iP_{i2}'.
	\end{align*}
	Then $EQ_i=\pi a_1^\top E(\mathbf{x}_iP_{i1}')-(1-\pi) a_2^\top E(\mathbf{x}_iP_{i2}')$, $R_i=a_1^\top \mathbf{x}_iP_{i1}+a_2^\top \mathbf{x}_iP_{i2}$,    \\ $EQ_i^2=b\left(\pi a_1^\top E(\mathbf{x}_iP_{i1}')-(1-\pi) a_2^\top E(\mathbf{x}_iP_{i2}')\right)^2$ where $b$ is defined in (\ref{def:b}). Note that similar to (\ref{pf2:aeqn:1}), (\ref{pf2:aeqn:2}), we have \begin{align}
	    \label{pf2:aeqn:4}
     E(\mathbf{x}_iP_{i1})=E(\mathbf{x}_iP_{i2})=\mathbf{0}.
	\end{align} 
 By (\ref{pf2:aeqn:4}) and the fact that $Q_i$ is independent of $R_i$, we have $E(R_iQ_i)=0$. Then 
	\begin{align*}
		\sigma_0^2+\pi(1-\pi)E(R_i+Q_i)^2=&
		Var(
		\pi a_1^\top \mathbf{x}_iP_{i1}
		-(1-\pi) a_2^\top \mathbf{x}_iP_{i2}+ a_3^\top \mathbf{x}_i
		)\\
		+& \pi(1-\pi) E\left(
		a_1^\top \mathbf{x}_iP_{i1}+a_2^\top \mathbf{x}_iP_{i2}
		\right)^2 \\
		+& \pi(1-\pi) b\left(\pi a_1^\top E(\mathbf{x}_iP_{i1}')-(1-\pi) a_2^\top E(\mathbf{x}_iP_{i2}')\right)^2 .
	\end{align*}

By the Cramer-Wold device, we have
	\begin{equation}
 \label{pf2:aeqn:5}
		\sqrt{n} 	\left(\begin{matrix}
			\frac{1}{n}\sum_{i=1}^n W_i\mathbf{x}_i(f_i(1,\mni)-\mathbf{x}_i^\top\beta_1)\\
			\frac{1}{n}\sum_{i=1}^n (1-W_i)\mathbf{x}_i(f_i(0,\mni)-\mathbf{x}_i^\top\beta_0)\\
			\frac{1}{n}\sum_{i=1}^n (\mathbf{x}_i-E\mathbf{x}_i)
		\end{matrix}\right)
		\overset{d}{\to} N(0,\Sigma),
	\end{equation}
	where $\Sigma_{11}=\pi E(P_{11}^2\mathbf{x}_1\mathbf{x}_1^\top)
	+ b\pi^3(1-\pi)E(P_{11}'^2\mathbf{x}_1\mathbf{x}_1^\top)$, 
	$\Sigma_{22}=(1-\pi) E(P_{12}^2\mathbf{x}_1\mathbf{x}_1^\top)
	+ b\pi(1-\pi)^3E(P_{12}'^2\mathbf{x}_1\mathbf{x}_1^\top)$, 
	$\Sigma_{33}=\bb{M}_{xx}-(E\mathbf{x}_1)(E\mathbf{x}_1)^\top$, 
	$\Sigma_{12}=b\pi^2(1-\pi)^2E(\mathbf{x}_1P_{11}')E(\mathbf{x}_1^\top P_{12}')$, 
	$\Sigma_{13}=\pi E(P_{11}\mathbf{x}_1\mathbf{x}_1^\top)$, and 
	$\Sigma_{23}=(1-\pi) E(P_{12}\mathbf{x}_1\mathbf{x}_1^\top)$.

 \tb{Step 1(b).} We note that
  \begin{align*}
      \left(\begin{matrix}
			\hat\beta_1-\beta_1 \\
			\hat\beta_0-\beta_0\\
			\frac{1}{n}\sum_{i=1}^n (\mathbf{x}_i-E\mathbf{x}_i)
		\end{matrix}\right)
  = 
      \left(\begin{matrix}
			\bb{X}^\top\bb{W}\bb{X} & & \\
   & \bb{X}^\top(\bb{I}-\bb{W})\bb{X}&\\
   & & n\bb{I} 
		\end{matrix}\right)^{-1}
    \begin{pmatrix}
         \bb{X}^\top \bb{W}(\mathbf{Y}- \bb{X}\beta_1)\\
          \bb{X}^\top (\bb{I}-\bb{W})(\mathbf{Y}- \bb{X}\beta_0)\\
           \sum_{i=1}^n (\mathbf{x}_i-E\mathbf{x}_i)
     \end{pmatrix}.
  \end{align*}
By (\ref{pf1:eqn1}), $\bb{X}^\top\bb{W}\bb{X}/n\overset{a.s.}{\to}\pi\bb{M}_{xx}$. Similarly, $\bb{X}^\top(\bb{I}-\bb{W})\bb{X}/n\overset{a.s.}{\to}(1-\pi)\bb{M}_{xx}$. Then, applying (\ref{pf2:aeqn:5}) and Slutsky's theorem, we obtain
	\begin{equation}
		\sqrt{n} 	\left(\begin{matrix}
			\hat\beta_1-\beta_1\\
			\hat\beta_0-\beta_0\\
			\frac{1}{n}\sum_{i=1}^n (\mathbf{x}_i-E\mathbf{x}_i)
		\end{matrix}\right)
		\overset{d}{\to} N(0,\tilde\Sigma),
	\end{equation}
	where 
	$\tilde \Sigma_{11}=\frac{1}{\pi^2}\bb{M}_{xx}^{-1}\Sigma_{11}\bb{M}_{xx}^{-1}$, $\tilde \Sigma_{22}=\frac{1}{(1-\pi)^2}\bb{M}_{xx}^{-1}\Sigma_{22}\bb{M}_{xx}^{-1}$, $\tilde \Sigma_{33}=\Sigma_{33}$, $\tilde \Sigma_{12}=\frac{1}{\pi(1-\pi)}\bb{M}_{xx}^{-1}\Sigma_{12}\bb{M}_{xx}^{-1}$, $\tilde \Sigma_{13}=\frac{1}{\pi}\bb{M}_{xx}^{-1}\Sigma_{13}$, and 
	$\tilde \Sigma_{23}=\frac{1}{1-\pi}\bb{M}_{xx}^{-1}\Sigma_{23}$.

 \tb{Step 1(c).} 
 Recall (\ref{eqn:2}), (\ref{pf2:aeqn:6}), and the Slutsky's theorem, we have  $	\sqrt{n}(\hat\tau-\tau)\overset{d}{\to} N\left(0,V_{reg}\right)$, where 
	\begin{align}
		\label{eqn:24}
		V_{reg}&=(E\mathbf{x}_1)^\top \tilde\Sigma_{11}(E\mathbf{x}_1)
		+(E\mathbf{x}_1)^\top \tilde \Sigma_{22} (E\mathbf{x}_1)
		+(\beta_1-\beta_0)^\top \tilde \Sigma_{33} (\beta_1-\beta_0)\notag\\
			&- 2(E\mathbf{x}_1)^\top \tilde\Sigma_{12} (E\mathbf{x}_1)
		+2(E\mathbf{x}_1)^\top \tilde\Sigma_{13}(\beta_1-\beta_0)
		-2(E\mathbf{x}_1)^\top \tilde\Sigma_{23}(\beta_1-\beta_0)\notag\\
		&=\frac{1}{\pi}EP_{11}^2
		+\frac{1}{1-\pi}EP_{12}^2
		+(\beta_1-\beta_0)^\top (\bb{M}_{xx}-(E\mathbf{x}_1)(E\mathbf{x}_1)^\top)(\beta_1-\beta_0)\notag\\
		&+b\pi(1-\pi)(E(P_{11}'-P_{12}'))^2.
	\end{align}
	The second equality above is due to the fact that $E(\mathbf{x}_iP_{i1})=E(\mathbf{x}_iP_{i2})=\mathbf{0}$ in (\ref{pf2:aeqn:4}) and $E\mathbf{x}_1^\top \bb{M}_{xx}^{-1}=(1,0,\cdots,0)$ in (\ref{pf2:aeqn:7}).  Then \eqref{def:Vreg}  follows by noting that \eqref{eqn:24} and \eqref{def:Vreg} are identical.
	
	\textbf{Step 2 ($V_{dim}-V_{reg}\ge 0$).} In this step, we prove that \begin{equation}
 \label{pf2:aeqn:8}
		V_{dim}-V_{reg}
		=\frac{1}{\pi(1-\pi)}\left(E(\mathbf{z}_1\tilde f)-E\mathbf{z}_1E\tilde f\right)^\top 
		\left(E(\mathbf{z}_1\mathbf{z}_1^\top)\right)^{-1}
		\left(E(\mathbf{z}_1\tilde f)-E\mathbf{z}_1E\tilde f\right)
		\ge 0,
	\end{equation}
 where $\tilde f=(1-\pi)f_1(1,\pi)+\pi f_1(0,\pi)$. 
 
 We recall that $V_{dim}$ is the asymptotic variance when $\mathbf{x}_i=1$.  
 By (\ref{eqn:24}) and noticing that the term $b\pi(1-\pi)(E(P_{11}'-P_{12}'))^2$ in (\ref{eqn:24}) remains unchanged in $V_{reg}$ and $V_{dim}$, we derive that 
\begin{align}\label{pf2:aeqn:11}
    V_{dim}-V_{reg}&=\frac{1}{\pi}\left[
    E(f_1(1,\pi)-Ef_1(1,\pi))^2-
    E\left(f_1(1,\pi)-\mathbf{x}_i^\top\beta_1 \right)^2
    \right]\notag\\
    &
    +\frac{1}{1-\pi}\left[
    E(f_1(0,\pi)-Ef_1(0,\pi))^2-
    E\left(f_1(0,\pi)-\mathbf{x}_i^\top\beta_0\right)^2
    \right]\notag\\
    &
    -(\beta_1-\beta_0)^\top(\bb{M}_{xx}-E\mathbf{x}_1(E\mathbf{x}_1)^\top)(\beta_1-\beta_0).
\end{align}

Since $\mathbf{x}_1=(1,\mathbf{z}_i^\top)^\top$, and $E(\mathbf{x}_iP_{i1})=\mathbf{0}$ in (\ref{pf2:aeqn:4}), we have 
	\begin{equation*}
		E\left((f_i(1,\pi)-\mathbf{x}_i^\top \beta_1)(\mathbf{x}_i^\top\beta_1-Ef_i(1,\pi))\right)=E(P_{i1}\mathbf{x}_i^\top )\beta_1-Ef_1(1,\pi)E(P_{i1})=0.
	\end{equation*}
	Therefore, 
	\begin{align}
 \label{pf2:aeqn:9}
		E\left(f_i(1,\pi)-Ef_i(1,\pi)\right)^2&=E\left(f_i(1,\pi)-\mathbf{x}_i^\top\beta_1\right)^2+E\left(\mathbf{x}_i^\top\beta_1-Ef_i(1,\pi)\right)^2.
	\end{align}
	Similarly, 
 \begin{align}
 \label{pf2:aeqn:10}
		E\left(f_i(0,\pi)-Ef_i(0,\pi)\right)^2&=E\left(f_i(0,\pi)-\mathbf{x}_i^\top\beta_0\right)^2+E\left(\mathbf{x}_i^\top\beta_0-Ef_i(0,\pi)\right)^2.
  \end{align}
Moreover, by the definitions of $\beta_1$ and $\beta_0$, we have \begin{align}
    \label{pf2:aeqn:12}
    E\left(\mathbf{x}_i^\top\beta_1-Ef_i(1,\pi)\right)=
    E\left(\mathbf{x}_i^\top\beta_0-Ef_i(0,\pi)\right)=0.
\end{align}

By (\ref{pf2:aeqn:9}), (\ref{pf2:aeqn:10}), and (\ref{pf2:aeqn:12}), equation (\ref{pf2:aeqn:11}) reduces to 
	\begin{align}
		\label{eqn:pf1}
		V_{dim}-V_{reg}=&\frac{1}{\pi}E\left(\mathbf{x}_i^\top\beta_1-Ef_i(1,\pi)\right)^2+\frac{1}{1-\pi}E\left(\mathbf{x}_i^\top\beta_0-Ef_i(0,\pi)\right)^2\notag\\
		&-(\beta_1-\beta_0)^\top(\bb{M}_{xx}-E\mathbf{x}_1(E\mathbf{x}_1)^\top)(\beta_1-\beta_0)\notag\\
		=&\frac{1}{\pi}\left(\beta_1^\top \bb{M}_{xx}\beta_1-(Ef_1(1,\pi))^2\right)
		+\frac{1}{1-\pi}\left(\beta_0^\top \bb{M}_{xx}\beta_0-(Ef_i(0,\pi))^2\right)\notag\\&
		-\beta_1^\top \bb{M}_{xx}\beta_1
		-\beta_0^\top \bb{M}_{xx}\beta_0
		+2\beta_1^\top \bb{M}_{xx}\beta_0
		+(E(f_1(1,\pi)-f_1(0,\pi)))^2\notag\\
		=&\frac{1-\pi}{\pi}\beta_1^\top \bb{M}_{xx}\beta_1+
		\frac{\pi}{1-\pi}\beta_0^\top \bb{M}_{xx}\beta_0+
		2\beta_1^\top \bb{M}_{xx}\beta_0\notag\\
		&-\frac{1-\pi}{\pi}(Ef_1(1,\pi))^2
		-\frac{\pi}{1-\pi}(Ef_1(0,\pi))^2
		-2(Ef_1(1,\pi))(Ef_1(0,\pi))\notag\\
		=&\frac{1}{\pi(1-\pi)}\left[((1-\pi)\beta_1+\pi\beta_0)^\top \bb{M}_{xx}((1-\pi)\beta_1+\pi\beta_0)\right.\notag\\
		&\left.	-((1-\pi)Ef_1(1,\pi)+\pi Ef_1(0,\pi))^2\right]\notag\\
		=&\frac{1}{\pi(1-\pi)}\left[E(\mathbf{x}_1\tilde f)^\top \bb{M}_{xx}^{-1}E(\mathbf{x}_1\tilde f)-(E\tilde f)^2\right],
	\end{align}
	where $\tilde f=(1-\pi)f_1(1,\pi)+\pi f_1(0,\pi)$. Recall that $\mathbf{x}_1=(1,\mathbf{z}_1^\top)^\top$, and that $\bb{M}_{xx}=\left(\begin{matrix}
		1 & E\mathbf{z}_1^\top\\
		E\mathbf{z}_1 &  E(\mathbf{z}_1 \mathbf{z}_1^\top)
	\end{matrix}\right)$, with its inverse given by 
	\begin{equation*}
		\bb{M}_{xx}^{-1}=\left(\begin{matrix}
			1+E\mathbf{z}_1^\top\left(E(\mathbf{z}_1\mathbf{z}_1^\top)\right)^{-1}E\mathbf{z}_1 & -E\mathbf{z}_1^\top \left(E(\mathbf{z}_1\mathbf{z}_1^\top)\right)^{-1}\\
			-\left(E(\mathbf{z}_1\mathbf{z}_1^\top)\right)^{-1}E\mathbf{z}_1 &  \left(E(\mathbf{z}_1\mathbf{z}_1^\top)\right)^{-1}
		\end{matrix}\right).
	\end{equation*}
	It follows that
	\begin{align*}
		&E(\mathbf{x}_1\tilde f)^\top \bb{M}_{xx}^{-1}E(\mathbf{x}_1\tilde f)-(E\tilde f)^2\\
		&=(E\tilde f)^2 E\mathbf{z}_1^\top (E(\mathbf{z}_1\mathbf{z}_1^\top))^{-1} E\mathbf{z}_1+ E(\mathbf{z}_1^\top\tilde f)(E(\mathbf{z}_1\mathbf{z}_1^\top))^{-1} E(\mathbf{z}_1\tilde f)-2E\tilde f E\mathbf{z}_1^\top 
		(E(\mathbf{z}_1\mathbf{z}_1^\top))^{-1}
		E(\mathbf{z}_1\tilde f)\\
		&=\left(E(\mathbf{z}_1\tilde f)-E\mathbf{z}_1E\tilde f\right)^\top
		(E(\mathbf{z}_1\mathbf{z}_1^\top))^{-1}
		\left(E(\mathbf{z}_1\tilde f)-E\mathbf{z}_1E\tilde f\right)\ge 0.
	\end{align*}
	Thus, (\ref{pf2:aeqn:8}) follows. 
\end{proof}

\begin{proof}[Proof of Theorem~\ref{th:4}]
The outline of this proof is similar to that of Theorem~\ref{th:main}. Let $n_1=\sum_{i=1}^n W_i, n_0=\sum_{i=1}^n (1-W_i)$. Recall the definition of \(\hat{\tau}(\alpha_1, \alpha_0)\) in (\ref{def:class of regression adjustment}). In \tb{Step 1}, we establish the asymptotic normality of \(\hat{\tau}(\alpha_1, \alpha_0)\). In \tb{Step 2}, we clarify the distinction between the asymptotic variance \(\tilde{V}(\alpha_1, \alpha_0)\) of \(\hat{\tau}(\alpha_1, \alpha_0)\) and \(V_{reg}\), as defined in Theorem~\ref{th:main}.

Let \begin{equation*}
	\tilde{\tau}(\alpha_1, \alpha_0) = \sum_{i=1}^n \left( \frac{W_i(Y_i - \alpha_1^\top (\mathbf{x}_i - \bar{x}))}{\sum_{j=1}^n W_j} - \frac{(1-W_i)(Y_i - \alpha_0^\top (\mathbf{x}_i - \bar{x}))}{\sum_{j=1}^n (1-W_j)} \right).
\end{equation*}
Note that $\tilde \tau(\alpha_1,\alpha_0)=\hat\tau(\alpha_1^{(-1)},\alpha_0^{(-1)})$. Then Theorem~\ref{th:3} reduces to:

\textit{For any fixed $\alpha_1,\alpha_0$, $
	\sqrt{n}(\tilde{\tau}(\alpha_1, \alpha_0) - \tau) \overset{d}{\to} N\left(0, \tilde{V}(\alpha_1, \alpha_0)\right),
$
	where 
	\[
		\tilde{V}(\alpha_1, \alpha_0) - V_{reg} = \frac{1}{\pi(1-\pi)} u(\alpha_1, \alpha_0)^\top \left( \bb{M}_{xx} - E\mathbf{x}_1 E\mathbf{x}_1^\top \right) 
		u(\alpha_1, \alpha_0) \ge  0,\]
		and \(u(\alpha_1, \alpha_0) = (1 - \pi)(\alpha_1  - \beta_1 ) + \pi(\alpha_0 - \beta_0 )\).
	}

\textbf{Step 1 (asymptotic normality of $\tilde \tau(\alpha_1,\alpha_0)$).} 

Notice that
 \begin{align*}
		\sqrt{n}(\tilde\tau(\alpha_1,\alpha_0)-\tau)=\sqrt{n}\left(\begin{matrix}
			\frac{n}{n_1}  & -\frac{n}{n_0} & 
			(\alpha_1-\alpha_0)^\top
		\end{matrix}\right)
		\left(\begin{matrix}
		\frac{1}{n}\sum_{i=1}^n W_i(Y_i-(\mathbf{x}_i^\top-E\mathbf{x}_1^\top)\alpha_1-Ef_i(1,\pi)) \\
			\frac{1}{n}\sum_{i=1}^n (1-W_i)(Y_i-(\mathbf{x}_i^\top-E\mathbf{x}_1^\top)\alpha_0-Ef_i(0,\pi)) \\
			\frac{1}{n}\sum_{i=1}^n (\mathbf{x}_i-E\mathbf{x}_1)
		\end{matrix}\right).
	\end{align*}
	We use the Cramer-Wold device to give the proof. For arbitrary but fixed real numbers $a_1,a_2$ and vector $a_3$, we first consider the term 
	\begin{align*}
		J=
		\left(\begin{matrix}
			a_1  & -a_2  & 
			a_3^\top
		\end{matrix}\right)
		\left(\begin{matrix}
		\frac{1}{n}\sum_{i=1}^n W_i(Y_i-(\mathbf{x}_i^\top-E\mathbf{x}_1^\top)\alpha_1-Ef_i(1,\pi)) \\
		\frac{1}{n}\sum_{i=1}^n (1-W_i)(Y_i-(\mathbf{x}_i^\top-E\mathbf{x}_1^\top)\alpha_0-Ef_i(0,\pi)) \\
		\frac{1}{n}\sum_{i=1}^n (\mathbf{x}_i-E\mathbf{x}_1)
	\end{matrix}\right).
	\end{align*}
 Note that $Y_i=f_i(W_i,M_i/N_i)$. 
	Straightforward calculations show that \begin{align*}
		J=\frac{1}{n}\sum_{i=1}^n \left[
		\frac{W_i}{\pi}q_i\left(1,\mni\right)-\frac{1-W_i}{1-\pi}q_i\left(0,\mni\right)-a_3^\top E\mathbf{x}_1
		\right],
	\end{align*}
	where
	\begin{align*}
		q_i\left(1,t\right)&=
		\pi a_1 (f_i(1,t)-Ef_i(1,\pi)-(\mathbf{x}_i^\top -E\mathbf{x}_i^\top )\alpha_1)+\pi a_3^\top \mathbf{x}_i,\\
		q_i\left(0,t\right)&=
		(1-\pi) a_2 (f_i(0,t)-Ef_i(0,\pi)-(\mathbf{x}_i^\top -E\mathbf{x}_i^\top )\alpha_0)-(1-\pi) a_3^\top \mathbf{x}_i.
	\end{align*}
	Using arguments similar to those in (\ref{pf2:aeqn:3}), we obtain that
	\begin{align*}
		\frac{1}{n}\sum_{i=1}^n E\left[
		q_i(1,\mni)-q_i(0,\mni)
		\right]=a_3^\top E\mathbf{x}_i+o(1/\sqrt{n}). 
	\end{align*}
 Then similar to (\ref{eqn:23}), we have  
\begin{align*}
	\sqrt{n}J\overset{d}{\to} N(0,\sigma_0^2+\pi(1-\pi)E(R_i+Q_i)^2),
\end{align*}
where 
\begin{align*}
	\sigma_0^2&=Var(q_i(1,\pi)-q_i(0,\pi)),\\
	R_i&=q_i(1,\pi)/\pi+q_i(0,\pi)/(1-\pi),\\
	Q_i&=E\left(\frac{h(U_i,U_j)}{E(h(U_i,U_j)|U_j)}\bigg|U_i
		\right)
  E\left(\frac{\partial}{\partial t}q_i(1,t)\bigg|_{t=\pi}
   -\frac{\partial}{\partial t}q_i(0,t)\bigg|_{t=\pi}\right).
\end{align*} 
Then by the Cramer-Wold device, 
\begin{align*}
	\sqrt{n}\left(\begin{matrix}
		\frac{1}{n}\sum_{i=1}^n W_i(Y_i-(\mathbf{x}_i^\top-E\mathbf{x}_1^\top)\alpha_1-Ef_i(1,\pi)) \\
		\frac{1}{n}\sum_{i=1}^n (1-W_i)(Y_i-(\mathbf{x}_i^\top-E\mathbf{x}_1^\top)\alpha_0-Ef_i(0,\pi)) \\
		\frac{1}{n}\sum_{i=1}^n (\mathbf{x}_i-E\mathbf{x}_1)
	\end{matrix}\right)
\end{align*} is asymptotically normal. As a result, $	\sqrt{n}(\tilde\tau(\alpha_1,\alpha_0)-\tau)$ is also asymptotically normal. By carefully deriving the asymptotic variance $\tilde V(\alpha_1,\alpha_0)$ which is similar to the derivation from (\ref{eqn:23}) to (\ref{eqn:24}), we have that
\begin{align*} 
	\tilde V(\alpha_1,\alpha_0)=&\frac{1}{\pi} E\left(f_i(1,\pi)-Ef_i(1,\pi)-(\mathbf{x}_i^\top-E\mathbf{x}_i^\top)\alpha_1\right)^2\\
	+&\frac{1}{1-\pi} E\left(f_i(0,\pi)-Ef_i(0,\pi)-(\mathbf{x}_i^\top-E\mathbf{x}_i^\top)\alpha_0\right)^2\\
	+&(\alpha_1-\alpha_0)^\top(\bb{M}_{xx}-(E\mathbf{x}_1)(E\mathbf{x}_1)^\top)(\alpha_1-\alpha_0)
	\\+&b\pi(1-\pi)(E(f_i'(1,\pi)-f_i'(0,\pi)))^2
	\\+&2(\alpha_1-\alpha_0)^\top E(\mathbf{x}_i(
	(f_i(1,\pi)-Ef_i(1,\pi)-(\mathbf{x}_i^\top-E\mathbf{x}_i^\top)\alpha_1)-\\
	&(f_i(0,\pi)-Ef_i(0,\pi)-(\mathbf{x}_i^\top-E\mathbf{x}_i^\top)\alpha_0)
	)),
\end{align*}
with $b$ defined in (\ref{def:b}). 

\textbf{Step 2 ($\tilde V(\alpha_1,\alpha_0)-V_{reg}\ge 0$).} In this step, we prove (\ref{eqn:th3 2}).  
By (\ref{pf2:aeqn:4}), 
\begin{align}\label{pf3:aeqn:3}
    E\left\{[f_i(0,\pi)- \mathbf{x}_i^\top \beta_0]x_i^\top (\beta_0-\alpha_0)\right\}=0,\notag\\
     E\left\{[f_i(1,\pi)- \mathbf{x}_i^\top \beta_1]x_i^\top (\beta_1-\alpha_1)\right\}=0.
\end{align}
Consequently,  we have
\begin{align*}
	&E\left(f_i(0,\pi)-Ef_i(0,\pi)-(\mathbf{x}_i^\top-E\mathbf{x}_i^\top)\alpha_0\right)^2=E\left(f_i(0,\pi)-\mathbf{x}_i^\top\alpha_0\right)^2
-(Ef_i(0,\pi)-E\mathbf{x}_i^\top\alpha_0)^2\\
&=E\left(f_i(0,\pi)-\mathbf{x}_i^\top\beta_0\right)^2+
E\left(\mathbf{x}_i^\top\beta_0-\mathbf{x}_i^\top\alpha_0\right)^2
-(Ef_i(0,\pi)-E\mathbf{x}_i^\top\alpha_0)^2,\\
&E\left(f_i(1,\pi)-Ef_i(1,\pi)-(\mathbf{x}_i^\top-E\mathbf{x}_i^\top)\alpha_1\right)^2=E\left(f_i(1,\pi)-\mathbf{x}_i^\top\alpha_1\right)^2-(Ef_i(1,\pi)-E\mathbf{x}_i^\top\alpha_1)^2\\&=E\left(f_i(1,\pi)-\mathbf{x}_i^\top\beta_1\right)^2+
E\left(\mathbf{x}_i^\top\beta_1-\mathbf{x}_i^\top\alpha_1\right)^2
-(Ef_i(1,\pi)-E\mathbf{x}_i^\top\alpha_1)^2.
\end{align*}
Therefore,
\begin{align}\label{pf3:aeqn:5}
  \tilde V(\alpha_1,\alpha_0)-V_{reg}=&
\frac{1}{\pi} E\left(\mathbf{x}_i^\top\beta_1-\mathbf{x}_i^\top\alpha_1\right)^2
-\frac{1}{\pi} (Ef_i(1,\pi)-E\mathbf{x}_i^\top\alpha_1)^2\notag\\
+& \frac{1}{1-\pi} E\left(\mathbf{x}_i^\top\beta_0-\mathbf{x}_i^\top\alpha_0\right)^2
-\frac{1}{1-\pi} (Ef_i(0,\pi)-E\mathbf{x}_i^\top\alpha_0)^2\notag\\
+& (\alpha_1-\alpha_0)^\top(\bb{M}_{xx}-(E\mathbf{x}_1)(E\mathbf{x}_1)^\top)(\alpha_1-\alpha_0)\notag\\
-& (\beta_1-\beta_0)^\top(\bb{M}_{xx}-(E\mathbf{x}_1)(E\mathbf{x}_1)^\top)(\beta_1-\beta_0)\notag\\
+& 2(\alpha_1-\alpha_0)^\top E(\mathbf{x}_i(
(f_i(1,\pi)-Ef_i(1,\pi)-(\mathbf{x}_i^\top-E\mathbf{x}_i^\top)\alpha_1)-\notag\\
&(f_i(0,\pi)-Ef_i(0,\pi)-(\mathbf{x}_i^\top-E\mathbf{x}_i^\top)\alpha_0)
)).
\end{align}
By noticing that  $Ef_i(1,\pi)=E\mathbf{x}_1^\top \beta_1, Ef_i(0,\pi)=E\mathbf{x}_1^\top \beta_0$, we have
\begin{align}
\label{pf3:aeqn:6}
E\left(\mathbf{x}_i^\top\beta_1-\mathbf{x}_i^\top\alpha_1\right)^2
-  (Ef_i(1,\pi)-E\mathbf{x}_i^\top\alpha_1)^2 =&var(\mathbf{x}_i^\top (\beta_1-\alpha_1)),\notag\\
E\left(\mathbf{x}_i^\top\beta_0-\mathbf{x}_i^\top\alpha_0\right)^2
-  (Ef_i(0,\pi)-E\mathbf{x}_i^\top\alpha_0)^2 =&var(\mathbf{x}_i^\top (\beta_0-\alpha_0)).
\end{align}

Moreover, by the definitions of $\beta_1, \beta_0$ and the fact   that  $Ef_i(1,\pi)=E\mathbf{x}_1^\top \beta_1, Ef_i(0,\pi)=E\mathbf{x}_1^\top \beta_0$, we derive that
\begin{align}
    \label{pf3:aeqn:7}
   E(\mathbf{x}_i
(f_i(1,\pi)-Ef_i(1,\pi)-(\mathbf{x}_i^\top-E\mathbf{x}_i^\top)\alpha_1))&=
\bb{M}_{xx}\left[
\beta_1-\alpha_1-\bb{M}_{xx}^{-1}E\mathbf{x}_iE\mathbf{x}_i^\top(\beta_1- \alpha_1)
\right]\notag\\
&=\bb{M}_{xx}(
\beta_1-\alpha_1)- E\mathbf{x}_iE\mathbf{x}_i^\top(\beta_1- \alpha_1), \notag\\
 E(\mathbf{x}_i
(f_i(0,\pi)-Ef_i(0,\pi)-(\mathbf{x}_i^\top-E\mathbf{x}_i^\top)\alpha_0))&= \bb{M}_{xx}(
\beta_0-\alpha_0)- E\mathbf{x}_iE\mathbf{x}_i^\top(\beta_0- \alpha_0).
\end{align}
By (\ref{pf3:aeqn:6}) and (\ref{pf3:aeqn:7}), we get that (\ref{pf3:aeqn:5}) simplifies to
\begin{align*}
	 \tilde V(\alpha_1,\alpha_0)-V_{reg}
	 =\frac{1}{\pi}&var(\mathbf{x}_i^\top (\beta_1-\alpha_1))+
	 \frac{1}{1-\pi}var(\mathbf{x}_i^\top (\beta_0-\alpha_0))\\
	  +& (\alpha_0-\alpha_1)^\top \bb{M}_{xx} (\alpha_0-\alpha_1)
	  -(\alpha_0-\alpha_1)^\top E\mathbf{x}_1E\mathbf{x}_1 ^\top  (\alpha_0-\alpha_1)\\
	  -& (\beta_0-\beta_1)^\top \bb{M}_{xx} (\beta_0-\beta_1)
	  +(\beta_0-\beta_1)^\top E\mathbf{x}_1E\mathbf{x}_1 ^\top  (\beta_0-\beta_1)\\
	  -& 2(\alpha_0-\alpha_1)^\top \bb{M}_{xx} (\beta_1-\alpha_1)
	  +2(\alpha_0-\alpha_1)^\top E\mathbf{x}_1E\mathbf{x}_1 ^\top  (\beta_1-\alpha_1)\\
	  +& 2(\alpha_0-\alpha_1)^\top \bb{M}_{xx} (\beta_0-\alpha_0)
	  -2(\alpha_0-\alpha_1)^\top E\mathbf{x}_1E\mathbf{x}_1 ^\top  (\beta_0-\alpha_0)\\
	  = \frac{1}{\pi}&(\beta_1-\alpha_1)^\top (\bb{M}_{xx}-E\mathbf{x}_1E\mathbf{x}_1 ^\top ) (\beta_1-\alpha_1)\\
	  +& \frac{1}{1-\pi}(\beta_0-\alpha_0)^\top (\bb{M}_{xx}-E\mathbf{x}_1E\mathbf{x}_1 ^\top ) (\beta_0-\alpha_0)\\
	  -&(\beta_1-\alpha_1-\beta_0+\alpha_0)^\top (\bb{M}_{xx}-E\mathbf{x}_1E\mathbf{x}_1 ^\top ) (\beta_1-\alpha_1-\beta_0+\alpha_0)\\
	  =  	\frac{1}{\pi}&\frac{1}{1-\pi}
	 ((1-\pi)(\alpha_1-\beta_1)+\pi(\alpha_0-\beta_0))^\top
	  \left(\bb{M}_{xx}-E\mathbf{x}_1E\mathbf{x}_1^\top\right)\\&
	 ((1-\pi)(\alpha_1-\beta_1)+\pi(\alpha_0-\beta_0))\ge 0.
\end{align*}
Then the proof is complete.
\end{proof}

\begin{proof}[Proof of Theorem \ref{th:3}]

By Theorem~\ref{th: Consistency of hat tau} we have that $\hat\beta_1\overset{p}{\to}\beta_1$ and $\hat\beta_0\overset{p}{\to}\beta_0$. By Lemma~\ref{lem:1}, we have  $\hat b\overset{p}{\to} b$. By Lemma~\ref{lem:2}, we obtain that 
	\begin{align*}
		\frac{\sum_{i=1}^n (1-W_i)(Y_i-\mathbf{x}_i^\top \hat\beta_0)^2}{\sum_{i=1}^n (1-W_i)}&\overset{p}{\to} E(f_1(0,\pi)-\mathbf{x}_1^\top\beta_0)^2,\\
			\frac{\sum_{i=1}^n W_i(Y_i-\mathbf{x}_i^\top \hat\beta_1)^2}{\sum_{i=1}^n W_i}&\overset{p}{\to} E(f_1(1,\pi)-\mathbf{x}_1^\top\beta_1)^2.
	\end{align*}
By Lemma~\ref{lem:3}, we have that $	\hat{E}f_1^{(1)}(1,\pi)\overset{p}{\to} Ef_1^{(1)}(1,\pi)$ and $	\hat{E}f_1^{(1)}(0,\pi)\overset{p}{\to} Ef_1^{(1)}(0,\pi)$. Moreover, it's obvious that $
   \sum_{i=1}^n \mathbf{x}_i/n\overset{p}{\to}E\mathbf{x}_1.$ 
Then the consistency of $\hat V_{reg}$ in (\ref{eqn:4}) follows by combining all the above results.
\end{proof}

\begin{proof}[Proof of Theorem \ref{th:6}]
The proof can be obtained by following the steps in proof of Theorem \ref{th:4}. We omit it as there is no further difficulty.
\end{proof}

\begin{proof}[Proof of Theorem \ref{th:5}]

Recall that 
	\begin{align*} 
		\hat\tau_{np}=&\frac{1}{n}\sum_{i=1}^n \left(
		\frac{\frac{1}{n   \tilde h^p}\sum_{j=1}^n K\left(\frac{\mathbf{z}_i-\mathbf{z}_j}{\tilde h}\right)Y_j\frac{W_j}{\pi}}{\hat p_1(\mathbf{z}_i)
		}-
		\frac{\frac{1}{n\tilde   h^p}\sum_{j=1}^n K\left(\frac{\mathbf{z}_i-\mathbf{z}_j}{\tilde h}\right)Y_j\frac{1-W_j}{1-\pi}}{
			\hat p_2(\mathbf{z}_i)
		}
		\right)\\&
		I\left(	\tilde  p_1(\mathbf{z}_i)>\tilde b, \tilde p_2(\mathbf{z}_i)>\tilde b,\hat p(\mathbf{z}_i)>1.01\tilde b\right),
\end{align*}
	where
	\begin{align*}
		\hat p(\mathbf{z})&= \frac{1}{n   \tilde h^p}\sum_{j=1}^n	K\left(\frac{\mathbf{z}-\mathbf{z}_j}{\tilde h}\right),\\
		\hat p_1(\mathbf{z})&= \frac{1}{n   \tilde h^p\pi}\sum_{j=1}^n	K\left(\frac{\mathbf{z}-\mathbf{z}_j}{\tilde h}\right)W_j, 	
		\tilde p_1(\mathbf{z})= \frac{1}{n \tilde h^p\hat\pi}\sum_{j=1}^n	K\left(\frac{\mathbf{z}-\mathbf{z}_j}{\tilde h}\right)W_j,\\
		\hat p_2(\mathbf{z})&= \frac{1}{n \tilde h^p(1-\pi)}\sum_{j=1}^n	K\left(\frac{\mathbf{z}-\mathbf{z}_j}{\tilde h}\right)(1-W_j), 
			\tilde p_2(\mathbf{z})= \frac{1}{n \tilde h^p(1-\hat\pi)}\sum_{j=1}^n	K\left(\frac{\mathbf{z}-\mathbf{z}_j}{\tilde h}\right)(1-W_j). 
	\end{align*}
	
	We provide a sketch of the proof. In \tb{Step 1}, we define a related value $\bar{\delta}$ in (\ref{pf6:eqn1}), which differs from $\hat\tau_{np}$ in terms of the indicator function. Intuitively, \(\bar{\delta}\) is close to \(\hat{\tau}_{np}\), since all \(\tilde{p}_1(\mathbf{z}), \tilde{p}_2(\mathbf{z}), \hat{p}(\mathbf{z})\) estimate \(p(\mathbf{z})\) in the sense of Lemma~\ref{lem4th6:1}. We prove in this step that
	\begin{align*}
		\sqrt{n}(\bar{\delta}-E\bar{\delta})&=
		\frac{1}{\sqrt{n}}\sum_{i=1}^n\left(
		f_i(W_i,\pi)\left(\frac{W_i}{\pi}-\frac{1-W_i}{1-\pi}\right)+
		E(f_i(1,\pi)-f_i(0,\pi)|\mathbf{z}_i)
		\right.\\&-\left.
		\left(
		E(f_i(1,\pi)|\mathbf{z}_i)
		\frac{W_i}{\pi}-
		E(f_i(0,\pi)|\mathbf{z}_i)\frac{1-W_i}{1-\pi}
		\right)-\tau
		\right.\\&+\left.
		(W_i-\pi)
		\sum_{j:j\neq i}\frac{E_{ij}}{N_j}
		E\left(
		f_j^{(1)}\left(1,\pi\right)-	f_j^{(1)}\left(0,\pi\right)
		\right)
		\right). 
	\end{align*}
	 In \tb{Step 2}, we prove that $E\bar{\delta}=\tau+o(n^{-\frac{1}{2}})$.  In \tb{Step 3}, we show that $\hat \tau_{np}-\bar{\delta}=o_p(n^{-\frac{1}{2}})$. Finally, in \tb{Step 4}, we establish the asymptotic normality of $	\sqrt{n}(\bar{\delta}-E\bar{\delta})$. Then the proof is complete by combining these four steps. 
	
	\tb{Step 1.}  
	Let \begin{align}
		\label{pf6:eqn1}
		\bar \delta=\frac{1}{n}\sum_{i=1}^n \left(
		\frac{\frac{1}{n \tilde h^p}\sum_{j=1}^n K\left(\frac{\mathbf{z}_i-\mathbf{z}_j}{\tilde h}\right)Y_j\left(\frac{W_j}{\pi}\right)}{\hat p_1(\mathbf{z}_i)
		}-
		\frac{\frac{1}{n \tilde h^p}\sum_{j=1}^n K\left(\frac{\mathbf{z}_i-\mathbf{z}_j}{\tilde h}\right)Y_j\left(\frac{1-W_j}{1-\pi}\right)}{
			\hat p_2(\mathbf{z}_i)
		}
		\right)
		I\left(	p(\mathbf{z}_i)>\tilde b\right).
	\end{align}
	Notice that \[\frac{1}{\hat p_k(\mathbf{z}_i)}=\frac{2}{p(\mathbf{z}_i)}-\frac{\hat p_k(\mathbf{z}_i)}{p^2(\mathbf{z}_i)}+\frac{(\hat p_k(\mathbf{z}_i)-p(\mathbf{z}_i))^2}{\hat p_k(\mathbf{z}_i)p^2(\mathbf{z}_i)}\] for $k=1,2$. Accordingly, we decompose $\bar \delta$ as:
	\begin{align}
		\label{eqn:21}
		\bar\delta=\delta_1-\delta_2+\delta_3,
	\end{align}
	where
	\begin{align}
		\label{pf6:def:delta1}
		\delta_1&=\frac{2}{n^2 \tilde h^p}\sum_{i=1}^n\sum_{j=1}^n \frac{1}{p(\mathbf{z}_i)} K\left(\frac{\mathbf{z}_i-\mathbf{z}_j}{\tilde h}\right)Y_j\left(\frac{W_j}{\pi}-\frac{1-W_j}{1-\pi}\right)	I\left(	p(\mathbf{z}_i)>\tilde b\right),\\
			\label{pf6:def:delta2}
		\delta_2&=\frac{1}{n^2 \tilde h^p}\sum_{i=1}^n\sum_{j=1}^n \frac{1}{p^2(\mathbf{z}_i)} K\left(\frac{\mathbf{z}_i-\mathbf{z}_j}{\tilde h}\right)Y_j\left(\hat p_1(\mathbf{z}_i)\frac{W_j}{\pi}-\hat p_2(\mathbf{z}_i)\frac{1-W_j}{1-\pi}\right)	I\left(	p(\mathbf{z}_i)>\tilde b\right),\\
			\label{pf6:def:delta3}
		\delta_3&=\frac{1}{n^2 \tilde h^p}\sum_{i=1}^n \left(
		\frac{(\hat p_1(\mathbf{z}_i)-p(\mathbf{z}_i))^2}{\hat p_1(\mathbf{z}_i)p^2(\mathbf{z}_i)}
		\sum_{j=1}^n K\left(\frac{\mathbf{z}_i-\mathbf{z}_j}{\tilde h}\right)Y_j\frac{W_j}{\pi}\right.\notag\\ 
		&\left.-
		\frac{(\hat p_2(\mathbf{z}_i)-p(\mathbf{z}_i))^2}{\hat p_2(\mathbf{z}_i)p^2(\mathbf{z}_i)}\sum_{j=1}^n K\left(\frac{\mathbf{z}_i-\mathbf{z}_j}{\tilde h}\right)Y_j\frac{1-W_j}{1-\pi}
		\right)
		I\left(	p(\mathbf{z}_i)>\tilde b\right).
	\end{align}
	
	By Lemma~\ref{lem4th6:2}, we have that 
	\begin{align}
		\label{pf6:delta1}
		\sqrt{n}(\delta_1-E\delta_1)&=\frac{2}{\sqrt{n}}\sum_{i=1}^n\left(
	f_i\left(W_i,\pi\right)\left(\frac{W_i}{\pi}-\frac{1-W_i}{1-\pi}\right)+E\left(f_i\left(W_i,\pi\right)\left(\frac{W_i}{\pi}-\frac{1-W_i}{1-\pi}\right)\bigg|\mathbf{z}_i\right)\right.\notag\\
	&\left.-2\tau+
	(W_i-\pi)
	\sum_{j:j\neq i}\frac{E_{ij}}{N_j}
	E\left(
	f_j^{(1)}\left(W_j,\pi\right)\left(\frac{W_j}{\pi}-\frac{1-W_j}{1-\pi}\right)
	\right)
	\right)+
	o_p(1).
\end{align}
By Lemma~\ref{lem4th6:3}, we have
\begin{align}
	\label{pf6:delta2}
 \sqrt{n}(\delta_{2}-E\delta_{2})
		=&
		\frac{1}{\sqrt{n}}\sum_{i=1}^n \left(	E(f_i(1,\pi)-f_i(0,\pi)|\mathbf{z}_i)+
		f_i(W_i,\pi)\left(\frac{W_i}{\pi}-\frac{1-W_i}{1-\pi}\right)	\right.\notag\\+&\left.
		E(f_i(1,\pi)|\mathbf{z}_i)
		\frac{W_i}{\pi}-	E(f_i(0,\pi)|\mathbf{z}_i)\frac{1-W_i}{1-\pi}-3\tau
		\right.\notag\\+&\left.
		(W_i-\pi)
		\sum_{j:j\neq i}\frac{E_{ij}}{N_j}
		E\left(
		f_j^{(1)}\left(W_j,\pi\right)
		\left(\frac{W_j}{\pi}-\frac{(1-W_j)}{(1-\pi)}\right)
		\right)
		\right)+o_p(1).
\end{align}
By Lemma~\ref{lem4th6:4}, we have that \begin{align}
	\label{pf6:delta3}
	\delta_3=o_p(n^{-\frac{1}{2}}).
\end{align}
Then by combining (\ref{pf6:delta1}), (\ref{pf6:delta2}) and (\ref{pf6:delta3}), we have
\begin{align}
	\label{eqn:251}
		\sqrt{n}(\bar{\delta}-E\bar{\delta})&=
		\frac{1}{\sqrt{n}}\sum_{i=1}^n\left(
		f_i(W_i,\pi)\left(\frac{W_i}{\pi}-\frac{1-W_i}{1-\pi}\right)+
		E(f_i(1,\pi)-f_i(0,\pi)|\mathbf{z}_i)
		\right.\notag\\&-\left.
		\left(
		E(f_i(1,\pi)|\mathbf{z}_i)
		\frac{W_i}{\pi}-
		E(f_i(0,\pi)|\mathbf{z}_i)\frac{1-W_i}{1-\pi}
		\right)-\tau
		\right.\notag\\&+\left.
		(W_i-\pi)
		\sum_{j:j\neq i}\frac{E_{ij}}{N_j}
		E\left(
		f_j^{(1)}\left(1,\pi\right)-	f_j^{(1)}\left(0,\pi\right)
		\right)
		\right)+o_p(1). 
\end{align}

\tb{Step 2.} In this step, we show that   $E\bar{\delta}=\tau+o(n^{-\frac{1}{2}})$. Recall the decomposition that $	\bar\delta=\delta_1-\delta_2+\delta_3$ in (\ref{eqn:21}). Furthermore, in the proof of  Lemma~\ref{lem4th6:4}, we show that $E\bar{\delta_3}=o(n^{-\frac{1}{2}})$. Therefore, it suffices to show that $E(\delta_1-\delta_2)=\tau+o(n^{-\frac{1}{2}})$. 

For $\delta_1$, we decompose it as $\delta_1=\delta_{11}+\delta_{12}+\delta_{13}$ in the proof of  Lemma~\ref{lem4th6:2}. Moreover, we show that $E\delta_{13}=o(n^{-\frac{1}{2}})$  in the proof of  Lemma~\ref{lem4th6:2}.  {We prove  that $E\delta_{11}=2\tau+o(n^{-\frac{1}{2}})$ in Lemma~\ref{lem:new1}, and $E\delta_{12}=o(n^{-\frac{1}{2}})$  in Lemma~\ref{lem:new3}. } 
Therefore, $E\delta_{1}=2\tau+o(n^{-\frac{1}{2}})$.

For $\delta_2$, we decompose it as $\delta_2=\delta_{21}+\delta_{22}+\delta_{23}$ in the proof of  Lemma~\ref{lem4th6:3}.  Moreover, we show that $E\delta_{23}=o(n^{-\frac{1}{2}})$  in the proof of  Lemma~\ref{lem4th6:3}. {We prove  that $E\delta_{21}=\tau+o(n^{-\frac{1}{2}})$ in Lemma~\ref{lem:new2}, and $E\delta_{22}=o(n^{-\frac{1}{2}})$  in Lemma~\ref{lem:new3}. }  
Therefore, $E\delta_{2}=\tau+o(n^{-\frac{1}{2}})$. 
Then we have $E(\delta_1-\delta_2)=\tau+o(n^{-\frac{1}{2}})$, and consequently, $E\bar{\delta}=\tau+o(n^{-\frac{1}{2}})$. 

\tb{Step 3.} In this step, we   prove that $\hat \tau_{np}-\bar{\delta}=o_p(n^{-\frac{1}{2}})$.  
Let $c_N=c_f\left(\sqrt{\frac{\log(n)}{n \tilde h^p}}+\tilde h^q\right)$ where $c_f$ is a large enough constant. Define $\bar{\delta}_f$ as the estimator trimming with respect to the bound $\tilde b-c_N$, i.e., 
\begin{align*}
\bar \delta_f=\frac{1}{n}\sum_{i=1}^n \left(
\frac{\frac{1}{n \tilde h^p}\sum_{j=1}^n K\left(\frac{\mathbf{z}_i-\mathbf{z}_j}{\tilde h}\right)Y_j\frac{W_j}{\pi}}{\hat p_1(\mathbf{z}_i)
}-
\frac{\frac{1}{n \tilde h^p}\sum_{j=1}^n K\left(\frac{\mathbf{z}_i-\mathbf{z}_j}{\tilde h}\right)Y_j\frac{1-W_j}{1-\pi}}{
\hat p_2(\mathbf{z}_i)
}
\right)
I\left(	p(\mathbf{z}_i)>\tilde b-c_N\right).
\end{align*}
Since $\sqrt{\frac{\log(n)}{n \tilde h^p}}+\tilde h^q=o(\tilde b)$ by Assumption~\ref{ass:main3},  $\bar\delta_f$ also satisfies (\ref{eqn:251}),  and $\bar\delta_f$ and $\bar\delta$ are asymptotically equivalent, i.e., $\sqrt{n}(\bar\delta_f-\bar\delta)=o_p(1)$. 
We note that by Lemma~\ref{lem4th6:1}, the fact that $\sqrt{\frac{\log(n)}{n \tilde h^p}}+\tilde h^q=o(\tilde b)$ and that $c_f$ is large enough, we have  
when $n$ is large enough, 
\begin{align*}
&	I(p(\mathbf{z}_i)>\tilde b-c_N)-I(\hat p(\mathbf{z}_i)>1.01\tilde b,
	\tilde p_1(\mathbf{z}_i)>\tilde b,\tilde p_2(\mathbf{z}_i)>\tilde b)
	\\&=I(p(\mathbf{z}_i)>\tilde b-c_N, 
	\hat p(\mathbf{z}_i)<1.01\tilde b
	)\\&=
	I(p(\mathbf{z}_i)>\tilde b-c_N, 
	\hat p(\mathbf{z}_i)<1.01\tilde b, p(\mathbf{z}_i)<1.01\tilde b+c_N
	).
\end{align*} Therefore, when $n$ is large enough, 
\begin{align}
	\label{eqn:252}
\sqrt{n}(\bar{\delta}_f-\hat \tau_{np})&= 
\frac{1}{\sqrt n}\sum_{i=1}^n \left(
\frac{\frac{1}{n \tilde h^p}\sum_{j=1}^n K\left(\frac{\mathbf{z}_i-\mathbf{z}_j}{\tilde h}\right)Y_j\frac{W_j}{\pi}}{\hat p_1(\mathbf{z}_i)
}-
\frac{\frac{1}{n \tilde h^p}\sum_{j=1}^n K\left(\frac{\mathbf{z}_i-\mathbf{z}_j}{\tilde h}\right)Y_j\frac{1-W_j}{1-\pi}}{
\hat p_2(\mathbf{z}_i)
}
\right)\notag\\&
	I(p(\mathbf{z}_i)>\tilde b-c_N, 
\hat p(\mathbf{z}_i)<1.01\tilde b, p(\mathbf{z}_i)<1.01\tilde b+c_N
).
\end{align} 
For the right-hand-side of (\ref{eqn:252}), we follow \tb{Step 1} and \tb{Step 2} of this proof to show that $\sqrt{n}(\bar{\delta}_f-\hat \tau_{np})=o_p(1)$. We specifically outline the key differences below:
\begin{enumerate}
	\item $	I(p(\mathbf{z}_i)>\tilde b-c_N, 
	\hat p(\mathbf{z}_i)<1.01\tilde b, p(\mathbf{z}_i)<1.01\tilde b+c_N
	)/p(\mathbf{z}_i)\le 1/(\tilde b-c_N)\lesssim 1/\tilde b$.
	\item $I(p(\mathbf{z}_i)>\tilde b-c_N, 
	\hat p(\mathbf{z}_i)<1.01\tilde b, p(\mathbf{z}_i)<1.01\tilde b+c_N
	)\in\{0,1\}$, which is bounded. As a result, the dominated covergence theorems in Lemma~\ref{lem4th6:2}, \ref{lem4th6:6}, \ref{lem4th6:8}, \ref{lem4th6:9}   are still applicable. 
	\item $I(p(\mathbf{z}_i)>\tilde b-c_N, 
	\hat p(\mathbf{z}_i)<1.01\tilde b, p(\mathbf{z}_i)<1.01\tilde b+c_N
	)\to 0$ as $n\to\infty$. As a result,  the dominated covergence theorems in Lemma~\ref{lem4th6:2}, \ref{lem4th6:6}, \ref{lem4th6:8}, \ref{lem4th6:9}  
	lead to the desired result. 
	\item  {Assumption~\ref{ass:main3}(6) which helps to show that the expectation of $\sqrt{n}(\bar{\delta}_f-\hat \tau_{np})$ is $o(1)$.}
\end{enumerate}
 Therefore, we have $\sqrt{n}(\bar{\delta}-\hat \tau_{np})=o_p(1)$.

\tb{Step 4.} We now prove that (\ref{eqn:251}) is asymptotically normal. Note that in the summation  terms, 
$f_i(W_i,\pi)\left(\frac{W_i}{\pi}-\frac{1-W_i}{1-\pi}\right)+
	E(f_i(1,\pi)-f_i(0,\pi)|\mathbf{z}_i)
-
	\left(
	E(f_i(1,\pi)|\mathbf{z}_i)
	\frac{W_i}{\pi}-
	E(f_i(0,\pi)|\mathbf{z}_i)\frac{1-W_i}{1-\pi}
	\right)-\tau$ are i.i.d., and $(W_i-\pi)
	\sum_{j\neq i}\frac{E_{ij}}{N_j}
	E\left(
	f_j^{(1)}\left(1,\pi\right)-	f_j^{(1)}\left(0,\pi\right)
	\right)$ are non-i.i.d. due to the random graph setting.  
By following the proof of Theorem 4 in \citet{li2022random}, we obtain that  
\begin{align*}
	\frac{1}{\sqrt{n}}&\sum_{i=1}^n\left(
	f_i(W_i,\pi)\left(\frac{W_i}{\pi}-\frac{1-W_i}{1-\pi}\right)+
	E(f_i(1,\pi)-f_i(0,\pi)|\mathbf{z}_i)
	\right.\notag\\&-\left.
	\left(
	E(f_i(1,\pi)|\mathbf{z}_i)
	\frac{W_i}{\pi}-
	E(f_i(0,\pi)|\mathbf{z}_i)\frac{1-W_i}{1-\pi}
	\right)-\tau
	\right.\notag\\&+\left.
	(W_i-\pi)
	\sum_{j\neq i}\frac{E_{ij}}{N_j}
	E\left(
	f_j^{(1)}\left(1,\pi\right)-	f_j^{(1)}\left(0,\pi\right)
	\right)
	\right)\notag\\&=
\frac{1}{\sqrt n}\sum_{i=1}^n\left(
R_i-\tau+Q_i(W_i-\pi)
\right)+o_p(1),
\end{align*}
where
\begin{align*}
R_i&=(W_i-\pi)\left(
\frac{f_i(W_i,\pi)}{\pi(1-\pi)}
-\left(
\frac{E(f_i(1,\pi)|\mathbf{z}_i)}{\pi}+
\frac{E(f_i(0,\pi)|\mathbf{z}_i)}{1-\pi}
\right)
\right),\\
Q_i&=E\left(\frac{h(U_i,U_j)}{E(h(U_i,U_j)|U_j)}\bigg|U_i\right)
E\left(
f_1^{(1)}\left(1,\pi\right)-	f_1^{(1)}\left(0,\pi\right)
\right),
\end{align*}
and $h(\cdot)$ is the graphon defined in Assumption~\ref{assump:random_graph}. Note that $R_i-\tau+Q_i(W_i-\pi)$ are i.i.d., thus, the asymptotic distribution is $N(0,V_{np})$ with
\begin{align*}
V_{np}&=E\left(R_i-\tau+Q_i(W_i-\pi)\right)^2\\
&=E(R_i-\tau)^2\\&+
E\left(E\left(\frac{h(U_i,U_j)}{E(h(U_i,U_j)|U_j)}\bigg|U_i\right)\right)^2
\pi(1-\pi)(E f_1^{(1)} (1,\pi)-E f_1^{(1)} (0,\pi))^2\\
&=Var(f_i(1,\pi)-f_i(0,\pi))\\
&+\frac{1}{\pi(1-\pi)}E\left[
(1-\pi)\left(f_i(1,\pi)-E(f_i(1,\pi)|\mathbf{z}_i)
\right)+
\pi\left(f_i(0,\pi)-E(f_i(0,\pi)|\mathbf{z}_i)
\right)
\right]^2\\
&+	E\left(E\left(\frac{h(U_i,U_j)}{E(h(U_i,U_j)|U_j)}\bigg|U_i\right)\right)^2\pi(1-\pi)(E f_1^{(1)} (1,\pi)-E f_1^{(1)} (0,\pi))^2.
\end{align*}
Then the proof is complete. 
\end{proof}

\section*{Lemmas}
\begin{lemma}
    \label{lem:1}
    Under Assumption~\ref{assumption:anonymous interference}, \ref{assump:random_graph}, \ref{assumption:moment_conditions} and  \ref{ass:main2}, we have $\hat b\overset{p}{\to}b$ where $\hat b$ is defined in (\ref{def:hat b}) and $b$ is defined in (\ref{def:b}).
\end{lemma}
\begin{proof}
Recall that 
\[\hat b=\frac{1}{n}\sum_{i=1}^n \left(
 \sum_{j=1}^n \frac{E_{ij}}{\sum_{k=1}^n E_{jk}}
\right)^2, 
b=\int_0^1 \left(
\int_0^1 \frac{  h(x,y)}{\int_0^1   h(x,z)dz}dx
 \right)^2dy.\]

Let $N_i=\sum_{j=1}^n E_{ij}$ be the degree of the $i$-th node. By taking $s(x)=1$ in Lemma~\ref{lem:11}, we have
\begin{align}
    \label{lem1:eqn1}
\sup_{i=1,\cdots,n}\left|
\frac{N_i}{n\rho_n} - \int_0^1  {h(U_i,y)} dy
        \right|=O_p\left(\sqrt{\frac{\log(n)}{n\rho_n}}\right).
\end{align}
By taking $s(x)=\int_0^1 h(x,z)dz$ in Lemma~\ref{lem:11}, we have
\begin{align}
    \label{lem1:eqn2}
\sup_{i=1,\cdots,n}\frac{1}{n\rho_n}\left|\sum_{j=1}^n \frac{E_{ij}}{\int_0^1 h(U_j,z)dz} -n\rho_n\int_0^1  \frac{h(U_i,y)}{\int_0^1 h(y,z)dz} dy \right|=O_p\left(\sqrt{\frac{\log(n)}{n\rho_n}}\right).
\end{align}

Then by (\ref{lem1:eqn1}),(\ref{lem1:eqn2}) and the fact that $\min_i N_i\gtrsim_p n\rho_n, \max_i N_i\lesssim_p n\rho_n$ from Lemma 15 in \cite{li2022random}, we have 
\begin{align}
\label{lem1:eqn3}
   & \sup_i\left|\sum_{j=1}^n \frac{E_{ij}}{N_j}-\int_0^1 \frac{h(x,U_i)}{\int_0^1 h(x,z)dz}dx\right|\notag\\
   &\le \sup_i\left|\sum_{j=1}^n E_{ij} \left(\frac{1}{N_j}-\frac{1}{n\rho_n\int_0^1 h(U_j,z)dz}\right)\right|
   +
   \sup_i\left|\sum_{j=1}^n \frac{E_{ij} }{n\rho_n\int_0^1 h(U_j,z)dz}-
   \int_0^1 \frac{h(x,U_i)}{\int_0^1 h(x,z)dz}dx
   \right|\notag\\
   &=O_p \left(n\rho_n\left(\frac{\sqrt{n\rho_n\log(n)}}{\tilde c_l(n\rho_n)^2}\right)\right)
   +
   O_p \left(\sqrt{\frac{\log(n)}{n\rho_n}}\right)= O_p \left(\sqrt{\frac{\log(n)}{n\rho_n}}\right).
\end{align}
Note that \begin{align}
    \label{lem1:eqn4}
    \sup_i\left|\int_0^1 \frac{h(x,U_i)}{\int_0^1 h(x,z)dz}dx\right|\le 
    \frac{\tilde c_u}{\tilde c_l}<\infty.
\end{align}
Therefore, by combining (\ref{lem1:eqn3}) and (\ref{lem1:eqn4}), we have
\[
\sup_i\left|\left(\sum_{j=1}^n \frac{E_{ij}}{N_j}\right)^2-\left(\int_0^1 \frac{h(x,U_i)}{\int_0^1 h(x,z)dz}dx\right)^2\right|=
O_p \left(\sqrt{\frac{\log(n)}{n\rho_n}} + {\frac{\log(n)}{n\rho_n}}\right).
\]
As a result, 
\begin{align}
	\label{lem1:aeqn:9}
\left|\frac{1}{n}\sum_{i=1}^n\left(\sum_{j=1}^n \frac{E_{ij}}{N_j}\right)^2-
\frac{1}{n}\sum_{i=1}^n\left(\int_0^1 \frac{h(x,U_i)}{\int_0^1 h(x,z)dz}dx\right)^2\right|=
O_p \left(\sqrt{\frac{\log(n)}{n\rho_n}} + {\frac{\log(n)}{n\rho_n}}\right).
\end{align}
By the law of large numbers, we have 
\begin{align}
	\label{lem1:aeqn:10}
	\left|\frac{1}{n}\sum_{i=1}^n\left(\int_0^1 \frac{h(x,U_i)}{\int_0^1 h(x,z)dz}dx\right)^2
	-
	\int_0^1\left(\int_0^1 \frac{h(x,y)}{\int_0^1 h(x,z)dz}dx\right)^2dy
	\right|=o_p(1).
\end{align}
Then the result follows by combining (\ref{lem1:aeqn:9}), (\ref{lem1:aeqn:10}) and  the assumption that $\sqrt{n}\rho_n\to\infty$.

\end{proof}

\begin{lemma}
	\label{lem:2}
	 Under Assumption~\ref{assumption:anonymous interference}, \ref{assump:random_graph}, \ref{assumption:moment_conditions} and  \ref{ass:main2}, we have that 
	\begin{align*}
		&\frac{\sum_{i=1}^n W_i(Y_i-\mathbf{x}_i^\top \hat\beta_1)^2}{\sum_{i=1}^n W_i}\overset{p}{\to} E(f_1(1,\pi)-\mathbf{x}_1^\top\beta_1)^2,\\
		&\frac{\sum_{i=1}^n (1-W_i)(Y_i-\mathbf{x}_i^\top \hat\beta_0)^2}{\sum_{i=1}^n (1-W_i)}\overset{p}{\to} E(f_1(0,\pi)-\mathbf{x}_1^\top\beta_0)^2.
	\end{align*}
\end{lemma}
\begin{proof}
	We only prove \[\frac{\sum_{i=1}^n W_i(Y_i-\mathbf{x}_i^\top \hat\beta_1)^2}{\sum_{i=1}^n W_i}\overset{p}{\to} E(f_1(1,\pi)-\mathbf{x}_1^\top\beta_1)^2,\] as the second one follows similarly.
	
	Since $W_i(f_i(1,\pi)-\mathbf{x}_i^\top\beta_1)^2, i=1,\cdots,n$ are i.i.d., and $E(f_1(1,\pi)-\mathbf{x}_1^\top\beta_1)^2<\infty$ by Assumption~\ref{assumption:moment_conditions}, we have \[
	\frac{1}{n}\sum_{i=1}^n W_i(f_i(1,\pi)-\mathbf{x}_i^\top\beta_1)^2=\pi E(f_1(1,\pi)-\mathbf{x}_1^\top\beta_1)^2+o_p(1),
	\]
	it suffices to show that 
	\begin{align}
		\label{lem2:eqn1}
		\frac{1}{n}\sum_{i=1}^n W_i
		\left[
		(f_i\left(1,\mni\right)-\mathbf{x}_i^\top\hat\beta_1)^2
		-(f_i(1,\pi)-\mathbf{x}_i^\top\beta_1)^2
		\right]=o_p(1).
		\end{align}
		Note that by Taylor expansion,  Assumption~\ref{assumption:anonymous interference}, \ref{assump:random_graph}, \ref{assumption:moment_conditions}, and $\min_i N_i\gtrsim_p n\rho_n$, we have
		\begin{align*}
		\left|
			f_i\left(1,\mni\right)-\mathbf{x}_i^\top\hat\beta_1-
			f_i(1,\pi)+\mathbf{x}_i^\top\beta_1 	\right|
		&	\le \left|\mni-\pi\right| \sup_{u\in [0,1]}|f_i^{(1)}(1,u)|+\mathbf{x}_i^\top|\hat\beta_1-\beta_1|,\\
		\sup_i \left|\mni-\pi\right| \sup_{u\in [0,1]}|f_i^{(1)}(1,u)|&=o_p(1).
		\end{align*}
		As a result,
		\begin{align*}
				&\frac{1}{n}\sum_{i=1}^n W_i
			\left[
			(f_i\left(1,\mni\right)-\mathbf{x}_i^\top\hat\beta_1)^2
			-(f_i(1,\pi)-\mathbf{x}_i^\top\beta_1)^2
			\right]\\
			&\le \frac{1}{n}\sum_{i=1}^n W_i
			\left[
			\left|
			f_i\left(1,\mni\right)-\mathbf{x}_i^\top\hat\beta_1-
			f_i(1,\pi)+\mathbf{x}_i^\top\beta_1 	\right|^2\right.\\
			&\left.+
			2	\left|
			f_i(1,\pi)-\mathbf{x}_i^\top\beta_1 	\right|
			\left|
			f_i\left(1,\mni\right)-\mathbf{x}_i^\top\hat\beta_1-
			f_i(1,\pi)+\mathbf{x}_i^\top\beta_1 	\right|
			\right]\\
			&\le \frac{1}{n}\sum_{i=1}^n W_i
			\left[
	\left( \sup_i	\left|\mni-\pi\right| \sup_{u\in [0,1]}|f_i^{(1)}(1,u)|+\mathbf{x}_i^\top|\hat\beta_1-\beta_1|\right)^2\right.\\
			&\left.+
			2	\left|
			f_i(1,\pi)-\mathbf{x}_i^\top\beta_1 	\right|
		\left(\sup_i\left|\mni-\pi\right| \sup_{u\in [0,1]}|f_i^{(1)}(1,u)|+\mathbf{x}_i^\top|\hat\beta_1-\beta_1|\right)
			\right]\\
			&=\frac{1}{n}\sum_{i=1}^n W_i
			|\hat\beta_1-\beta_1|^\top \mathbf{x}_i\mathbf{x}_i^\top|\hat\beta_1-\beta_1|+o_p(1)\\
			&+\frac{1}{n}\sum_{i=1}^n W_i 	\left|
			f_i(1,\pi)-\mathbf{x}_i^\top\beta_1 	\right| (o_p(1)+\mathbf{x}_i^\top|\hat\beta_1-\beta_1|),
		\end{align*}
		where the $o_p(1)$ term above is independent of the index $i$. 
		Then it suffices to show that 
		\begin{align}
			\label{lem2:eqn2}
			\frac{1}{n}\sum_{i=1}^n W_i
			|\hat\beta_1-\beta_1|^\top \mathbf{x}_i\mathbf{x}_i^\top|\hat\beta_1-\beta_1|&=o_p(1),\\
			\label{lem2:eqn3}
			\frac{1}{n}\sum_{i=1}^n W_i 	\left|
			f_i(1,\pi)-\mathbf{x}_i^\top\beta_1 	\right|&=O_p(1),\\
				\label{lem2:eqn4}
			\frac{1}{n}\sum_{i=1}^n W_i 	\left|
			f_i(1,\pi)-\mathbf{x}_i^\top\beta_1 	\right|\mathbf{x}_i^\top|\hat\beta_1-\beta_1|&=o_p(1).
		\end{align}
		For (\ref{lem2:eqn2}), we first  recall that $
		\sum_{i=1}^n W_i\mathbf{x}_i\mathbf{x}_i^\top/n \overset{a.s.}{\to}\pi \bb{M}_{xx}$ by (\ref{pf1:eqn1}), and  $\hat\beta_1=\beta_1+o_p(1)$ by Theorem~\ref{th: Consistency of hat tau}. Then we derive that 
		\begin{align*}
				\frac{1}{n}\sum_{i=1}^n W_i
			|\hat\beta_1-\beta_1|^\top \mathbf{x}_i\mathbf{x}_i^\top|\hat\beta_1-\beta_1|&=trace\left[
				\frac{1}{n}\sum_{i=1}^n W_i
		 \mathbf{x}_i\mathbf{x}_i^\top|\hat\beta_1-\beta_1|	|\hat\beta_1-\beta_1|^\top
			\right]\\
		&	=trace(\pi\bb{M}_{xx}|\hat\beta_1-\beta_1|	|\hat\beta_1-\beta_1|^\top)+o_p(1)=o_p(1),
		\end{align*} 
	where the second equality follows from the continuous mapping theorem, noting that \( \mathbb{M}_{xx} \) is deterministic by definition, while the third equality holds due to the fixed dimension of the covariates \( \mathbf{x}_i \).
		
		For (\ref{lem2:eqn3}), noting that \( \mathbb{E} \left| f_i(1,\pi) - \mathbf{x}_i^\top \beta_1 \right| < \infty \), the result follows directly from the law of large numbers.
		
			For  (\ref{lem2:eqn4}), since $\hat\beta_1=\beta_1+o_p(1)$, it suffices to show that $ \sum_{i=1}^n W_i 	\left|
			f_i(1,\pi)-\mathbf{x}_i^\top\beta_1 	\right|\mathbf{x}_i /n=O_p(1)$, which follows from the law of large numbers and the fact that $E\left|
			f_i(1,\pi)-\mathbf{x}_i^\top\beta_1 	\right||\mathbf{x}_i|<\infty$. Thus, the proof is complete. 
\end{proof}

 \begin{lemma}
 	\label{lem:3}
 	 Under Assumption~\ref{assumption:anonymous interference}, \ref{assump:random_graph}, \ref{assumption:moment_conditions} and  \ref{ass:main2}, we have that $	\hat{E}f_1^{(1)}(1,\pi)\overset{p}{\to} Ef_1^{(1)}(1,\pi)$ and $	\hat{E}f_1^{(1)}(0,\pi)\overset{p}{\to} Ef_1^{(1)}(0,\pi)$	where $\hat{E}f_1^{(1)}(1,\pi)$ and $\hat{E}f_1^{(1)}(0,\pi)$ are defined in (\ref{eqn:dev1}) and (\ref{eqn:dev2}) respectively. 
 \end{lemma}
 \begin{proof}
 		We only prove
 		\begin{align}
 			\label{lem3:eqn1}
 			\hat{E}f_1^{(1)}(1,\pi)\overset{p}{\to} Ef_1^{(1)}(1,\pi),
 		\end{align}  as the second one follows similarly. 
 		
 		Let $\check{f}_i(W_i,M_i/N_i)=W_if_i(W_i,M_i/N_i)/\pi$, then  under Assumption~\ref{assumption:anonymous interference}, \ref{assump:random_graph}, \ref{assumption:moment_conditions} and  \ref{ass:main2},  Theorem 6 in \cite{li2022random} still holds (note that our assumption is slightly weaker than theirs). Therefore, 
 		\begin{align*}
 			\hat{E}f_1^{(1)}(1,\pi)\overset{p}{\to}E(\pi \check{f}_i^{(1)}(1,\pi)+(1-\pi)\check{f}_i^{(1)}(0,\pi))=E(f_1^{(1)}(1,\pi)).
 		\end{align*}
 	
 \end{proof}
 
\begin{lemma}
    \label{lem:11}
   Under Assumption~\ref{assumption:anonymous interference}, \ref{assump:random_graph}, \ref{assumption:moment_conditions} and  \ref{ass:main2}, for any function $s(x),x\in[0,1],$ satisfying $\inf_{x\in[0,1]}s(x)\ge \tilde c_l>0$, we have
    \begin{align*}
         \sup_{i=1,\cdots,n}\frac{1}{\rho_n(n-1)}\left|
\sum_{j:j\neq i}\frac{E_{ij}}{s(U_j)}-n\rho_n\int_0^1 \frac{h(U_i,y)}{s(y)}dy
        \right|=O_p\left(\sqrt{\frac{\log(n)}{n\rho_n}}\right).
    \end{align*}
\end{lemma}
\begin{proof}
    Since $E(E_{ij}|U_i,U_j)=\rho_nh(U_i,U_j)=\rho_n \sum_{k=1}^r \lambda_k\psi_k(U_i)\psi_k(U_j)$, we write that 
    \begin{align*}
       &\frac{1}{\rho_n(n-1)} \left|
\sum_{j:j\neq i}\frac{E_{ij}}{s(U_j)}-n\rho_n\int_0^1 \frac{h(U_i,y)}{s(y)}dy
        \right|\\
        &\le 
        \frac{1}{n-1}
 \left|
\sum_{j:j\neq i}\sum_{k=1}^r \lambda_k\psi_k(U_i)\frac{\psi_k(U_j)}{s(U_j)}-n\sum_{k=1}^r \lambda_k\psi_k(U_i)\int_0^1 \frac{\psi_k(y)}{s(y)}dy
        \right|\\
        &+
        \frac{1}{\rho_n(n-1)} \left|
\sum_{j:j\neq i}\frac{E_{ij}}{s(U_j)}-\rho_n\sum_{j:j\neq i}\sum_{k=1}^r \lambda_k\psi_k(U_i)\frac{\psi_k(U_j)}{s(U_j)}
        \right|=\mathbf{I}_i + \mathbf{II}_i.
    \end{align*}
In \tb{Step 1}, we prove that $\sup_i\mathbf{I}_i=O_p(\sqrt{1/n})$. In \tb{Step 2}, we prove that $\sup_i\mathbf{II}_i=O_p(\sqrt{\log(n)/(n\rho_n)})$. 

\tb{Step 1.} Note that $\sup_{x\in[0,1]}\max_{k=1,\cdots,r}|\psi_k(x)|\le M_1, \inf_{x\in[0,1]}s(x)\ge \tilde c_l>0$. From these conditions, we derive that
\begin{align*}
    \sup_{i=1,\cdots,n}\mathbf{I}_i&\le 
     \sup_{i=1,\cdots,n}\frac{M_1}{n-1}
\sum_{k=1}^r \lambda_k \left|
\sum_{j:j\neq i}  \frac{\psi_k(U_j)}{s(U_j)}-  n\int_0^1 \frac{\psi_k(y)}{s(y)}dy
        \right|\\
        &\le 
         \frac{M_1}{n-1}
\sum_{k=1}^r \lambda_k \left|
\sum_{j=1}^n  \frac{\psi_k(U_j)}{s(U_j)}-  n\int_0^1 \frac{\psi_k(y)}{s(y)}dy
        \right|+
        \sup_{i=1,\cdots,n}\frac{M_1}{n-1}
\sum_{k=1}^r \lambda_k \left|
\frac{\psi_k(U_i)}{s(U_i)}
        \right|\\
        &\le  \frac{M_1}{n-1}
\sum_{k=1}^r \lambda_k \left|
\sum_{j=1}^n  \frac{\psi_k(U_j)}{s(U_j)}-  n\int_0^1 \frac{\psi_k(y)}{s(y)}dy
        \right|+
         \frac{M_1}{n-1}
\sum_{k=1}^r \lambda_k \frac{M_1}{\tilde c_l}.
\end{align*}
Since 
\[var\left(\frac{\psi_k(U_j)}{s(U_j)}\right)\le E\left(\frac{\psi_k(U_j)}{s(U_j)}\right)^2\le \frac{M_1^2}{\tilde c_l^2},\]
$\{U_j\}$ are i.i.d. and $r$ is bounded, we have 
\[\frac{1}{n-1} \left|
\sum_{j=1}^n  \frac{\psi_k(U_j)}{s(U_j)}- n \int_0^1 \frac{\psi_k(y)}{s(y)}dy
        \right|=O_p(n^{-1/2}).\]
Therefore, $\sup_i\mathbf{I}_i=O_p(\sqrt{1/n})$.

\tb{Step 2.} It suffices to show that
\begin{align*}\label{pfeqn:01}
        \sup_{i=1\cdots,n}&\left|
        \frac{1}{\rho_n(n-1)}\sum_{j:j\neq i}\left(\frac{I\left(U_{ij}\le \rho_n\sum_{k=1}^r\lambda_k\psi_k(U_i)\psi_k(U_j) 
        \right)}{s(U_j)}
        -\rho_n
\sum_{k=1}^r\lambda_k\psi_k(U_i)\frac{\psi_k(U_j) }{s(U_j)}
        \right)
        \right|\\&=O_p(\sqrt{\log(n)/(n\rho_n)}),
    \end{align*}
    where $(U_{ij},i\le j)$ are i.i.d. uniformly distributed random variables on $[0,1]$, and $U_{ji}=U_{ij}$ for $i>j$.  
    Let \begin{align*}
\Psi_i= & 
        \frac{1}{\rho_n(n-1)}\sum_{j:j\neq i}\left(\frac{I\left(U_{ij}\le \rho_n\sum_{k=1}^r\lambda_k\psi_k(U_i)\psi_k(U_j) 
        \right)}{s(U_j)}
        -\rho_n
\sum_{k=1}^r\lambda_k\psi_k(U_i)\frac{\psi_k(U_j) }{s(U_j)}
        \right)
         .
 \end{align*}
 By Bernstein's inequality for bounded variables, we have for any $t>0$, 
\[
    P\left(
    \sqrt{n\rho_n}|\Psi_i|>t|U_1,\cdots,U_n
    \right)\le 2\exp\left(-\frac{ct^2}{1+t/\sqrt{n\rho_n}}\right),
\]
where $c>0$ is an absolute constant. Then \[ P\left(
    \sqrt{n\rho_n}|\Psi_i|>t\right)= {E}\left(  {P}\left(
    \sqrt{n\rho_n}|\Psi_i|>t|U_1,\cdots,U_n
    \right)\right)\le 2\exp\left(-\frac{ct^2}{1+t/\sqrt{n\rho_n}}\right).\] 
  By the assumption that $\sqrt n\rho_n\to\infty$, when $n$ is large enough, we have $(1+\log(n))/(n\rho_n)\le 4$. Then we derive that when $n$ is large enough, 
    \begin{align*}
        E\max_{i=1,\cdots,n} \frac{\sqrt{n\rho_n}|\Psi_i|}{\sqrt{1+\log(i)}}&\le t_0+\int_{t_0}^\infty P\left(\max_{i=1,\cdots,n} \frac{\sqrt{n\rho_n}|\Psi_i|}{\sqrt{1+\log(i)}}\ge t\right)dt\\
        &\le t_0+\int_{t_0}^\infty \sum_{i=1}^n P\left( \frac{\sqrt{n\rho_n}|\Psi_i|}{\sqrt{1+\log(i)}}\ge t\right)dt\\
         &\le t_0+\int_{t_0}^\infty \sum_{i=1}^n 
         2\exp\left(-\frac{c(1+\log(i))t^2}{1+\sqrt{1+\log(i)}\frac{t}{\sqrt{n\rho_n}}}\right)
         dt\\
           &\le t_0+
           \int_{t_0}^\infty \sum_{i=1}^n 
         2\exp\left(-\frac{ct^2}{1+\sqrt{1+\log(i)}\frac{t}{\sqrt{n\rho_n}}}\right)
        i^{-\frac{ct^2}{1+2t}}
         dt\\
          &\le t_0+
           \int_{t_0}^\infty \sum_{i=1}^n 
         2\exp\left(-\frac{ct^2}{1+2t}\right)
        i^{-4}
         dt\\
         &\le t_0+
           \int_{t_0}^\infty 
         2\exp\left(-\frac{ct^2}{1+2t}\right)dt
        \sum_{i=1}^\infty i^{-4}
         <C,
    \end{align*}
    where $t_0>0$ satisfies $ct_0^2=4(1+2t_0)$, and $C$ depends on $c$ only. 
    As a result, \[\bb{E}\max_{i=1\cdots,n}|\Psi_i|=O(\sqrt{\log(n)}/\sqrt{n\rho_n}),\] which indicates that $\max_{i=1\cdots,n}|\Psi_i|=O_p(\sqrt{\log(n)/(n\rho_n)})$.

\end{proof}

\begin{lemma}
	\label{lem4th6:2}
	Under the assumptions for Theorem~\ref{th:5}, we have 	
	\begin{align*} 
		\sqrt{n}(\delta_1-E\delta_1)&=\frac{2}{\sqrt{n}}\sum_{i=1}^n\left(
		f_i\left(W_i,\pi\right)\left(\frac{W_i}{\pi}-\frac{1-W_i}{1-\pi}\right)+E\left(f_i\left(W_i,\pi\right)\left(\frac{W_i}{\pi}-\frac{1-W_i}{1-\pi}\right)\bigg|\mathbf{z}_i\right)\right.\\
		&\left.-2\tau+
		(W_i-\pi)
		\sum_{j\neq i}\frac{E_{ij}}{N_j}
		E\left(
		f_j^{(1)}\left(W_j,\pi\right)\left(\frac{W_j}{\pi}-\frac{1-W_j}{1-\pi}\right)
		\right)
		\right)+
		o_p(1),
	\end{align*}
	 where $\delta_1$ is defined in (\ref{pf6:def:delta1}).  
\end{lemma}
\begin{proof}
	Recall that 
		\begin{align*} 
		\delta_1&=\frac{2}{n^2 \tilde h^p}\sum_{i=1}^n\sum_{j=1}^n \frac{1}{p(\mathbf{z}_i)} K\left(\frac{\mathbf{z}_i-\mathbf{z}_j}{\tilde h}\right)Y_j\left(\frac{W_j}{\pi}-\frac{1-W_j}{1-\pi}\right)	I\left(	p(\mathbf{z}_i)>\tilde b\right).
	\end{align*}
	By Taylor expansion, we have
		\begin{align*}
		Y_i&=f_i\left(W_i,\frac{M_i}{N_i}\right)
		=f_i\left(W_i,\pi\right)+f_i^{(1)}\left(W_i,\pi\right)\left(\frac{M_i}{N_i}-\pi\right)+\frac{1}{2}f_i^{(2)}\left(W_i,\pi_i^*\right)\left(\frac{M_i}{N_i}-\pi\right)^2,
	\end{align*}
	where $\pi_i^*\in [\frac{M_i}{N_i}\wedge\pi,\frac{M_i}{N_i}\vee \pi]$. Therefore, we decompose $ \delta_1$ as
	\begin{align*} 
		\delta_1=\delta_{11}+\delta_{12}+\delta_{13}
	\end{align*}
	where
	\begin{align} 
		\delta_{11}&=\frac{2}{n^2 \tilde h^p}\sum_{i=1}^n\sum_{j=1}^n \frac{1}{p(\mathbf{z}_i)} K\left(\frac{\mathbf{z}_i-\mathbf{z}_j}{\tilde h}\right)f_j\left(W_j,\pi\right)\left(\frac{W_j}{\pi}-\frac{1-W_j}{1-\pi}\right)	I\left(	p(\mathbf{z}_i)>\tilde b\right),\notag\\
			 	\delta_{12}&=\frac{2}{n^2 \tilde h^p}\sum_{i=1}^n\sum_{j=1}^n \frac{1}{p(\mathbf{z}_i)} K\left(\frac{\mathbf{z}_i-\mathbf{z}_j}{\tilde h}\right)f_j^{(1)}\left(W_j,\pi\right)\left(\frac{M_j}{N_j}-\pi\right)\left(\frac{W_j}{\pi}-\frac{1-W_j}{1-\pi}\right)	I\left(	p(\mathbf{z}_i)>\tilde b\right),\notag\\
			\label{lem4th6:2:delta13}
		\delta_{13}&=\frac{2}{n^2 \tilde h^p}\sum_{i=1}^n\sum_{j=1}^n \frac{1}{p(\mathbf{z}_i)} K\left(\frac{\mathbf{z}_i-\mathbf{z}_j}{\tilde h}\right)\frac{1}{2}f_j^{(2)}\left(W_j,\pi_j^*\right)\left(\frac{M_j}{N_j}-\pi\right)^2\left(\frac{W_j}{\pi}-\frac{1-W_j}{1-\pi}\right)	I\left(	p(\mathbf{z}_i)>\tilde b\right).
	\end{align}
	
	The proof is divided into three parts. In Lemma~\ref{lem4th6:5}, we prove that
	\begin{align}
		\label{lem4th6:delta11} 
			\sqrt{n}(\delta_{11}-E\delta_{11})&=
			\frac{2}{\sqrt{n}}\sum_{i=1}^n\left(
			f_i\left(W_i,\pi\right)\left(\frac{W_i}{\pi}-\frac{1-W_i}{1-\pi}\right)+E\left(f_i\left(1,\pi\right)-f_i\left(0,\pi\right) |\mathbf{z}_i\right)-2\tau
			\right)+o_p(1). 
	\end{align} 
	In Lemma~\ref{lem4th6:6}, we prove that 
	\begin{align}
			\label{lem4th6:delta12} 
		\sqrt{n}(\delta_{12}-E\delta_{12})=	\frac{2}{\sqrt n}\sum_{j=1}^n
			\left(\frac{M_j}{N_j}-\pi\right)
			E\left(
			f_j^{(1)}\left(W_j,\pi\right)\left(\frac{W_j}{\pi}-\frac{1-W_j}{1-\pi}\right)
			\right)+o_p(1).
	\end{align}
	In the following part of this proof, we prove that 
	\begin{align}
		\label{lem4th6:delta13} 
		 \delta_{13}=o_p(n^{-1/2}).
	\end{align}
	Then the results follows by combining (\ref{lem4th6:delta11}), (\ref{lem4th6:delta12}), and (\ref{lem4th6:delta13}).
	
	\tb{Proof of (\ref{lem4th6:delta13}).} Note that \begin{enumerate}[label=(\alph*)]
		\item By Assumption~\ref{ass:main3}, $	\max_{w=0,1}\sup_{\mathbf{z}_1}E(\sup_{y\in [0,1]}|f_1^{(2)}(w,y)|^2|w,\mathbf{z}_1)
		<\infty$. As a result, \[	\max_{w=0,1}\sup_{\mathbf{z}_1}E(\sup_{y\in [0,1]}|f_1^{(2)}(w,y)||w,\mathbf{z}_1)
		<\infty.\]
		\item The terms  $\left(\frac{W_j}{\pi}-\frac{1-W_j}{1-\pi}\right)	I\left(	p(\mathbf{z}_i)>\tilde b\right)$ in (\ref{lem4th6:2:delta13}) are bounded by $1/(\pi(1-\pi))$. 
		\item $\max_j\left(\frac{M_j}{N_j}-\pi\right)^2=O_p(n^{-1}\rho_n^{-1})=o_p(n^{-1/2}), \max_j E\left(\frac{M_j}{N_j}-\pi\right)^2=o(n^{-1/2})$ by Lemma 15 in \cite{li2022random} and Assumption~\ref{assump:random_graph}.
	\end{enumerate}
	Then we derive that
\begin{align*}
E|\delta_{13}|&\lesssim 
E\left|\frac{2}{n^2 \tilde h^p}\sum_{i=1}^n\sum_{j=1}^n \frac{1}{p(\mathbf{z}_i)} K\left(\frac{\mathbf{z}_i-\mathbf{z}_j}{\tilde h}\right)\frac{1}{2}f_j^{(2)}\left(W_j,\pi_j^*\right)\left(\frac{M_j}{N_j}-\pi\right)^2\right|\\
		&\le \frac{2}{n^2 \tilde h^p}\sum_{i=1}^n\sum_{j=1}^n E\left[\frac{1}{p(\mathbf{z}_i)} \left|K\left(\frac{\mathbf{z}_i-\mathbf{z}_j}{\tilde h}\right)\frac{1}{2}f_j^{(2)}\left(W_j,\pi_j^*\right)\right|\left(\frac{M_j}{N_j}-\pi\right)^2\right]\\
        &\le  \frac{2}{n^2 \tilde h^p}\sum_{i=1}^n\sum_{j=1}^n  E\left[\frac{1}{p(\mathbf{z}_i)} \left|K\left(\frac{\mathbf{z}_i-\mathbf{z}_j}{\tilde h}\right)\frac{1}{2}
        \sup_{v\in(0,1)}|f_j^{(2)}\left(W_j,v\right)|\right|  \left(\frac{M_j}{N_j}-\pi\right)^2 \right] \\
        &=
		\frac{2}{n^2 \tilde h^p}\sum_{i=1}^n\sum_{j=1}^n 
		 E\left[
			\frac{1}{p(\mathbf{z}_i)} |K\left(\frac{\mathbf{z}_i-\mathbf{z}_j}{\tilde h}\right)|\frac{1}{2}E(  \sup_{v\in(0,1)}|f_j^{(2)}\left(W_j,v\right)|| \mathbf{z}_j)
		\right]E\left[\left(\frac{M_j}{N_j}-\pi\right)^2 \right] \\ 
		&=	
			\frac{1}{n^2 \tilde h^p}\sum_{i=1}^n\sum_{j=1}^n \int_{\mathbb{R}^p}
\int_{\mathbb{R}^p} |K\left(\frac{\mathbf{z}_i-\mathbf{z}_j}{\tilde h}\right)|p(\mathbf{z}_j)d\mathbf{z}_id\mathbf{z}_j E\left[\left(\frac{M_j}{N_j}-\pi\right)^2\right]\\
&=		\frac{1}{n^2 }\sum_{i=1}^n\sum_{j=1}^n \int_{\mathbb{R}^p}
\int_{\mathbb{R}^p}|K\left(u\right)|p(\mathbf{z}_j)dud\mathbf{z}_jE\left[\left(\frac{M_j}{N_j}-\pi\right)^2\right] =o(n^{-1/2}),
	\end{align*}
where the last equality follow from the fact that $K(u)$ are bounded, and $K(u)=0$ for all $u\in\{u:|u|_\infty\ge 1\}$. Then (\ref{lem4th6:delta13}) holds.
\end{proof}

\begin{lemma}
	\label{lem4th6:3}
		Under the assumptions for Theorem~\ref{th:5}, we have  
	\begin{align*} 
		\sqrt{n}(\delta_{2}-E\delta_{2})
		=&
		\frac{1}{\sqrt{n}}\sum_{i=1}^n \left(	E(f_i(1,\pi)-f_i(0,\pi)|\mathbf{z}_i)+
		f_i(W_i,\pi)\left(\frac{W_i}{\pi}-\frac{1-W_i}{1-\pi}\right)	\right.\notag\\+&\left.
		E(f_i(1,\pi)|\mathbf{z}_i)
		\frac{W_i}{\pi}-	E(f_i(0,\pi)|\mathbf{z}_i)\frac{1-W_i}{1-\pi}-3\tau
		\right.\notag\\+&\left.
		(W_i-\pi)
		\sum_{j\neq i}\frac{E_{ij}}{N_j}
		E\left(
		f_j^{(1)}\left(W_j,\pi\right)
		\left(\frac{W_j}{\pi}-\frac{(1-W_j)}{(1-\pi)}\right)
		\right)
		\right)+o_p(1),
	\end{align*}
	where $\delta_2$ is defined in (\ref{pf6:def:delta2}).  
\end{lemma}
\begin{proof}
	Let   $I_j=\mathbf{1}\left(	p(\mathbf{z}_j)>\tilde b\right)$ for $1\le j\le n$. 
	Recall that  
	\begin{align*} 
		\delta_2&=\frac{1}{n^2 \tilde h^p}\sum_{i=1}^n\sum_{j=1}^n \frac{1}{p^2(\mathbf{z}_i)} K\left(\frac{\mathbf{z}_i-\mathbf{z}_j}{\tilde h}\right)Y_j\left(\hat p_1(\mathbf{z}_i)\frac{W_j}{\pi}-\hat p_2(\mathbf{z}_i)\frac{1-W_j}{1-\pi}\right)	I\left(	p(\mathbf{z}_i)>\tilde b\right)\\
		&=\frac{1}{n^3\tilde h^{2p}}\sum_{i=1}^n\sum_{j=1}^n
		\sum_{k=1}^n 
		\frac{I_i}{p^2(\mathbf{z}_i)} K\left(\frac{\mathbf{z}_i-\mathbf{z}_j}{\tilde h}\right)
		K\left(\frac{\mathbf{z}_i-\mathbf{z}_k}{\tilde h}\right)
		Y_j\left(\frac{W_jW_k}{\pi^2}-
		\frac{(1-W_j)(1-W_k)}{(1-\pi)^2}\right).
	\end{align*}
	By taylor expansion,  similar to Lemma~\ref{lem4th6:2}, we decompose $ \delta_2$ as
	\begin{align*} 
		\delta_2=\delta_{21}+\delta_{22}+\delta_{23}
	\end{align*}
	where 
\begin{align*}
	\delta_{21}&=\frac{1}{n^3\tilde h^{2p}}\sum_{i=1}^n\sum_{j=1}^n
	\sum_{k=1}^n 
	\frac{I_i}{p^2(\mathbf{z}_i)} K\left(\frac{\mathbf{z}_i-\mathbf{z}_j}{\tilde h}\right)
	K\left(\frac{\mathbf{z}_i-\mathbf{z}_k}{\tilde h}\right)
	f_j\left(W_j,\pi\right)
	\left(\frac{W_jW_k}{\pi^2}-
	\frac{(1-W_j)(1-W_k)}{(1-\pi)^2}\right),\\
	\delta_{22}&=\frac{1}{n^3\tilde h^{2p}}\sum_{i=1}^n\sum_{j=1}^n
	\sum_{k=1}^n 
	\frac{I_i}{p^2(\mathbf{z}_i)} K\left(\frac{\mathbf{z}_i-\mathbf{z}_j}{\tilde h}\right)
	K\left(\frac{\mathbf{z}_i-\mathbf{z}_k}{\tilde h}\right)
	f_j^{(1)}\left(W_j,\pi\right)\left(\frac{M_j}{N_j}-\pi\right)\\&
	\left(\frac{W_jW_k}{\pi^2}-
	\frac{(1-W_j)(1-W_k)}{(1-\pi)^2}\right),\\
	\delta_{23}&=\frac{1}{n^3\tilde h^{2p}}\sum_{i=1}^n\sum_{j=1}^n
	\sum_{k=1}^n 
	\frac{I_i}{p^2(\mathbf{z}_i)} K\left(\frac{\mathbf{z}_i-\mathbf{z}_j}{\tilde h}\right)
	K\left(\frac{\mathbf{z}_i-\mathbf{z}_k}{\tilde h}\right)
	\frac{1}{2}f_j^{(2)}\left(W_j,\pi_j^*\right)\left(\frac{M_j}{N_j}-\pi\right)^2\\&
	\left(\frac{W_jW_k}{\pi^2}-
	\frac{(1-W_j)(1-W_k)}{(1-\pi)^2}\right).
\end{align*}

	The proof is divided into three parts. In Lemma~\ref{lem4th6:8}, we prove that  
	\begin{align}
		\label{lem4th6:delta21}  
			\sqrt{n}(\delta_{21}-E\delta_{21})=&
			\frac{1}{\sqrt{n}}\sum_{i=1}^n \left(
			E(f_i(1,\pi)-f_i(0,\pi)|\mathbf{z}_i)+
			f_i(W_i,\pi)\left(\frac{W_i}{\pi}-\frac{1-W_i}{1-\pi}\right)\right.\notag\\+&\left.
			E(f_i(1,\pi)|\mathbf{z}_i)\frac{W_i}{\pi}-E(f_i(0,\pi)|\mathbf{z}_i)\frac{1-W_i}{1-\pi}-3\tau
			\right)+o_p(1).
	\end{align}  
	
	In Lemma~\ref{lem4th6:9}, we prove that 
	\begin{align}
		\label{lem4th6:delta22} 
		\sqrt{n}(\delta_{22}-E\delta_{22})=
		\frac{1}{\sqrt{n}}\sum_{i=1}^n 
		\left(\frac{M_i}{N_i}-\pi\right)
		E\left(
		f_j^{(1)}\left(W_j,\pi\right)
		\left(\frac{W_j }{\pi}-\frac{1-W_j}{1-\pi }\right)
		\right)+o_p(1).
	\end{align}
	In the following part of this proof, we prove that 
	\begin{align}
		\label{lem4th6:delta23} 
		\delta_{23}=o_p(n^{-1/2}).
	\end{align}
	Then the results follows by combining (\ref{lem4th6:delta21}), (\ref{lem4th6:delta22}), and (\ref{lem4th6:delta23}).
 
	\tb{Proof of (\ref{lem4th6:delta23}).} Note that   \begin{enumerate}[label=(\alph*)]
		\item By Assumption~\ref{ass:main3}, $	\max_{w=0,1}\sup_{\mathbf{z}_1}E(\sup_{y\in [0,1]}|f_1^{(2)}(w,y)|^2|w,\mathbf{z}_1)
		<\infty$. As a result, \[	\max_{w=0,1}\sup_{\mathbf{z}_1} E(\sup_{y\in [0,1]}|f_1^{(2)}(w,y)||w,\mathbf{z}_1)
		<\infty.\]
		\item The terms  $\left(\frac{W_jW_k}{\pi^2}-\frac{(1-W_j)(1-W_k)}{(1-\pi)^2}\right)$ are bounded.
		\item $\max_j\left(\frac{M_j}{N_j}-\pi\right)^2=O_p(n^{-1}\rho_n^{-1}), E\left(\frac{M_j}{N_j}-\pi\right)^2\lesssim (n\rho_n)^{-1}$ by Lemma 15 in \cite{li2022random}.
	\end{enumerate}
	Then we derive that
	\begin{align*}
		E|\delta_{23}|&\lesssim \frac{1}{n\rho_n}
		\frac{1}{n^3 \tilde h^{2p}}\sum_{i=1}^n\sum_{j=1}^n \sum_{k=1}^n 
		E\left[
		\frac{I_i}{p^2(\mathbf{z}_i)} |K\left(\frac{\mathbf{z}_i-\mathbf{z}_j}{\tilde h}\right)
		K\left(\frac{\mathbf{z}_i-\mathbf{z}_k}{\tilde h}\right)
		|\frac{1}{2}E(\sup_{v\in[0,1]}|f_j^{(2)}\left(W_j,v\right)||\mathbf{z}_j)
		\right]\\ 
		&	\lesssim \frac{1}{n\rho_n}
		\frac{1}{n^3 \tilde h^{2p}}\sum_{i=1}^n\sum_{j=1}^n \sum_{k=1}^n
		\iiint |K\left(\frac{\mathbf{z}_i-\mathbf{z}_j}{\tilde h}\right)	K\left(\frac{\mathbf{z}_i-\mathbf{z}_k}{\tilde h}\right)|
		 \frac{p(\mathbf{z}_j)p(\mathbf{z}_k)}{p(\mathbf{z}_i)}I_i d\mathbf{z}_id\mathbf{z}_jd\mathbf{z}_k \\
		&\lesssim \frac{1}{n\rho_n}	\frac{1}{n^3  }\sum_{i=1}^n\sum_{j=1}^n \sum_{k=1}^n
		\iiint
		|K\left(\mathbf{u}_1\right)K\left(\mathbf{u}_2\right)|p(\mathbf{z}_i+\tilde h\mathbf{u}_1)\frac{p(\mathbf{z}_i+\tilde h\mathbf{u}_2)}{p(\mathbf{z}_i)}I_id\mathbf{u}_1d\mathbf{u}_2d\mathbf{z}_i.
	\end{align*}
	Note that since $\tilde h=o(\tilde b)$, $\nabla p(z)$ is bounded, we have
	\begin{align*}
	I(|\mathbf{u}_2|_\infty\le 1)	\left|\frac{p(\mathbf{z}_i+\tilde h\mathbf{u}_2)}{p(\mathbf{z}_i)}I_i-I_i\right|
		\lesssim	I(|\mathbf{u}_2|_\infty\le 1) \left|
		\frac{\tilde h\mathbf{u}_2^\top \mathbf{1}}{\tilde b}
		\right|\le \frac{\tilde h}{\tilde b}\lesssim 1. 
	\end{align*}
	Therefore, since $I_i\in[0,1]$, \[
		E|\delta_{23}|\lesssim
		 \frac{1}{n\rho_n}	\frac{1}{n^3  }\sum_{i=1}^n\sum_{j=1}^n \sum_{k=1}^n
		\iiint
		|K\left(\mathbf{u}_1\right)K\left(\mathbf{u}_2\right)|p(\mathbf{z}_i+\tilde h\mathbf{u}_1) d\mathbf{u}_1d\mathbf{u}_2d\mathbf{z}_i.
	\]
	Then by the fact that $K(u)$ is bounded,  $\int p(\mathbf{z}_i)d\mathbf{z}_i=1, \sqrt{n}\rho_n\to\infty$, and  $K(u)=0$ for all $u\in\{u:|u|_\infty\ge 1\}$, we have that  $
			E|\delta_{23}|=o(n^{-1/2}).$ Then by Markov inequality, (\ref{lem4th6:delta23}) holds.
\end{proof}

\begin{lemma}
	\label{lem4th6:4}
	Under the assumptions for Theorem~\ref{th:5}, we have $
	\delta_3=o_p(n^{-\frac{1}{2}})$ 
	where $\delta_3$ is defined in (\ref{pf6:def:delta3}).  
\end{lemma}
\begin{proof}
	For ease of reading, we recall the definition of $\delta_3$:
	\begin{align*}
		\delta_3&=\frac{1}{n^2 \tilde h^p}\sum_{i=1}^n \left(
		\frac{(\hat p_1(\mathbf{z}_i)-p(\mathbf{z}_i))^2}{\hat p_1(\mathbf{z}_i)p^2(\mathbf{z}_i)}
		\sum_{j=1}^n K\left(\frac{\mathbf{z}_i-\mathbf{z}_j}{\tilde h}\right)Y_j\frac{W_j}{\pi}\right.\notag\\ 
		&\left.-
		\frac{(\hat p_2(\mathbf{z}_i)-p(\mathbf{z}_i))^2}{\hat p_2(\mathbf{z}_i)p^2(\mathbf{z}_i)}\sum_{j=1}^n K\left(\frac{\mathbf{z}_i-\mathbf{z}_j}{\tilde h}\right)Y_j\frac{1-W_j}{1-\pi}
		\right)
		I\left(	p(\mathbf{z}_i)>\tilde b\right).
	\end{align*}
	To prove $	\delta_3=o_p(n^{-\frac{1}{2}})$, it suffices to show that
	\begin{align}
		\label{lem4th6:4:eqn1}
		&	\frac{1}{n^2 \tilde h^p}E\left(\sum_{i=1}^n 	\sum_{j=1}^n
		\frac{(\hat p_1(\mathbf{z}_i)-p(\mathbf{z}_i))^2}{\hat p_1(\mathbf{z}_i)p^2(\mathbf{z}_i)}
		| K\left(\frac{\mathbf{z}_i-\mathbf{z}_j}{\tilde h}\right)Y_j|\frac{W_j}{\pi}	I\left(	p(\mathbf{z}_i)>\tilde b\right)\right)=o(n^{-1/2}),\\
		\label{lem4th6:4:eqn2}
		&	\frac{1}{n^2 \tilde h^p}E\left(\sum_{i=1}^n 	\sum_{j=1}^n
		\frac{(\hat p_2(\mathbf{z}_i)-p(\mathbf{z}_i))^2}{\hat p_2(\mathbf{z}_i)p^2(\mathbf{z}_i)}
		| K\left(\frac{\mathbf{z}_i-\mathbf{z}_j}{\tilde h}\right)Y_j|\frac{1-W_j}{1-\pi}	I\left(	p(\mathbf{z}_i)>\tilde b\right)\right)=o(n^{-1/2}),
	\end{align}
	since the result then follows by Markov's inequality. We shall only prove (\ref{lem4th6:4:eqn1}), as  (\ref{lem4th6:4:eqn2}) follows by similar arguments.
	
	Under Assumption~\ref{ass:main3}, by Lemma~\ref{lem4th6:1}, we have
	\begin{align*} 
		\sup_{\mathbf{z}\in\bb{R}^p}\left| 	\hat p_1(\mathbf{z})-p(\mathbf{z})\right|&=O_{a.s.}\left(\sqrt{\frac{\log(n)}{n \tilde h^p}}+\tilde h^q\right).
	\end{align*}
	Moreover, Assumption~\ref{ass:main3} implies that  $\sqrt{\frac{\log(n)}{n \tilde h^p}}+\tilde h^q=o(\tilde b n^{-\frac{1}{4}})$. Then we derive that  \begin{align*} 
		&	\frac{1}{n^2 \tilde h^p}E\left(\sum_{i=1}^n 	\sum_{j=1}^n
		\frac{(\hat p_1(\mathbf{z}_i)-p(\mathbf{z}_i))^2}{\hat p_1(\mathbf{z}_i)p^2(\mathbf{z}_i)}
		| K\left(\frac{\mathbf{z}_i-\mathbf{z}_j}{\tilde h}\right)Y_j|\frac{W_j}{\pi}	I\left(	p(\mathbf{z}_i)>\tilde b\right)\right)\\
		\lesssim& \frac{1}{n^2 \tilde h^p} \frac{\left(\sqrt{\frac{\log(n)}{n \tilde h^p}}+\tilde h^q\right)^2}{b|b-\sqrt{\frac{\log(n)}{n \tilde h^p}}-\tilde h^q|} \sum_{i=1}^n 	\sum_{j=1}^n E\left|\frac{1}{p(\mathbf{z}_i)}K\left(\frac{\mathbf{z}_i-\mathbf{z}_j}{\tilde h}\right)f_j\left(1,\mni\right)\right|\\
		= & o\left(n^{-\frac{1}{2}}
		\frac{1}{n^2\tilde h^p}\sum_{i=1}^n\sum_{j=1}^n 
		E\left[\left|\frac{1}{p(\mathbf{z}_i)}K\left(\frac{\mathbf{z}_i-\mathbf{z}_j}{\tilde h}\right)\right|E\left(\left|f_j\left(1,\mni\right)\right|\bigg|\mathbf{z}_j\right)\right]\right).\\
		= & o\left(n^{-\frac{1}{2}}
		\frac{1}{n^2\tilde h^p}\sum_{i=1}^n\sum_{j=1}^n 
		E \left|\frac{1}{p(\mathbf{z}_i)}K\left(\frac{\mathbf{z}_i-\mathbf{z}_j}{\tilde h}\right)\right| \right)\\
		=& o\left(n^{-\frac{1}{2}}
		\int_{\bb{R}^p}\int_{\bb{R}^p}|K(\mathbf{u})| p(\mathbf{z}_i+\tilde h\mathbf{u})d\mathbf{z}_id\mathbf{u}\right)
		=o (n^{-\frac{1}{2}}),
	\end{align*}
	where the second equality follows from the fact that $f_j$ is independent of $\mathbf{z}_i$ when $i\neq j$, the assumption that  $\sup_{\mathbf{z}_1}\sup_{y\in[0,1]}E(|f_1(1,y)||\mathbf{z}_1)<\infty$ and the Cauchy-Schwarz inequality. The last equality holds because both \( p({z}) \) and \( K(u) \) are bounded, and \( K(u) = 0 \) for all \( u \) such that \( |u|_\infty \geq 1 \). 

Thus, (\ref{lem4th6:4:eqn1}) holds, and consequently, the proof is complete.
\end{proof}

\begin{lemma}
	\label{lem4th6:5}
	Under the assumptions for Theorem~\ref{th:5}, we have 
		\begin{align*}
		\sqrt{n}(\delta_{11}-E\delta_{11})=
		\frac{2}{\sqrt{n}}\sum_{i=1}^n\left(
		f_i\left(W_i,\pi\right)\left(\frac{W_i}{\pi}-\frac{1-W_i}{1-\pi}\right)+E\left(f_i\left(1,\pi\right)-f_i\left(0,\pi\right) |\mathbf{z}_i\right)-2\tau
		\right)+o_p(1),
	\end{align*} 
	where 
	\begin{align*}
			\delta_{11}=\frac{2}{n^2 \tilde h^p}\sum_{i=1}^n\sum_{j=1}^n \frac{1}{p(\mathbf{z}_i)} K\left(\frac{\mathbf{z}_i-\mathbf{z}_j}{\tilde h}\right)f_j\left(W_j,\pi\right)\left(\frac{W_j}{\pi}-\frac{1-W_j}{1-\pi}\right)	I\left(	p(\mathbf{z}_i)>\tilde b\right).
	\end{align*}
\end{lemma}
\begin{proof}
	The proof is divided into three parts. In \tb{Step 1}, we relate $\delta_{11}$ to a U-statistic $U_{11}$. In \tb{Step 2}, we approximate \(U_{11}\) using conditional expectation. In \tb{Step 3}, we further approximate the expression to derive the final result.  
	
	\tb{Step 1.}  Let $R_j=\frac{W_j}{\pi}-\frac{1-W_j}{1-\pi}, I_j=I\left(	p(\mathbf{z}_j)>\tilde b\right)$ for $1\le j\le n$, and
	 \begin{align*}
	 	U_{11}&=\frac{2}{n(n-1) \tilde h^p}\sum_{i<j}   K\left(\frac{\mathbf{z}_i-\mathbf{z}_j}{\tilde h}\right)\left(
	 	\frac{	f_j\left(W_j,\pi\right)R_jI_i }{p(\mathbf{z}_i)}
	 	+
	 	\frac{	f_i\left(W_i,\pi\right)R_iI_j }{p(\mathbf{z}_j)}
	 	\right).
	 \end{align*}
	 Then by the symmetry of $K(\cdot)$,
	 \begin{align*}
	 	\delta_{11}-\frac{n-1}{n}U_{11}-E(\delta_{11}-\frac{n-1}{n}U_{11})&=\frac{2K(\mathbf{0})}{n}\sum_{i=1}^n
	 	\frac{1}{n \tilde h^p}
	 	\left(\frac{f_i\left(W_i,\pi\right)R_iI_i}{p(\mathbf{z}_i)}-
	 	E\frac{f_i\left(W_i,\pi\right)R_iI_i}{p(\mathbf{z}_i)}\right).
	 \end{align*}
	 Note that the right-hand-side of the above equation is an average of i.i.d. samples, and its variance
	 \begin{align*}
	 	&var(	\delta_{11}-\frac{n-1}{n}U_{11}-E(\delta_{11}-\frac{n-1}{n}U_{11}))
	 	=\frac{4 K^2(\mathbf{0})}{n}var\left(
	 	\frac{1}{n \tilde h^p}
	  \frac{f_i\left(W_i,\pi\right)R_iI_i}{p(\mathbf{z}_i)} 
	 	\right)\\
	 	&\le \frac{4 K^2(\mathbf{0})}{n^3\tilde h^{2p}}E\left(
 \frac{f_i\left(W_i,\pi\right)R_iI_i}{p(\mathbf{z}_i)} 
	 	\right)^2\le \frac{4 K^2(\mathbf{0})}{n^3\tilde h^{2p}}E\left(
	 	\frac{f_i\left(W_i,\pi\right)R_i }{\tilde b} 
	 	\right)^2\\
	 	&\le \frac{4 K^2(\mathbf{0})}{\pi^2(1-\pi)^2n^3\tilde h^{2p}\tilde b^2}E\left(|
	 	 {f_i\left(W_i,\pi\right)  } |
	 	\right)^2
	 	=
	 	O\left(\frac{1}{n^3h^{2p}\tilde b^2}\right).
	 \end{align*}
	 By Assumption~\ref{ass:main3} and the Markov inequality, we have that \begin{align}
	 	\label{lem4th6:5 eqn0}
	 	\delta_{11}-\frac{n-1}{n}U_{11}-E(\delta_{11}-\frac{n-1}{n}U_{11})=o_p(n^{-\frac{1}{2}}).
	 \end{align}	 
	 That is, $\sqrt{n}(\delta_{11}-E\delta_{11})$ is well-approximated by $\sqrt{n}(U_{11}-EU_{11})$.
	 
	 \tb{Step 2.}  Note that \( U_{11} \) is a U-statistic constructed from i.i.d. samples of the random vector \( (\mathbf{z}_i^\top, f_i, W_i)^\top \). However, it has a varying kernel due to the fact that the bandwidth \( \tilde{h} \) depends on \( n \).  By Lemma 3.1 in \cite{powell1986semiparametric}, to show 
	 \begin{align}
	 	\label{eqn:7}
	 	\sqrt{n}(U_{11}-EU_{11})&=
	 	\frac{1}{\sqrt{n}}\sum_{i=1}^n\left[
	 	2E\left(\frac{1}{ \tilde h^p}
	 	K\left(\frac{\mathbf{z}_i-\mathbf{z}_j}{\tilde h}\right)\left(
	 	\frac{	f_j\left(W_j,\pi\right)R_jI_i }{p(\mathbf{z}_i)}
	 	+
	 	\frac{	f_i\left(W_i,\pi\right)R_iI_j }{p(\mathbf{z}_j)}
	 	\right)
	 	\bigg |\mathbf{z}_i,W_i,f_i \right)\right.\notag\\&\left.-
	 	2E\left(\frac{1}{ \tilde h^p}
	 	K\left(\frac{\mathbf{z}_i-\mathbf{z}_j}{\tilde h}\right)\left(
	 	\frac{	f_j\left(W_j,\pi\right)R_jI_i }{p(\mathbf{z}_i)}
	 	+
	 	\frac{	f_i\left(W_i,\pi\right)R_iI_j }{p(\mathbf{z}_j)}
	 	\right)
	 	\right)
	 	\right]+o_p(1),
	 \end{align}
	 it suffices to prove that
	 \begin{align}
	 	\label{eqn:6}
	 	E\left(\frac{1}{ \tilde h^p}
	 	K\left(\frac{\mathbf{z}_i-\mathbf{z}_j}{\tilde h}\right)\left(
	 	\frac{	f_j\left(W_j,\pi\right)R_jI_i }{p(\mathbf{z}_i)}
	 	+
	 	\frac{	f_i\left(W_i,\pi\right)R_iI_j }{p(\mathbf{z}_j)}
	 	\right)
	 	\right)^2=o(n).
	 \end{align}
	 To see (\ref{eqn:6}), notice that  
	 \begin{align*}
	 	& \ E\left(\frac{1}{ \tilde h^p}
	 	K\left(\frac{\mathbf{z}_i-\mathbf{z}_j}{\tilde h}\right)\left(
	 	\frac{	f_j\left(W_j,\pi\right)R_jI_i }{p(\mathbf{z}_i)}
	 	+
	 	\frac{	f_i\left(W_i,\pi\right)R_iI_j }{p(\mathbf{z}_j)}
	 	\right)
	 	\right)^2\\
	 	&\le \frac{1}{\tilde h^{2p}\tilde b^2} E\left[
	 	K^2\left(\frac{\mathbf{z}_i-\mathbf{z}_j}{\tilde h}\right) 
	 	(f_j\left(W_j,\pi\right)R_jI_i+
	 	f_i\left(W_i,\pi\right)R_iI_j)^2
	 	\right]\\
	 	&\le \frac{1}{\tilde h^{2p}\tilde b^2} E\left[
	 	K^2\left(\frac{\mathbf{z}_i-\mathbf{z}_j}{\tilde h}\right) [
	 	2(f_j\left(W_j,\pi\right)R_jI_i)^2+
	 	2(f_i\left(W_i,\pi\right)R_iI_j)^2]
	 	\right]\\
	 	&\lesssim 
	 	\frac{1}{\tilde h^{2p}\tilde b^2} E\left[
	 	K^2\left(\frac{\mathbf{z}_i-\mathbf{z}_j}{\tilde h}\right) [
	 	E(f_j^2\left(W_j,\pi\right)|\mathbf{z}_j ) +
	 	E(f_i^2\left(W_i,\pi\right)|\mathbf{z}_i ) ]
	 	\right]\\
	 	&\lesssim \frac{1}{\tilde h^{2p}\tilde b^2} {\int K^2\left(\frac{\mathbf{z}_i-\mathbf{z}_j}{\tilde h}\right) 
	 		p(\mathbf{z}_i)p(\mathbf{z}_j)d\mathbf{z}_id\mathbf{z}_j}\\
	 	&\lesssim\frac{1}{ \tilde h^p\tilde b^2}\int K^2(\mathbf{u})p(\mathbf{z}_j+\tilde h\mathbf{u} )p(\mathbf{z}_j)d\mathbf{u} d\mathbf{z}_j\\
	 	&\lesssim\frac{1}{ \tilde h^p\tilde b^2}\int K^2(\mathbf{u}) p(\mathbf{z}_j)d\mathbf{u} d\mathbf{z}_j \le \frac{1}{ \tilde h^p\tilde b^2}\sup_{\mathbf{u}\in \bb{R}^p} K^2(\mathbf{u}) =o(n),
	 \end{align*}
	 where we use the boundedness of $R_i, E(f_i^2(W_i,\pi)|\mathbf{z}_i), K(\mathbf{u}), p(\mathbf{z})$, and $n \tilde h^p\tilde b^2\to \infty$ which is implied by  Assumption~\ref{ass:main3}. 
	 
	 \tb{Step 3.}  We now refine (\ref{eqn:7}) to show that the right-hand side can be approximated by \[\frac{1}{\sqrt{n}}\sum_{i=1}^n[2E(f_i(W_i,\pi)R_i|\mathbf{z}_i)+2f_i(W_i,\pi)R_i-4\tau],\] which is an average whose components do not vary with $n$.  Consider the difference term
	 \begin{align*}
	 	J=&\frac{1}{\sqrt{n}}\sum_{i=1}^n\left[
	 	2E\left(\frac{1}{ \tilde h^p}
	 	K\left(\frac{\mathbf{z}_i-\mathbf{z}_j}{\tilde h}\right)\left(
	 	\frac{	f_j\left(W_j,\pi\right)R_jI_i }{p(\mathbf{z}_i)}
	 	+
	 	\frac{	f_i\left(W_i,\pi\right)R_iI_j }{p(\mathbf{z}_j)}
	 	\right)
	 	\bigg |\mathbf{z}_i,W_i,f_i \right)\right.\notag\\&\left.-
	 	2E\left(\frac{1}{ \tilde h^p}
	 	K\left(\frac{\mathbf{z}_i-\mathbf{z}_j}{\tilde h}\right)\left(
	 	\frac{	f_j\left(W_j,\pi\right)R_jI_i }{p(\mathbf{z}_i)}
	 	+
	 	\frac{	f_i\left(W_i,\pi\right)R_iI_j }{p(\mathbf{z}_j)}
	 	\right)
	 	\right)
	 	\right.\notag\\&\left.-
	 	 2E(f_i(W_i,\pi)R_i|\mathbf{z}_i)-2f_i(W_i,\pi)R_i+4\tau\right].
	 \end{align*}
	 Note that $J/\sqrt{n}$ is an average of i.i.d. samples of random vector $(\mathbf{z}_i^\top, f_i, W_i)^\top$, and that $\tau=E(f_i(W_i,\pi)R_i)=E(f_i(1,\pi)-f_i(0,\pi))$. Therefore, to show $J=o_p(1)$, it suffices to show that
	 \begin{align}
	 	\label{lem4th6:5 J1}
	 	J_1=var&\left[
	 	2E\left(\frac{1}{ \tilde h^p}
	 	K\left(\frac{\mathbf{z}_i-\mathbf{z}_j}{\tilde h}\right)\left(
	 	\frac{	f_j\left(W_j,\pi\right)R_jI_i }{p(\mathbf{z}_i)}
	 	+
	 	\frac{	f_i\left(W_i,\pi\right)R_iI_j }{p(\mathbf{z}_j)}
	 	\right)
	 	\bigg |\mathbf{z}_i,W_i,f_i \right)\right.\notag\\&\left.-
	 	2E\left(\frac{1}{ \tilde h^p}
	 	K\left(\frac{\mathbf{z}_i-\mathbf{z}_j}{\tilde h}\right)\left(
	 	\frac{	f_j\left(W_j,\pi\right)R_jI_i }{p(\mathbf{z}_i)}
	 	+
	 	\frac{	f_i\left(W_i,\pi\right)R_iI_j }{p(\mathbf{z}_j)}
	 	\right)
	 	\right)
	 	\right.\notag\\&\left.-
	 	2E(f_i(W_i,\pi)R_i|\mathbf{z}_i)-2f_i(W_i,\pi)R_i+4\tau\right]
	 	=o(1).
	 \end{align}
	 By letting $S_1(u)=E(f_i(W_i,\pi)R_i|\mathbf{z}_i=u)$, and using the law of total expectation and integration by substitution, we derive that
	 \begin{align*}
	 	&2E\left(\frac{1}{ \tilde h^p}
	 	K\left(\frac{\mathbf{z}_i-\mathbf{z}_j}{\tilde h}\right)\left(
	 	\frac{	f_j\left(W_j,\pi\right)R_jI_i }{p(\mathbf{z}_i)}
	 	+
	 	\frac{	f_i\left(W_i,\pi\right)R_iI_j }{p(\mathbf{z}_j)}
	 	\right)
	 	\bigg |\mathbf{z}_i,W_i,f_i \right)\\
	 	&=2\int  \frac{1}{ \tilde h^p}
	 	K\left(\frac{\mathbf{z}_i-\mathbf{z}_j}{\tilde h}\right)\left(
	 	\frac{
	 		S_1(\mathbf{z}_j) I_i }{p(\mathbf{z}_i)}	 
	 +
	 \frac{	f_i\left(W_i,\pi\right)R_iI_j }{p(\mathbf{z}_j)}
	 \right) p(\mathbf{z}_j)d\mathbf{z}_j  \\
	 &=2\int 
	 K\left(\mathbf{u}\right)\left(
	 \frac{
	 	S_1(\mathbf{z}_i+\tilde h\mathbf{u})  I_i }{p(\mathbf{z}_i)}	 
	 		+
	 	\frac{	f_i\left(W_i,\pi\right)R_iI(p(\mathbf{z}_i+\tilde h\mathbf{u})>\tilde b) }{p(\mathbf{z}_i+\tilde h\mathbf{u})}
	 	\right) p(\mathbf{z}_i+\tilde h\mathbf{u})d\mathbf{u}\\
	 	&=2\int 
	 	K\left(\mathbf{u}\right)
	 	\frac{p(\mathbf{z}_i+\tilde h\mathbf{u})
	 S_1(\mathbf{z}_i+\tilde h\mathbf{u})  I_i }{p(\mathbf{z}_i)}d\mathbf{u}	 
	 +
2\int 
K(\mathbf{u})	f_i\left(W_i,\pi\right)R_iI(p(\mathbf{z}_i+hu)>\tilde b) d\mathbf{u}.
\end{align*}
Therefore, 
\begin{align*}
	J_1&=var\left(
	2\int 
	K\left(\mathbf{u}\right)
	\frac{p(\mathbf{z}_i+\tilde h\mathbf{u})
		S_1(\mathbf{z}_i+\tilde h\mathbf{u})  I_i }{p(\mathbf{z}_i)}d\mathbf{u}
		-	 2S_1(\mathbf{z}_i)
		\right.\\&\left.
	+
	2\int 
	K(\mathbf{u})	f_i\left(W_i,\pi\right)R_iI(p(\mathbf{z}_i+hu)>\tilde b) d\mathbf{u}
	-2f_i(W_i,\pi)R_i
	\right)\\
	&\le  2var\left(
	2\int 
	K\left(\mathbf{u}\right)
	\frac{p(\mathbf{z}_i+\tilde h\mathbf{u})
		S_1(\mathbf{z}_i+\tilde h\mathbf{u})  I_i }{p(\mathbf{z}_i)}d\mathbf{u}
	-	 2S_1(\mathbf{z}_i)\right)\\
	&+ 2var\left(	2\int 
	K(\mathbf{u})	f_i\left(W_i,\pi\right)R_iI(p(\mathbf{z}_i+\tilde h\mathbf{u})>\tilde b) d\mathbf{u}
	-2f_i(W_i,\pi)R_i
	\right)\\
	&\lesssim 8\iint 	\left|
	K\left(\mathbf{u}\right)
	\frac{p(\mathbf{z}_i+\tilde h\mathbf{u})	 I_i
		S_1(\mathbf{z}_i+\tilde h\mathbf{u})
		-p(\mathbf{z}_i)S_1(\mathbf{z}_i)}{p(\mathbf{z}_i)}\right|^2p(\mathbf{z}_i) dud\mathbf{z}_i \\
		&+8\iint 	\left|
		K\left(u\right) (1-I(p(\mathbf{z}_i+\tilde h\mathbf{u})>\tilde b))\right|^2 p(\mathbf{z}_i) d\mathbf{u}d \mathbf{z}_i,
\end{align*}
where in the last inequality we apply Jensen's inequality to exchange the order of integration and squaring, and we also use the boundedness of $E(f_i^2(W_i,\pi)R_i|\mathbf{z}_i)$.   Since $p(\mathbf{z}), \nabla p(\mathbf{z}) , K(\mathbf{u}), S_1(\mathbf{z}_i) $, the support of $K(\mathbf{u})$  are bounded and $\tilde h=o(\tilde b)$, we have
 \begin{align}
 	\label{eqn:26}
 &	\left|
 	K\left(\mathbf{u}\right)
 	\frac{p(\mathbf{z}_i+\tilde h\mathbf{u})	 I_i
 		S_1(\mathbf{z}_i+\tilde h\mathbf{u})
 		-p(\mathbf{z}_i)S_1(\mathbf{z}_i)}{p(\mathbf{z}_i)}\right|^2p(\mathbf{z}_i)\notag\\
 		\lesssim & K^2(\mathbf{u})p(\mathbf{z}_i)\left(1+ 
 		I_i\frac{p(\mathbf{z}_i+\tilde h\mathbf{u})}{p(\mathbf{z}_i)}\right)^2
 		\lesssim K^2(\mathbf{u})p(\mathbf{z}_i)\left(1+ 
 		 \frac{\tilde h^2 (\mathbf{1}^\top \mathbf{u})^2}{\tilde b^2}\right)\notag\\
\lesssim& K^2(\mathbf{u})p(\mathbf{z}_i)\left(1+
 { (\mathbf{1}^\top \mathbf{u})^2} \right),
 \end{align}
 which is integrable. Moreover, 
 \begin{align*}
	\left|
K\left(\mathbf{u}\right) (1-I(p(\mathbf{z}_i+\tilde h\mathbf{u})>\tilde b))\right|^2 p(\mathbf{z}_i)
\lesssim K^2(\mathbf{u})p(\mathbf{z}_i)
 \end{align*} which is also  integrable.  As a result, we can apply the  dominated convergence theorem to obtain that $J_1=o(1)$, and consequently, $J=o_p(1)$. 
 Therefore,   (\ref{eqn:7}) is equlaivent to 
\begin{align}\label{lem4th6:5 eqn2} 
\sqrt{n}(U_{11}-EU_{11})&=
\frac{2}{\sqrt{n}}\sum_{i=1}^n\left(
f_i\left(W_i,\pi\right)R_i+E(f_i\left(W_i,\pi\right)R_i|\mathbf{z}_i)-2\tau
\right)+o_p(1).
\end{align}
By (\ref{lem4th6:5 eqn0}), (\ref{lem4th6:5 eqn2}), and the Slutsky's theorem, the final result follows.
\end{proof}

\begin{lemma}
		\label{lem4th6:6}
			Under the assumptions for Theorem~\ref{th:5}, we have 
	\begin{align*} 
		\sqrt{n}(\delta_{12}-E\delta_{12})=	\frac{2}{\sqrt n}\sum_{j=1}^n
		\left(\frac{M_j}{N_j}-\pi\right)
		E\left(
		f_j^{(1)}\left(W_j,\pi\right)\left(\frac{W_j}{\pi}-\frac{1-W_j}{1-\pi}\right)
		\right)+o_p(1),
	\end{align*}
		where 
		\begin{align*}
				\delta_{12}&=\frac{2}{n^2 \tilde h^p}\sum_{i=1}^n\sum_{j=1}^n \frac{1}{p(\mathbf{z}_i)} K\left(\frac{\mathbf{z}_i-\mathbf{z}_j}{\tilde h}\right)f_j^{(1)}\left(W_j,\pi\right)\left(\frac{M_j}{N_j}-\pi\right)\left(\frac{W_j}{\pi}-\frac{1-W_j}{1-\pi}\right)	I\left(	p(\mathbf{z}_i)>\tilde b\right).
		\end{align*}
\end{lemma}
\begin{proof}
	 Let $R_j=\frac{W_j}{\pi}-\frac{1-W_j}{1-\pi}, I_j=I\left(	p(\mathbf{z}_j)>\tilde b\right)$ for $1\le j\le n$, 
	\begin{align*}
		U_{12}&=
		\frac{2}{n(n-1) \tilde h^p}\sum_{i=1}^n\sum_{j\neq i}
		\left(\frac{M_j}{N_j}-\pi\right)
		\frac{I_i}{p(\mathbf{z}_i)}
		K\left(\frac{\mathbf{z}_i-\mathbf{z}_j}{\tilde h}\right)
		f_j^{(1)}\left(W_j,\pi\right)R_j.
	\end{align*}
	Then by  the symmetry of $K(\cdot)$, 
	\begin{align*}
		&	\delta_{12}-\frac{n-1}{n}U_{12}-E(\delta_{12}-\frac{n-1}{n}U_{12})\\&=\frac{2K(0)}{n}\sum_{i=1}^n
		\frac{1}{n \tilde h^p}
		\left(\frac{f_i^{(1)}\left(W_i,\pi\right)\left(\frac{M_i}{N_i}-\pi\right)R_iI_i}{p(\mathbf{z}_i)}-
		E\frac{f_i^{(1)}\left(W_i,\pi\right)\left(\frac{M_i}{N_i}-\pi\right)R_iI_i}{p(\mathbf{z}_i)}\right),
	\end{align*}
	and its variance 
	\begin{align*}
			&var(\delta_{12}-\frac{n-1}{n}U_{12})\\
			&=\frac{4K^2(\mathbf{0})}{n^4\tilde h^{2p}}\sum_{i=1}^n\sum_{j=1}^n E\left[
			\left(\frac{f_i^{(1)}\left(W_i,\pi\right)\left(\frac{M_i}{N_i}-\pi\right)R_iI_i}{p(\mathbf{z}_i)}-
			E\frac{f_i^{(1)}\left(W_i,\pi\right)\left(\frac{M_i}{N_i}-\pi\right)R_iI_i}{p(\mathbf{z}_i)}\right)
			\right.\\&\left.
			\left(\frac{f_j^{(1)}\left(W_j,\pi\right)\left(\frac{M_j}{N_j}-\pi\right)R_jI_j}{p(\mathbf{z}_j)}-
			E\frac{f_j^{(1)}\left(W_j,\pi\right)\left(\frac{M_j}{N_j}-\pi\right)R_jI_j}{p(\mathbf{z}_j)}\right)
			\right]\\
			&=\frac{4K^2(\mathbf{0})}{n^4\tilde h^{2p}}\sum_{i=1}^n E\left[
			\left(\frac{f_i^{(1)}\left(W_i,\pi\right)\left(\frac{M_i}{N_i}-\pi\right)R_iI_i}{p(\mathbf{z}_i)}-
			E\frac{f_i^{(1)}\left(W_i,\pi\right)\left(\frac{M_i}{N_i}-\pi\right)R_iI_i}{p(\mathbf{z}_i)}\right)\right]^2\\
			&+
		 \frac{4K^2(\mathbf{0})}{n^4\tilde h^{2p}}\sum_{i\neq j}E\left[
		 	\left(\frac{f_i^{(1)}\left(W_i,\pi\right)\left(\frac{M_i}{N_i}-\pi\right)R_iI_i}{p(\mathbf{z}_i)}-
		 E\frac{f_i^{(1)}\left(W_i,\pi\right)\left(\frac{M_i}{N_i}-\pi\right)R_iI_i}{p(\mathbf{z}_i)}\right)
		 	\right.\\&\left.
			\left(\frac{f_j^{(1)}\left(W_j,\pi\right)\left(\frac{M_j}{N_j}-\pi\right)R_jI_j}{p(\mathbf{z}_j)}-
			E\frac{f_j^{(1)}\left(W_j,\pi\right)\left(\frac{M_j}{N_j}-\pi\right)R_jI_j}{p(\mathbf{z}_j)}\right)
			\right]\\
			&\le \frac{4K^2(\mathbf{0})}{n^4\tilde h^{2p}}\sum_{i=1}^n E 
			\left(\frac{f_i^{(1)}\left(W_i,\pi\right)\left(\frac{M_i}{N_i}-\pi\right)R_iI_i}{p(\mathbf{z}_i)} \right) ^2\\
			&+
			\frac{4K^2(\mathbf{0})}{n^4\tilde h^{2p}}\sum_{i\neq j}E\left[\left|
			\left(\frac{f_i^{(1)}\left(W_i,\pi\right)\left(\frac{M_i}{N_i}-\pi\right)R_iI_i}{p(\mathbf{z}_i)}\right)\left(\frac{f_j^{(1)}\left(W_j,\pi\right)\left(\frac{M_j}{N_j}-\pi\right)R_jI_j}{p(\mathbf{z}_j)}\right)\right|\right]\\
			&+
				\frac{4K^2(\mathbf{0})}{n^4\tilde h^{2p}}\sum_{i\neq j}
				E\left|\frac{f_i^{(1)}\left(W_i,\pi\right)\left(\frac{M_i}{N_i}-\pi\right)R_iI_i}{p(\mathbf{z}_i)}\right|
				E\left|\frac{f_j^{(1)}\left(W_j,\pi\right)\left(\frac{M_j}{N_j}-\pi\right)R_jI_j}{p(\mathbf{z}_j)}\right|  \\
					&\lesssim \frac{1}{n^4\tilde h^{2p}\tilde b^2}\sum_{i=1}^n E 
				\left( f_i^{(1)}\left(W_i,\pi\right)  \right) ^2\\
				&+
				\frac{1}{n^4\tilde h^{2p}\tilde b^2}\sum_{i\neq j}E\left[\left|
				 E(|f_i^{(1)}\left(W_i,\pi\right)||W_i)\left(\frac{M_i}{N_i}-\pi\right)   E(|f_j^{(1)}\left(W_j,\pi\right)||W_j)\left(\frac{M_j}{N_j}-\pi\right) \right|\right]\\
				&+
				\frac{1}{n^4\tilde h^{2p}\tilde b^2}\sum_{i\neq j}
				E\left|E(|f_i^{(1)}\left(W_i,\pi\right)||W_i)\left(\frac{M_i}{N_i}-\pi\right) \right|
				E\left|E(|f_j^{(1)}\left(W_j,\pi\right)||W_j)\left(\frac{M_j}{N_j}-\pi\right) \right|  
				\\
				&\lesssim 
		\frac{1}{n^3\tilde h^{2p}\tilde b^2}+\frac{1}{n^4\tilde h^{2p}\tilde b^2}\sum_{i<j}\left(E\left|\left(
		\frac{M_i}{N_i}-\pi
		\right)\left(
		\frac{M_j}{N_j}-\pi
		\right)\right|+E\left|
		\frac{M_i}{N_i}-\pi
		\right|E\left|
		\frac{M_j}{N_j}-\pi
		\right|
		\right)\\
		&\lesssim 	\frac{1}{n^3\tilde h^{2p}\tilde b^2}+\frac{1}{n^4\tilde h^{2p}\tilde b^2}\sum_{i<j}
		\left(\sqrt{E \left(
		\frac{M_i}{N_i}-\pi
		\right)^2}
		\sqrt{E \left(
			\frac{M_j}{N_j}-\pi
			\right)^2}
		\right)\\
		&\lesssim  \frac{1}{n^3\tilde h^{2p}\tilde b^2}+
		\frac{1}{n^4\tilde h^{2p}\tilde b^2}\frac{n^2}{n\rho_n} 
		\lesssim \frac{1}{n^3\tilde h^{2p}\tilde b^2\rho_n},
	\end{align*}
	where we use (i) $\max_{w=0,1}E[|f_i^{(1)}(w,\pi)|^2|w]<\infty$, (ii) $I_i/p(\mathbf{z}_i)\le 1/\tilde b$ for $1\le i\le n$, (iii) the boundedness of $R_i,I_i,M_i/N_i$ for $1\le i\le n$, (iv) Lemma 15(c) in \citet{li2022random}.  By Assumption~\ref{ass:main3} and the Markov inequality, we have $\delta_{12}-\frac{n-1}{n}U_{12}-E(\delta_{12}-\frac{n-1}{n}U_{12})=o_p(n^{-\frac{1}{2}})$.  As a result,
	\begin{align}
	\label{lem4th6:7:eqn1}
	\sqrt{n}(\delta_{12}-E\delta_{12})-\sqrt{n}(U_{12}-EU_{12})=o_p(1).
	\end{align} 
	
	We now use the  projection technique to approximate $U_{12}$.  Let
	\begin{equation}
		\label{eqn:rn}
		r_{n}=
		\frac{2}{n(n-1) \tilde h^p}\sum_{i=1}^n\sum_{j\neq i}
		\left(\frac{M_j}{N_j}-\pi\right)
		\frac{I_i}{p(\mathbf{z}_i)}
		E\left(
		K\left(\frac{\mathbf{z}_i-\mathbf{z}_j}{\tilde h}\right)
		f_j^{(1)}\left(W_j,\pi\right)R_j
		\bigg|\mathbf{z}_i\right).
	\end{equation}
	We prove that $U_{12}-r_n$ is negligible by showing that $E\left(U_{12}-r_n\right)^2=o(n^{-1})$.
	\begin{align}
		&E(U_{12}-r_n)^2\notag\\&=
		E\left[
		\frac{2}{n(n-1)}\sum_{i=1}^n\sum_{j:j\neq i}
		\frac{I_i}{p(\mathbf{z}_i)}
		\left(\frac{M_j}{N_j}-\pi\right)
		\left(
		\frac{1}{ \tilde h^p} K\left(\frac{\mathbf{z}_i-\mathbf{z}_j}{\tilde h}\right)
		f_j^{(1)}\left(W_j,\pi\right)R_j \right.\right.\notag\\&-\left.\left.
		E\left(
		\frac{1}{ \tilde h^p}K\left(\frac{\mathbf{z}_i-\mathbf{z}_j}{\tilde h}\right)
		f_j^{(1)}\left(W_j,\pi\right)R_j
		\bigg|\mathbf{z}_i\right)\right)
		\right]^2\notag\\
		&=
		\frac{4}{n^2(n-1)^2}\sum_{i=1}^n\sum_{j:j\neq i}\sum_{l=1}^n\sum_{m:m\neq l}
		E
        \frac{\sum_{k:k\neq j}E_{kj}(W_k-\pi)
			\sum_{q:q\neq m}E_{qm}(W_q-\pi)
		}{N_jN_m}
	\notag	\\
		&
		\left(
		\frac{I_i}{p(\mathbf{z}_i) \tilde h^p} K\left(\frac{\mathbf{z}_i-\mathbf{z}_j}{\tilde h}\right)
		f_j^{(1)}\left(W_j,\pi\right)R_j -
		E\left(
		\frac{I_i}{p(\mathbf{z}_i) \tilde h^p}K\left(\frac{\mathbf{z}_i-\mathbf{z}_j}{\tilde h}\right)
		f_j^{(1)}\left(W_j,\pi\right)R_j
		\bigg|\mathbf{z}_i\right)\right)\notag\\
		&
		\left(
		\frac{I_l}{p(\mathbf{z}_l) \tilde h^p} K\left(\frac{\mathbf{z}_l-\mathbf{z}_m}{\tilde h}\right)
		f_m^{(1)}\left(W_m,\pi\right)R_m -
		E\left(
		\frac{I_l}{p(\mathbf{z}_l) \tilde h^p}K\left(\frac{\mathbf{z}_l-\mathbf{z}_m}{\tilde h}\right)
		f_m^{(1)}\left(W_m,\pi\right)R_m
		\bigg|\mathbf{z}_l\right)\right).	\label{eqn:8}
	\end{align}
	We note that most of the terms in (\ref{eqn:8}) are zero because $EW_k=\pi$ for all $1\le k\le n$,  $\{W_i\}$ are i.i.d., $\{\mathbf{z}_i\}$ are i.i.d., and $\{W_i\}$ are independent of $\{\mathbf{z}_i\}$. Specifically, the summation term is nonzero only if both of the following conditions occur:
	\begin{enumerate}
		\item $m=j$, or $j=l, m=i$;
		\item $k=m,q=j$, or $k=q$.
	\end{enumerate}
	Moreover, by Assumption~\ref{ass:main3} we have
	\begin{align}\label{eqn:9}
		&\max_{w=0,1}E	\left(\left|
		\frac{I_i}{p(\mathbf{z}_i) \tilde h^p} K\left(\frac{\mathbf{z}_i-\mathbf{z}_j}{\tilde h}\right)
		f_j^{(1)}\left(W_j,\pi\right)R_j\right|\bigg|W_j=w\right) \notag\\
		&\lesssim 
		\int 	\frac{I_i}{p(\mathbf{z}_i) \tilde h^p} |K\left(\frac{\mathbf{z}_i-\mathbf{z}_j}{\tilde h}\right)| p(\mathbf{z}_i)p(\mathbf{z}_j)d\mathbf{z}_id\mathbf{z}_j
		\le
		\int |K(\mathbf{u})|p(\mathbf{z}_i+\tilde h\mathbf{u})d\mathbf{u}d\mathbf{z}_i,
	\end{align}
	which is bounded. Then (\ref{eqn:8}) reduces to
	\begin{align}
		\label{eqn20}
		E(U_{12}-r_n)^2&\lesssim	\frac{1}{n^4}\sum_{i=1}^n\sum_{j:j\neq i}\sum_{l=1}^n\sum_{m:m\neq l}
		E\frac{\sum_{k:k\neq j}E_{kj}
			\sum_{q:q\neq m}E_{qm}
		}{N_jN_m}\notag\\
		&\quad I(\{m=j\}\cup \{j=l,m=i\} )I(\{k=q\}\cup \{k=m,q=j\})\notag\\
		&\lesssim 
		\frac{1}{n^4}E\sum_{i,j,l:i\neq j,l\neq j}\frac{\sum_{k:k\neq j}E_{kj}^2}{N_j^2}+ 
		\frac{1}{n^4}E \sum_{i,j:i\neq j}\frac{\sum_{k:k\neq i, j}E_{kj}E_{ki}+E_{ij}}{N_jN_i}.
	\end{align}
	By Lemma 15 and the proof of Lemma 17 in \cite{li2022random}, \[
P(N_i<cn\rho_n)<\exp(-Cn\rho_n), E(1/N_i)\le C/(n\rho_n), E\sum_{i,j:i\neq j}\sum_{k:k\neq i, j}E_{kj}E_{ki}=O(n^3\rho_n^2)\] for some constants $c,C>0$. Therefore, $E(U_{12}-r_n)^2=O(n^{-2}\rho_n^{-1}+n^{-3})=o(n^{-1})$, which indicates that 
		\begin{align}
		\label{lem4th6:7:eqn2}
		\sqrt{n}(U_{12}-EU_{12})-\sqrt{n}(r_n-Er_n)=o_p(1).
	\end{align} 
	
	 We now refine (\ref{eqn:rn}) to simplify it. Note that
	\begin{align*}
		r_n&=\frac{2}{n}\sum_{j=1}^n
		\left(\frac{M_j}{N_j}-\pi\right)
		\frac{1}{n-1}
		\sum_{i\neq j}
		E\left(
		\frac{I_i}{ \tilde h^p p(\mathbf{z}_i)}
		K\left(\frac{\mathbf{z}_i-\mathbf{z}_j}{\tilde h}\right)
		f_j^{(1)}\left(W_j,\pi\right)R_j
		\bigg|\mathbf{z}_i\right)\\
		&=\frac{2}{n}\sum_{j=1}^n
		\left(\frac{M_j}{N_j}-\pi\right)
		E\left(
		f_j^{(1)}\left(W_j,\pi\right)R_j
		\right)+\tilde r_{n1}+\tilde r_{n2}+\tilde r_{n3},
	\end{align*}
	where 
	\begin{align*}
		\tilde r_{n1}&=\frac{2}{n}\sum_{j=1}^n
		\left(\frac{M_j}{N_j}-\pi\right)
		\left(
		\frac{1}{n(n-1)}
		\sum_{i:i\neq j}
		E\left(
		\frac{I_i}{ \tilde h^p p(\mathbf{z}_i)}
		K\left(\frac{\mathbf{z}_i-\mathbf{z}_j}{\tilde h}\right)
		f_j^{(1)}\left(W_j,\pi\right)R_j
		\bigg|\mathbf{z}_i\right)\right.\\
		&-\left.
		\frac{1}{n}
		E\left(
		\frac{I_j}{ \tilde h^p p(\mathbf{z}_j)}
		K\left(\mathbf{0}\right)
		f_j^{(1)}\left(W_j,\pi\right)R_j
		\bigg|\mathbf{z}_j\right)	
		\right),\\
		\tilde r_{n2}&=\frac{2}{n}\sum_{j=1}^n
		\left(\frac{M_j}{N_j}-\pi\right)
		\left(
		\frac{1}{n}
		\sum_{i=1}^n
		E\left(
		\frac{I_i}{ \tilde h^p p(\mathbf{z}_i)}
		K\left(\frac{\mathbf{z}_i-\mathbf{z}_j}{\tilde h}\right)
		f_j^{(1)}\left(W_j,\pi\right)R_j
		\bigg|\mathbf{z}_i\right)\right.\\
		&-\left.
		E\left(
		\frac{I_j}{ \tilde h^p p(\mathbf{z}_j)}
		K\left(\frac{\mathbf{z}_i-\mathbf{z}_j}{\tilde h}\right)
		f_j^{(1)}\left(W_j,\pi\right)R_j
		\right)	
		\right),\\
		\tilde r_{n3}&=\frac{2}{n}\sum_{j=1}^n
		\left(\frac{M_j}{N_j}-\pi\right)
		\left(
		E\left(
		\frac{I_i}{ \tilde h^p p(\mathbf{z}_i)}
		K\left(\frac{\mathbf{z}_i-\mathbf{z}_j}{\tilde h}\right)
		f_j^{(1)}\left(W_j,\pi\right)R_j
		\right)-
		E\left(
		f_j^{(1)}\left(W_j,\pi\right)R_j
		\right)	\right).
	\end{align*}
	
	We now prove that $\tilde r_{n1}=o_p(n^{-1/2}), \tilde r_{n2}=o_p(n^{-1/2}), \tilde r_{n3}=o_p(n^{-1/2})$, so that $r_n$ can be well-approximated by $\frac{2}{n}\sum_{j=1}^n
	\left(\frac{M_j}{N_j}-\pi\right)
	E\left(
	f_j^{(1)}\left(W_j,\pi\right)R_j
	\right)$. 
	
	\tb{For $\tilde r_{n1}$:}  By (\ref{eqn:9}) and $|\frac{M_i}{N_i}-\pi|<1$, we have $E|\tilde r_{n1}|\lesssim O(n^{-1})$. Therefore,  $\tilde r_{n1}=O_p(n^{-1})$. 
	
	\tb{For $\tilde r_{n2}$:} By Cauchy-Schwarz inequality, 
	\begin{align*}
	E|	\tilde r_{n2}|&\le 
		\sqrt{E\left(\frac{2}{n}\sum_{j=1}^n\left(\frac{M_j}{N_j}-\pi\right)\right)^2}
		\sqrt{var\left(
		\frac{1}{n}
		\sum_{i=1}^n
		E\left(
		\frac{I_i}{ \tilde h^p p(\mathbf{z}_i)}
		K\left(\frac{\mathbf{z}_i-\mathbf{z}_j}{\tilde h}\right)
		f_j^{(1)}\left(W_j,\pi\right)R_j
		\bigg|\mathbf{z}_i\right)\right) }.
	\end{align*} 
	For the first term in the right-hand-side, we have
	\begin{align}\label{eqn:11}
		E\left(\frac{1}{n}\sum_{j=1}^n
		\left(\frac{M_j}{N_j}-\pi\right)\right)^2&=\frac{1}{n^2}\sum_{i,j} E\left(	\left(\frac{M_i}{N_i}-\pi\right)	\left(\frac{M_j}{N_j}-\pi\right)\right)\notag\\
		&=O(\frac{1}{n^2\rho_n})+\frac{1}{n^2}\sum_{i\neq j} E\left(
		\frac{
			\sum_{k\neq i}E_{ik}(W_k-\pi)
			\sum_{l\neq j}E_{lj}(W_l-\pi)
		}{N_iN_j}\right)\notag\\
		&=O(\frac{1}{n^2\rho_n})+\frac{1}{n^2}\sum_{i\neq j}E\left(
		\frac{
			\sum_{k\neq i,j} E_{ik}E_{kj}
		}{N_iN_j}\pi(1-\pi)
		\right)\notag\\
		&=O(\frac{1}{n^2\rho_n}+\frac{1}{n})=O(\frac{1}{n}),
	\end{align}
	where we apply Lemmas 15 and the proof of Lemma 17 in \citet{li2022random}, as similarly done in the derivations in (\ref{eqn20}). 
	
	For the second term in the right-hand-side, we have
	\begin{align*}
	&var\left(	\frac{1}{n}
		\sum_{i=1}^n
		E\left(
		\frac{I_i}{ \tilde h^p p(\mathbf{z}_i)}
		K\left(\frac{\mathbf{z}_i-\mathbf{z}_j}{\tilde h}\right)
		f_j^{(1)}\left(W_j,\pi\right)R_j
		\bigg|\mathbf{z}_i\right)\right)\\
		&\le \frac{1}{n}
		var\left(	 
		E\left(
		\frac{I_i}{ \tilde h^p p(\mathbf{z}_i)}
		K\left(\frac{\mathbf{z}_i-\mathbf{z}_j}{\tilde h}\right)
		f_j^{(1)}\left(W_j,\pi\right)R_j
		\bigg|\mathbf{z}_i\right) \right)\\
			&\le \frac{1}{n}
		E\left(
		\frac{I_i}{ \tilde h^p p(\mathbf{z}_i)}
		K\left(\frac{\mathbf{z}_i-\mathbf{z}_j}{\tilde h}\right)
		f_j^{(1)}\left(W_j,\pi\right)R_j
		 \right)^2\\
			&\lesssim \frac{1}{n}
		E\left(
		\left(
		\frac{I_i}{ \tilde h^p p(\mathbf{z}_i)}
		K\left(\frac{\mathbf{z}_i-\mathbf{z}_j}{\tilde h}\right)
		\right)^2
	E(	(f_j^{(1)}\left(W_j,\pi\right))^2|\mathbf{z}_j)
		\right) \\
			&\lesssim \frac{1}{n\tilde h^p\tilde b}
		\int 
	 {  p(\mathbf{z}_j)}
		K^2\left(\frac{\mathbf{z}_i-\mathbf{z}_j}{\tilde h}\right) d\mathbf{z}_i d\mathbf{z}_j\\
			&\lesssim \frac{1}{n\tilde h^p\tilde b}
		\int 
		 {  p(\mathbf{z}_j)}
		K^2\left(\mathbf{u}\right) d\mathbf{u} d\mathbf{z}_j \lesssim \frac{1}{n\tilde h^p\tilde b}=o(1).
	\end{align*} 
	Therefore, $E|\tilde r_{n2}|=o(n^{-1/2})$. 
	
	\tb{For $\tilde r_{n3}$:}  similar to the previous derivations, it suffices to notice that 
	\begin{align*}
		&E\left(
		\frac{I_i}{ \tilde h^p p(\mathbf{z}_i)}
		K\left(\frac{\mathbf{z}_i-\mathbf{z}_j}{\tilde h}\right)
		f_j^{(1)}\left(W_j,\pi\right)R_j
		\right)-E\left(
		f_j^{(1)}\left(W_j,\pi\right)R_j
		\right)\\
		&=E\left(
		\int \frac{I_i}{ \tilde h^p }
		K\left(\frac{\mathbf{z}_i-\mathbf{z}_j}{\tilde h}\right)d\mathbf{z}_i
		f_j^{(1)}\left(W_j,\pi\right)R_j
		-	f_j^{(1)}\left(W_j,\pi\right)R_j
		\right)\\
		&=E\left(\left(
		\int I(p(\mathbf{z}_j+\tilde h\mathbf{u})>\tilde b)
		K\left(\mathbf{u}\right)d\mathbf{u}-1\right)
	E(	f_j^{(1)}\left(W_j,\pi\right)R_j|\mathbf{z}_j)
		\right)\to 0,
	\end{align*}
	where we use the boundedness of $E(	f_j^{(1)}\left(W_j,\pi\right)R_j|\mathbf{z}_j), K(\mathbf{u})$, the fact that $\tilde h,\tilde b\to 0$ as $n\to\infty$, and the dominated convergence theorem. Then we have $\tilde r_{n3}=o_p(n^{-\frac{1}{2}})$.
	
	Therefore, we have 
	\begin{align}
		\label{eqn:210}
			\sqrt{n}(r_n-Er_n)=\frac{2}{\sqrt n}\sum_{j=1}^n
		\left(\frac{M_j}{N_j}-\pi\right)
		E\left(
		f_j^{(1)}\left(W_j,\pi\right)R_j
		\right).
	\end{align} 
	By combining (\ref{eqn:210}) with (\ref{lem4th6:7:eqn1}), (\ref{lem4th6:7:eqn2}), the final result follows.
\end{proof}

\begin{lemma}
	\label{lem4th6:8}
	Under the assumptions for Theorem~\ref{th:5}, we have  
	\begin{align*} 
		\sqrt{n}(\delta_{21}-E\delta_{21})=&
		\frac{1}{\sqrt{n}}\sum_{i=1}^n \left(
		E(f_i(1,\pi)-f_i(0,\pi)|\mathbf{z}_i)+
		f_i(W_i,\pi)\left(\frac{W_i}{\pi}-\frac{1-W_i}{1-\pi}\right)\right.\notag\\+&\left.
		E(f_i(1,\pi)|\mathbf{z}_i)\frac{W_i}{\pi}-E(f_i(0,\pi)|\mathbf{z}_i)\frac{1-W_i}{1-\pi}-3\tau
		\right)+o_p(1),
	\end{align*} 
	where 
	\begin{align*}
		\delta_{21}&=\frac{1}{n^3\tilde h^{2p}}\sum_{i=1}^n\sum_{j=1}^n
		\sum_{k=1}^n 
		\frac{I(p(\mathbf{z}_i)>\tilde b)}{p^2(\mathbf{z}_i)} K\left(\frac{\mathbf{z}_i-\mathbf{z}_j}{\tilde h}\right)
		K\left(\frac{\mathbf{z}_i-\mathbf{z}_k}{\tilde h}\right)
		f_j\left(W_j,\pi\right)
		\left(\frac{W_jW_k}{\pi^2}-
		\frac{(1-W_j)(1-W_k)}{(1-\pi)^2}\right).
	\end{align*} 
\end{lemma}
\begin{proof} 
	The proof is divided into three parts. In \tb{Step 1}, we associate $\delta_{21}$ with a U-statistic $U_{21}$. In \tb{Step 2}, we approximate \(U_{21}\) using conditional expectation. In \tb{Step 3}, we further approximate the expression to derive the final result.  
	
	\tb{Step 1.}  Let $R_{j}=\frac{W_j }{\pi}-\frac{1-W_j}{1-\pi}, I_j=I(p(\mathbf{z}_j)>\tilde b),  R_{jk}=\frac{W_jW_k}{\pi^2}-\frac{(1-W_j)(1-W_k)}{(1-\pi)^2}$ for $1\le j,k\le n$, 
	\begin{align*}
		U_{21}=\frac{6}{n(n-1)(n-2)}\sum_{i<j<k}&\frac{1}{6}\left(
		K\left(\frac{\mathbf{z}_i-\mathbf{z}_j}{\tilde h}\right)
		K\left(\frac{\mathbf{z}_i-\mathbf{z}_k}{\tilde h}\right)
		\frac{I_iR_{jk}}{\tilde h^{2p}p^2(\mathbf{z}_i)} (
		f_j\left(W_j,\pi\right)+	f_k\left(W_k,\pi\right)
		)\right.\\&\left.+
		K\left(\frac{\mathbf{z}_j-\mathbf{z}_k}{\tilde h}\right)
		K\left(\frac{\mathbf{z}_j-\mathbf{z}_i}{\tilde h}\right)
		\frac{I_jR_{ik}}{\tilde h^{2p}p^2(\mathbf{z}_j)} (
		f_k\left(W_k,\pi\right)+	f_i\left(W_i,\pi\right)
		)\right.\\&\left.+
		K\left(\frac{\mathbf{z}_k-\mathbf{z}_i}{\tilde h}\right)
		K\left(\frac{\mathbf{z}_k-\mathbf{z}_j}{\tilde h}\right)
		\frac{I_kR_{ij}}{\tilde h^{2p}p^2(\mathbf{z}_k)} (
		f_i\left(W_i,\pi\right)+
		f_j\left(W_j,\pi\right)
		)
		\right).
	\end{align*}
	Then by the symmetry of $K(\cdot)$,
	\begin{align*} 
		\delta_{21}-\frac{(n-1)(n-2)}{n^2}U_{21}\lesssim &
		\frac{1}{n^3\tilde h^{2p}}\sum_{i,j: i\neq j}\frac{I_i}{p^2(\mathbf{z}_i)}	K^2\left(\frac{\mathbf{z}_i-\mathbf{z}_j}{\tilde h}\right)f_j(W_j,\pi)\left(\frac{W_j^2}{\pi^2}-\frac{(1-W_j)^2}{(1-\pi)^2}\right)\\
		+& 	\frac{1}{n^3\tilde h^{2p}}\sum_{i,j: i\neq j}\frac{I_i}{p^2(\mathbf{z}_i)}	K(\mathbf{0})K\left(\frac{\mathbf{z}_i-\mathbf{z}_j}{\tilde h}\right)(f_j(W_j,\pi)+f_i(W_i,\pi))R_{ij}\\
		+&\frac{1}{n^3\tilde h^{2p}}\sum_{i=1}^n \frac{I_i}{p^2(\mathbf{z}_i)}	K^2(\mathbf{0})f_i(W_i,\pi)\left(\frac{W_i^2}{\pi^2}-\frac{(1-W_i)^2}{(1-\pi)^2}\right).
	\end{align*} 
	
	To show that $\sqrt{n}(\delta_{21}-E\delta_{21})-\sqrt{n}(U_{21}-EU_{21})=o_p(1)$, it suffices to show that $E|\delta_{21}-\frac{(n-1)(n-2)}{n^2}U_{21}|=o(1/\sqrt n)$.
	
	First, consider the term $	\frac{1}{n^3\tilde h^{2p}}\sum_{i\neq j}\frac{I_i}{p^2(\mathbf{z}_i)}	K^2\left(\frac{\mathbf{z}_i-\mathbf{z}_j}{\tilde h}\right)f_j(W_j,\pi)\left(\frac{W_j^2}{\pi^2}-\frac{(1-W_j)^2}{(1-\pi)^2}\right)$, we derive that
	\begin{align}
		\label{eqn:25}
		&E\left|\frac{1}{n^3\tilde h^{2p}}\sum_{i\neq j}\frac{I_i}{p^2(\mathbf{z}_i)}	K^2\left(\frac{\mathbf{z}_i-\mathbf{z}_j}{\tilde h}\right)f_j(W_j,\pi)\left(\frac{W_j^2}{\pi^2}-\frac{(1-W_j)^2}{(1-\pi)^2}\right)\right|\notag\\
		&\le 
		\frac{1}{n^3\tilde h^{2p} \pi^2(1-\pi)^2}
		E\left|\sum_{i\neq j}\frac{I_i}{p^2(\mathbf{z}_i)}	K^2\left(\frac{\mathbf{z}_i-\mathbf{z}_j}{\tilde h}\right)E(f_j(W_j,\pi)|\mathbf{z}_j) \right|\notag\\
		&\lesssim 
		\frac{1}{n^3\tilde h^{2p}\tilde b}
		\sum_{i\neq j}
		E\left|\frac{1}{p(\mathbf{z}_i)}	K^2\left(\frac{\mathbf{z}_i-\mathbf{z}_j}{\tilde h}\right) \right|\notag\\
		&\le 	\frac{1}{n \tilde h^{p}\tilde b}
		\iint
		p(\mathbf{z}_j)	K^2\left(\mathbf{u}\right) d\mathbf{u}d\mathbf{z}_j\lesssim \frac{1}{n \tilde h^{p}\tilde b}=o(n^{-1/2}),
	\end{align}
	where the last equality is due to Assumption~\ref{ass:main3}. 
	
	Second, consider the term $\frac{1}{n^3\tilde h^{2p}}\sum_{i\neq j}\frac{I_i}{p^2(\mathbf{z}_i)}	K(\mathbf{0})K\left(\frac{\mathbf{z}_i-\mathbf{z}_j}{\tilde h}\right)(f_j(W_j,\pi)+f_i(W_i,\pi))R_{ij}$. Similar to the derivations in (\ref{eqn:25}), we have  
	\begin{align*}
		&E\left|\frac{1}{n^3\tilde h^{2p}}\sum_{i\neq j}\frac{I_i}{p^2(\mathbf{z}_i)}	K(\mathbf{0})K\left(\frac{\mathbf{z}_i-\mathbf{z}_j}{\tilde h}\right)(f_j(W_j,\pi)+f_i(W_i,\pi))R_{ij}\right|\\
		&\lesssim
		\frac{1}{n^3\tilde h^{2p}}
		E\left|\sum_{i\neq j}\frac{I_i}{p^2(\mathbf{z}_i)}	K\left(\frac{\mathbf{z}_i-\mathbf{z}_j}{\tilde h}\right)\right|\\
		&\le \frac{1}{n\tilde h^{p}\tilde b}
		\iint p(\mathbf{z}_j)|K(\mathbf{u})|d\mathbf{u}d\mathbf{z}_j =o(1/\sqrt n).
	\end{align*}
	
	Finally, for the term $\frac{1}{n^3\tilde h^{2p}}\sum_{i=1}^n \frac{I_i}{p^2(\mathbf{z}_i)}	K^2(\mathbf{0})f_i(W_i,\pi)\left(\frac{W_i^2}{\pi^2}-\frac{(1-W_i)^2}{(1-\pi)^2}\right)$, we have
	\begin{align*}
		&\frac{1}{n^3\tilde h^{2p}}E\left|\sum_{i=1}^n \frac{I_i}{p^2(\mathbf{z}_i)}	K^2(\mathbf{0})f_i(W_i,\pi)\left(\frac{W_i^2}{\pi^2}-\frac{(1-W_i)^2}{(1-\pi)^2}\right)\right|\\
		&\lesssim
		\frac{1}{n^2\tilde h^{2p}\tilde b^2}=o(1/\sqrt n).
	\end{align*}
	Therefore, we have
	\begin{align}
		\label{lem4th6:5 eqn0 new}
		\sqrt{n}(\delta_{21}-E\delta_{21})-\sqrt{n}(U_{21}-EU_{21})=o_p(1).
	\end{align}

	\tb{Step 2.}   Similar to  {Step 2} in Lemma~\ref{lem4th6:5}, by the proof of Theorem 4.2.1 in \cite{korolyuk2013theory} and the $c_r$-inequality for expectation, to show 
	\begin{align}
		\label{eqn:18}
		\sqrt{n}(U_{21}-EU_{21})=\frac{1}{\sqrt{n}}\sum_{i=1}^n \left(
		g_{21}(\mathbf{z}_i,W_i,f_i)-Eg_{21}(\mathbf{z}_i,W_i,f_i)
		\right)+o_p(1),
	\end{align}
	where 
	\begin{align*}
		g_{21}(\mathbf{z}_i,W_i,f_i)&=
		\frac{1}{2}E\left(
		K\left(\frac{\mathbf{z}_i-\mathbf{z}_j}{\tilde h}\right)
		K\left(\frac{\mathbf{z}_i-\mathbf{z}_k}{\tilde h}\right)
		\frac{I_iR_{jk}}{\tilde h^{2p}p^2(\mathbf{z}_i)} (
		f_j\left(W_j,\pi\right)+	f_k\left(W_k,\pi\right)
		)\right.\\&\left.+
		K\left(\frac{\mathbf{z}_j-\mathbf{z}_k}{\tilde h}\right)
		K\left(\frac{\mathbf{z}_j-\mathbf{z}_i}{\tilde h}\right)
		\frac{I_jR_{ik}}{\tilde h^{2p}p^2(\mathbf{z}_j)} (
		f_k\left(W_k,\pi\right)+	f_i\left(W_i,\pi\right)
		)\right.\\&\left.+
		K\left(\frac{\mathbf{z}_k-\mathbf{z}_i}{\tilde h}\right)
		K\left(\frac{\mathbf{z}_k-\mathbf{z}_j}{\tilde h}\right)
		\frac{I_kR_{ij}}{\tilde h^{2p}p^2(\mathbf{z}_k)} (
		f_i\left(W_i,\pi\right)+
		f_j\left(W_j,\pi\right)
		)
		\bigg|\mathbf{z}_i,W_i,f_i
		\right),
	\end{align*}
	it suffices to prove that
	\begin{align}
		\label{eqn:15}
		&E\left(K\left(\frac{\mathbf{z}_i-\mathbf{z}_j}{\tilde h}\right)
		K\left(\frac{\mathbf{z}_i-\mathbf{z}_k}{\tilde h}\right)
		\frac{I_iR_{jk}}{\tilde h^{2p}p^2(\mathbf{z}_i)} 
		f_j\left(W_j,\pi\right)
		\right)	^2=o(n).
	\end{align}
	
	To see (\ref{eqn:15}), notice that  by the boundedness of $R_{jk}, E(|f_j(W_j,\pi)||\mathbf{z}_j)$, we have
	\begin{align*}
		&E\left(K\left(\frac{\mathbf{z}_i-\mathbf{z}_j}{\tilde h}\right)
		K\left(\frac{\mathbf{z}_i-\mathbf{z}_k}{\tilde h}\right)
		\frac{I_iR_{jk}}{\tilde h^{2p}p^2(\mathbf{z}_i)} 
		f_j\left(W_j,\pi\right)
		\right)	^2\\
		&\lesssim
		\frac{1}{h^{4p}}
		\int K^2\left(\frac{\mathbf{z}_i-\mathbf{z}_j}{\tilde h}\right)
		K^2\left(\frac{\mathbf{z}_i-\mathbf{z}_k}{\tilde h}\right)
		\frac{I_i}{p^3(\mathbf{z}_i)}p(\mathbf{z}_j)p(\mathbf{z}_k)d\mathbf{z}_id\mathbf{z}_jd\mathbf{z}_k\\
		&	\lesssim \frac{1}{h^{2p}b^3}	\int K^2\left(\mathbf{u}_1\right)
		K^2\left(\mathbf{u}_2\right)
		p(\mathbf{z}_i+\tilde h\mathbf{u}_1)p(\mathbf{z}_i+\tilde h\mathbf{u}_2)d\mathbf{u}_1d\mathbf{u}_2d\mathbf{z}_i\\&
		\lesssim \frac{1}{h^{2p}b^3}
		=o(n).
	\end{align*}
	
	\tb{Step 3.}  We now refine (\ref{eqn:18}) to show that the summands on the  right-hand side can be approximated by $	E(f_i(1,\pi)-f_i(0,\pi)|\mathbf{z}_i)+
	f_i(W_i,\pi)\left(\frac{W_i}{\pi}-\frac{1-W_i}{1-\pi}\right)
	+E(f_i(1,\pi)|\mathbf{z}_i)\frac{W_i}{\pi}-E(f_i(0,\pi)|\mathbf{z}_i)\frac{1-W_i}{1-\pi}
	-3\tau$, which do not vary with $n$. Let \[S_{2}(u)=	E(	f_j\left(W_j,\pi\right)R_{jk}|\mathbf{z}_j)\bigg|_{\mathbf{z}_j=u}, S_{3}(v,u)=	E(	f_j\left(W_j,\pi\right)R_{jk}|W_j,\mathbf{z}_j)\bigg|_{W_j=v,\mathbf{z}_j=u}.\]  Consider the difference term
	\begin{align*}
		J=&\frac{1}{\sqrt{n}}\sum_{i=1}^n\left[g_{21}(\mathbf{z}_i,W_i,f_i)-Eg_{21}(\mathbf{z}_i,W_i,f_i)-  E(f_i(1,\pi)-f_i(0,\pi)|\mathbf{z}_i)
		\right.\\&\left.
		-
		f_i(W_i,\pi)\left(\frac{W_i}{\pi}-\frac{1-W_i}{1-\pi}\right)
		-E(f_i(1,\pi)|\mathbf{z}_i)\frac{W_i}{\pi}+E(f_i(0,\pi)|\mathbf{z}_i)\frac{1-W_i}{1-\pi}
		+3\tau\right].
	\end{align*}
	Note that $J/\sqrt{n}$ is an average of functions of i.i.d. samples of random vector $(\mathbf{z}_i^\top, f_i, W_i)^\top$. By the definition of \(\tau\) in Theorem~\ref{th: Consistency of hat tau}, each summand of \(J\) has a mean of zero. Therefore, to show $J=o_p(1)$, it suffices to show that
	\begin{align}
		\label{lem4th6:5 J2}
		J_1=var&\left[g_{21}(\mathbf{z}_i,W_i,f_i)-Eg_{21}(\mathbf{z}_i,W_i,f_i)-  E(f_i(1,\pi)-f_i(0,\pi)|\mathbf{z}_i)
		\right.\notag\\&\left.
		-
		f_i(W_i,\pi)\left(\frac{W_i}{\pi}+\frac{1-W_i}{1-\pi}\right)
		-E(f_i(1,\pi)|\mathbf{z}_i)\frac{W_i}{\pi}+E(f_i(0,\pi)|\mathbf{z}_i)\frac{1-W_i}{1-\pi}
		+3\tau\right]
		=o(1).
	\end{align}
	By the law of total expectation and integration by substitution, we derive that
	\begin{align*}
		g_{21}(\mathbf{z}_i,W_i,f_i)&=
		\frac{1}{2}E\left(
		K\left(\frac{\mathbf{z}_i-\mathbf{z}_j}{\tilde h}\right)
		K\left(\frac{\mathbf{z}_i-\mathbf{z}_k}{\tilde h}\right)
		\frac{I_iR_{jk}}{\tilde h^{2p}p^2(\mathbf{z}_i)} (
		f_j\left(W_j,\pi\right)+	f_k\left(W_k,\pi\right)
		)\right.\\&\left.+
		K\left(\frac{\mathbf{z}_j-\mathbf{z}_k}{\tilde h}\right)
		K\left(\frac{\mathbf{z}_j-\mathbf{z}_i}{\tilde h}\right)
		\frac{I_jR_{ik}}{\tilde h^{2p}p^2(\mathbf{z}_j)} (
		f_k\left(W_k,\pi\right)+	f_i\left(W_i,\pi\right)
		)\right.\\&\left.+
		K\left(\frac{\mathbf{z}_k-\mathbf{z}_i}{\tilde h}\right)
		K\left(\frac{\mathbf{z}_k-\mathbf{z}_j}{\tilde h}\right)
		\frac{I_kR_{ij}}{\tilde h^{2p}p^2(\mathbf{z}_k)} (
		f_i\left(W_i,\pi\right)+
		f_j\left(W_j,\pi\right)
		)
		\bigg|\mathbf{z}_i,W_i,f_i
		\right),\\
		&=
		\frac{1}{2} \int 
		K\left(\mathbf{u}_1\right)
		K\left(\mathbf{u}_2\right)
		\frac{I_i}{p^2(\mathbf{z}_i)} [
		S_{2}(\mathbf{z}_i+\tilde h\mathbf{u}_1)+ 
		S_{2}(\mathbf{z}_i+\tilde h\mathbf{u}_2)
		] p(\mathbf{z}_i+\tilde h\mathbf{u}_1)
		p(\mathbf{z}_i+\tilde h\mathbf{u}_2)d\mathbf{u}_1d\mathbf{u}_2\\
		&+	 \frac{1}{2} \int 
		K\left(\mathbf{u}_1\right)
		K\left(\mathbf{u}_2\right)
		\frac{I(p(\mathbf{z}_i+\tilde h\mathbf{u}_1)>\tilde b) }{ p^2(\mathbf{z}_i+\tilde h\mathbf{u}_1)} \\&[
		S_{3}(W_i,\mathbf{z}_i+\tilde h\mathbf{u}_1+\tilde h\mathbf{u}_2)+ 
		f_i(W_i,\pi)R_i
		]
		p(\mathbf{z}_i+\tilde h\mathbf{u}_1)
		p(\mathbf{z}_i+\tilde h\mathbf{u}_1+\tilde h\mathbf{u}_2)
		d\mathbf{u}_1d\mathbf{u}_2\\
		&+	\frac{1}{2} \int 
		K\left(\mathbf{u}_1\right)
		K\left(\mathbf{u}_2\right)
		\frac{I(p(\mathbf{z}_i+\tilde h\mathbf{u}_1)>\tilde b)}{ p^2(\mathbf{z}_i+\tilde h\mathbf{u}_1)} \\&[
		f_i(W_i,\pi)R_i+ 
		S_{3}(W_i,\mathbf{z}_i+\tilde h\mathbf{u}_1+\tilde h\mathbf{u}_2)
		]
		p(\mathbf{z}_i+\tilde h\mathbf{u}_1)
		p(\mathbf{z}_i+\tilde h\mathbf{u}_1+\tilde h\mathbf{u}_2)
		d\mathbf{u}_1d\mathbf{u}_2.
	\end{align*}
	Recall that $S_{2}(u)=	E(	f_j\left(W_j,\pi\right)R_{jk}|\mathbf{z}_j)\bigg|_{\mathbf{z}_j=u}, S_{3}(v,u)=	E(	f_j\left(W_j,\pi\right)R_{jk}|W_j,\mathbf{z}_j)\bigg|_{W_j=v,\mathbf{z}_j=u}$. We have
	\begin{align*}
		&S_2(\mathbf{z}_i)+f_i(W_i,\pi)R_i+S_3(W_i,\mathbf{z}_i))\\&=  E(f_i(1,\pi)-f_i(0,\pi)|\mathbf{z}_i)
		+
		f_i(W_i,\pi)\left(\frac{W_i}{\pi}-\frac{1-W_i}{1-\pi}\right)
		+E(f_i(1,\pi)|\mathbf{z}_i)\frac{W_i}{\pi}-E(f_i(0,\pi)|\mathbf{z}_i)\frac{1-W_i}{1-\pi}.
	\end{align*}
	Then
	\begin{align*}
		&	J_1 \\
		&\lesssim E\left(
		\int 
		K\left(\mathbf{u}_1\right)
		K\left(\mathbf{u}_2\right)
		\frac{I_i}{p^2(\mathbf{z}_i)} [
		S_{2}(\mathbf{z}_i+\tilde h\mathbf{u}_1)+ 
		S_{2}(\mathbf{z}_i+\tilde h\mathbf{u}_2)
		] p(\mathbf{z}_i+\tilde h\mathbf{u}_1)
		p(\mathbf{z}_i+\tilde h\mathbf{u}_2)
		d\mathbf{u}_1d\mathbf{u}_2
		- 
		2S_{2}(\mathbf{z}_i ) 
		\right)^2\\
		&+
		E\left(
		\int 
		K\left(\mathbf{u}_1\right)
		K\left(\mathbf{u}_2\right)
		\frac{I(p(\mathbf{z}_i+\tilde h\mathbf{u}_1)>\tilde b) }{ p^2(\mathbf{z}_i+\tilde h\mathbf{u}_1)}
		[
		S_{3}(W_i,\mathbf{z}_i+\tilde h\mathbf{u}_1+\tilde h\mathbf{u}_2)+ 
		f_i(W_i,\pi)R_i
		]	\right. \\&\left.
		p(\mathbf{z}_i+\tilde h\mathbf{u}_1)
		p(\mathbf{z}_i+\tilde h\mathbf{u}_1+\tilde h\mathbf{u}_2)
		d\mathbf{u}_1d\mathbf{u}_2
		-
		[S_3(W_i,\mathbf{z}_i)+f_i(W_i,\pi)R_i]
		\right)^2 \\
		&\lesssim 
		\int \left(
		K\left(\mathbf{u}_1\right)
		K\left(\mathbf{u}_2\right)
		\left[
		\frac{I_i}{p^2(\mathbf{z}_i)} [
		S_{2}(\mathbf{z}_i+\tilde h\mathbf{u}_1)+ 
		S_{2}(\mathbf{z}_i+\tilde h\mathbf{u}_2)
		] p(\mathbf{z}_i+\tilde h\mathbf{u}_1)
		p(\mathbf{z}_i+\tilde h\mathbf{u}_2)
		-
		2S_{2}(\mathbf{z}_i ) 
		\right]\right)^2
		\\&
		p(\mathbf{z}_i)
		d\mathbf{u}_1d\mathbf{u}_2d\mathbf{z}_i
		\\&+
		\int\left(
		K\left(\mathbf{u}_1\right)
		K\left(\mathbf{u}_2\right)
		\left[
		\frac{I(p(\mathbf{z}_i+\tilde h\mathbf{u}_1)>\tilde b) }{ p^2(\mathbf{z}_i+\tilde h\mathbf{u}_1)}
		\right.\right. \\&\left.\left.[
		S_{3}(W_i,\mathbf{z}_i+\tilde h\mathbf{u}_1+\tilde h\mathbf{u}_2)+ 
		f_i(W_i,\pi)R_i
		]
		p(\mathbf{z}_i+\tilde h\mathbf{u}_1)
		p(\mathbf{z}_i+\tilde h\mathbf{u}_1+\tilde h\mathbf{u}_2)
		-
		[S_3(W_i,\mathbf{z}_i)+f_i(W_i,\pi)R_i]
		\right]
		\right)^2 
		\\&
		p(\mathbf{z}_i) d\mathbf{u}_1d\mathbf{u}_2d\mathbf{z}_i,
	\end{align*}
	where  in the last inequality we apply Jensen's inequality to exchange the order of integration and squaring. Similar to the arguments in (\ref{eqn:26}), using the fact that $\tilde h=o(\tilde b)$, we can find  integrable dominating functions for the two terms on the right-hand side of the above equation, allowing us to apply the dominated convergence theorem. By the dominated convergence theorem, we have $J_1=o(1)$. Then the proof is complete by referring (\ref{lem4th6:5 J2}), (\ref{eqn:18}) and (\ref{lem4th6:5 eqn0 new}).
	
\end{proof}

\begin{lemma}
		\label{lem4th6:9}
			Under the assumptions for Theorem~\ref{th:5}, we have  
	\begin{align*} 
		\sqrt{n}(\delta_{22}-E\delta_{22})=
		\frac{1}{\sqrt{n}}\sum_{i=1}^n 
		\left(\frac{M_i}{N_i}-\pi\right)
		E\left(
		f_j^{(1)}\left(W_j,\pi\right)
		\left(\frac{W_j}{\pi}-\frac{1-W_j }{1-\pi }\right)
		\right)+o_p(1),
	\end{align*}
	where 
	\begin{align*}
			\delta_{22}&=\frac{1}{n^3\tilde h^{2p}}\sum_{i=1}^n\sum_{j=1}^n
		\sum_{k=1}^n 
		\frac{I(p(\mathbf{z}_i)>\tilde b)}{p^2(\mathbf{z}_i)} K\left(\frac{\mathbf{z}_i-\mathbf{z}_j}{\tilde h}\right)
		K\left(\frac{\mathbf{z}_i-\mathbf{z}_k}{\tilde h}\right)
		f_j^{(1)}\left(W_j,\pi\right)\left(\frac{M_j}{N_j}-\pi\right)\\&
		\left(\frac{W_jW_k}{\pi^2}-
		\frac{(1-W_j)(1-W_k)}{(1-\pi)^2}\right).
	\end{align*}
\end{lemma}
\begin{proof}
	Let $R_{j}=\frac{W_j }{\pi}-\frac{1-W_j}{1-\pi}, I_j=I(p(\mathbf{z}_j)>\tilde b),  R_{jk}=\frac{W_jW_k}{\pi^2}-\frac{(1-W_j)(1-W_k)}{(1-\pi)^2}$ for $1\le j,k\le n$, 
	\begin{align*}
	U_{22}=\frac{1}{n(n-1)(n-2)}\sum_{i,j,k \text{ are distinct }}
	\frac{I_i}{\tilde h^{2p}p^2(\mathbf{z}_i)} K\left(\frac{\mathbf{z}_i-\mathbf{z}_j}{\tilde h}\right)
	K\left(\frac{\mathbf{z}_i-\mathbf{z}_k}{\tilde h}\right)
	f_j^{(1)}\left(W_j,\pi\right)\left(\frac{M_j}{N_j}-\pi\right)R_{jk}.
	\end{align*}
	Then
	\begin{align*} 
			&\delta_{22}-\frac{(n-1)(n-2)}{n^2}U_{22}\\
			&\lesssim \frac{1}{n^3}
			\sum_{i\neq j}	\frac{I_i}{\tilde h^{2p}p^2(\mathbf{z}_i)}
			K^2\left(\frac{\mathbf{z}_i-\mathbf{z}_j}{\tilde h}\right) f_j^{(1)}(W_j,\pi)
			\left(\frac{M_j}{N_j}-\pi\right)\left(\frac{W_j}{\pi^2}-\frac{1-W_j}{(1-\pi)^2}\right)\\
			&+\frac{1}{n^3}
			\sum_{i\neq j}	\frac{I_i}{\tilde h^{2p}p^2(\mathbf{z}_i)}K(\mathbf{0})
			K\left(\frac{\mathbf{z}_i-\mathbf{z}_j}{\tilde h}\right)\left[
			f_j^{(1)}(W_j,\pi)
			\left(\frac{M_j}{N_j}-\pi\right)+
			f_i^{(1)}(W_i,\pi)
			\left(\frac{M_i}{N_i}-\pi\right)
			\right]R_{ij}\\
			&+
			\frac{1}{n^3}
			\sum_{i=1}^n	\frac{I_i}{\tilde h^{2p}p^2(\mathbf{z}_i)}K^2(\mathbf{0})\left(\frac{M_i}{N_i}-\pi\right)
			\left(\frac{W_i}{\pi^2}-\frac{1-W_i}{(1-\pi)^2}\right).
	\end{align*}
	Similar to the proofs in Lemma~\ref{lem4th6:6}, we have that 
	\begin{align*}
		&E\left|\delta_{22}-\frac{(n-1)(n-2)}{n^2}U_{22}\right|
			\\&
			\lesssim
	\frac{1}{n^3\tilde h^p\tilde b}\sum_{i,j:i\neq j} E\frac{	K^2\left(\frac{\mathbf{z}_i-\mathbf{z}_j}{\tilde h}\right)}{\tilde h^p p(\mathbf{z}_i)}E\left|\frac{M_j}{N_j}-\pi\right|
	+
	\frac{1}{n^3\tilde h^p\tilde b}\sum_{i,j:i\neq j} E\frac{	|K\left(\frac{\mathbf{z}_i-\mathbf{z}_j}{\tilde h}\right)|}{\tilde h^p p(\mathbf{z}_i)}E\left|\frac{M_j}{N_j}-\pi\right|
	\\&
	+	\frac{1}{n^2} 	\frac{1}{\tilde h^{2p}\tilde b^2}E\left|\mni-\pi\right|\\
	&\lesssim \frac{1}{n\tilde h^p\tilde b\sqrt{n\rho_n}} \int K^2(\mathbf{u})p(\mathbf{z}_j)d\mathbf{u}d\mathbf{z}_j 
	+ \frac{1}{n\tilde h^p\tilde b\sqrt{n\rho_n}} \int |K(\mathbf{u})|p(\mathbf{z}_j)d\mathbf{u}d\mathbf{z}_j 
	+	 
	 	\frac{1}{n^2\tilde h^{2p}\tilde b^2\sqrt{n\rho_n}}\\
	 	&\lesssim
	 	 \frac{1}{n\tilde h^p\tilde b\sqrt{n\rho_n}} +	\frac{1}{n^2\tilde h^{2p}\tilde b^2\sqrt{n\rho_n}}=o(n^{-1/2}).
	\end{align*}
	Therefore, \[\sqrt{n}(\delta_{22}-E\delta_{22})-\sqrt{n}(U_{22}-EU_{22})=o_p(1).\]
	
	We now use the  projection technique to approximate $U_{22}$.  Let
	\begin{equation}
		\label{eqn:rn U22}
		r_{n}=
		\frac{1}{n(n-1)(n-2) \tilde h^{2p} }\sum_{i,j,k \text{ are distinct }}
		\left(\frac{M_j}{N_j}-\pi\right)
		\frac{I_i}{p^2(\mathbf{z}_i)}
		E\left(
		K\left(\frac{\mathbf{z}_i-\mathbf{z}_j}{\tilde h}\right)
	K\left(\frac{\mathbf{z}_i-\mathbf{z}_k}{\tilde h}\right)
		f_j^{(1)}\left(W_j,\pi\right)R_{jk}
		\bigg|\mathbf{z}_i\right).
	\end{equation}
	We prove that $U_{22}-r_n$ is negligible by showing that $E\left(U_{22}-r_n\right)^2=o(n^{-1})$.
	\begin{align}
    \label{eqn:128}
		&E\left(U_{22}-r_n\right)^2\notag\\&=
		E\left(\frac{1}{n(n-1)(n-2)}\sum_{i,j,k \text{ are distinct }}
		\frac{I_i}{\tilde h^{2p}p^2(\mathbf{z}_i)}\left(\frac{M_j}{N_j}-\pi\right)\right.\notag\\&\left.
		\left\{
		K\left(\frac{\mathbf{z}_i-\mathbf{z}_j}{\tilde h}\right)
		K\left(\frac{\mathbf{z}_i-\mathbf{z}_k}{\tilde h}\right)
		f_j^{(1)}\left(W_j,\pi\right)R_{jk}
		-E\left[K\left(\frac{\mathbf{z}_i-\mathbf{z}_j}{\tilde h}\right)
		K\left(\frac{\mathbf{z}_i-\mathbf{z}_k}{\tilde h}\right)
		f_j^{(1)}\left(W_j,\pi\right)R_{jk}\bigg|\mathbf{z}_i\right]
		\right\}\right)^2\notag\\
		&\lesssim
		\frac{1}{n^6}\sum_{i,j,k \text{ are distinct }}\sum_{i_1,j_1,k_1 \text{ are distinct }} E\left(
		\frac{I_i}{\tilde h^{2p}p^2(\mathbf{z}_i)}	\frac{I_{i_1}}{\tilde h^{2p}p^2(\mathbf{z}_{i_1})}
	\frac{
	\sum_{l:l\neq j} (W_l-\pi)E_{jl}	\sum_{l_1:l_1\neq j_1} (W_{l_1}-\pi)E_{j_1l_1}
		}{N_jN_{j_1}}
		\right.\notag\\&\left.
				\left\{
		K\left(\frac{\mathbf{z}_i-\mathbf{z}_j}{\tilde h}\right)
		K\left(\frac{\mathbf{z}_i-\mathbf{z}_k}{\tilde h}\right)
		f_j^{(1)}\left(W_j,\pi\right)R_{jk}
		-E\left[K\left(\frac{\mathbf{z}_i-\mathbf{z}_j}{\tilde h}\right)
		K\left(\frac{\mathbf{z}_i-\mathbf{z}_k}{\tilde h}\right)
		f_j^{(1)}\left(W_j,\pi\right)R_{jk}\bigg|\mathbf{z}_i\right]
		\right\}
		\right.\notag\\&\left.
				\left\{
		K\left(\frac{\mathbf{z}_{i_1}-\mathbf{z}_{j_1}}{\tilde h}\right)
		K\left(\frac{\mathbf{z}_{i_1}-\mathbf{z}_{k_1}}{\tilde h}\right)
		f_{j_1}^{(1)}\left(W_{j_1},\pi\right)R_{j_1k_1}
		\right.	\right.\notag\\&\left.\left.
		-E\left[K\left(\frac{\mathbf{z}_{i_1}-\mathbf{z}_{j_1}}{\tilde h}\right)
		K\left(\frac{\mathbf{z}_{i_1}-\mathbf{z}_{k_1}}{\tilde h}\right)
		f_{j_1}^{(1)}\left(W_{j_1},\pi\right)R_{j_1k_1}\bigg|\mathbf{z}_{i_1}\right]
		\right\}
		\right).
	\end{align}
	Similar to the arguments after (\ref{eqn:8}), most of the summands are zero. We notice that for the summand to be nonzero, the eight indices \(i, j, k, i_1, j_1, k_1, l, l_1\) can have at most six distinct values. Then (\ref{eqn:128}) reduces to
    \begin{align*}
        &E\left(U_{22}-r_n\right)^2 \lesssim \frac{1}{n^6 }
                (A+B+C+D)
    \end{align*}
    where 
    \begin{align*}
        A=&\sum_{\mathcal{M}
				}  
                E\left(
		\frac{I_i}{\tilde h^{2p}p^2(\mathbf{z}_i)}	\frac{I_{i_1}}{\tilde h^{2p}p^2(\mathbf{z}_{i_1})}
	\frac{
	 (W_l-\pi)E_{jl}	
     (W_{l_1}-\pi)E_{j_1l_1}
		}{N_jN_{j_1}}
		\right.\\&\left.
		K\left(\frac{\mathbf{z}_i-\mathbf{z}_j}{\tilde h}\right)
		K\left(\frac{\mathbf{z}_i-\mathbf{z}_k}{\tilde h}\right)
		f_j^{(1)}\left(W_j,\pi\right)R_{jk}
K\left(\frac{\mathbf{z}_{i_1}-\mathbf{z}_{j_1}}{\tilde h}\right)
		K\left(\frac{\mathbf{z}_{i_1}-\mathbf{z}_{k_1}}{\tilde h}\right)
		f_{j_1}^{(1)}\left(W_{j_1},\pi\right)R_{j_1k_1}
		\right),\\
        \mathcal{M}&=\{i, j, k, i_1, j_1, k_1, l, l_1 \text{   have at most six distinct values } , i,j,k,\text{ distinct}, i_1,j_1,k_1\text{ distinct}\},
    \end{align*}
    and $B,C,D$ are defined in the obvious way. For $A$, since \[\sup_{\mathbf{z}_j}\max_{w=0,1}E(|f_j^{(1)}(W_j,\pi)|^2|W_j=w,\mathbf{z}_j), R_{jk}, W_l-\pi \text{ are bounded },\] we note that
    \begin{align*}
          A\lesssim&\sum_{\mathcal{M}} 
                E\left(
		\frac{I_i}{\tilde h^{2p}p^2(\mathbf{z}_i)}	\frac{I_{i_1}}{\tilde h^{2p}p^2(\mathbf{z}_{i_1})}
	\frac{E_{jl}E_{j_1l_1}
		}{N_jN_{j_1}}
		K\left(\frac{\mathbf{z}_i-\mathbf{z}_j}{\tilde h}\right)
		K\left(\frac{\mathbf{z}_i-\mathbf{z}_k}{\tilde h}\right)
K\left(\frac{\mathbf{z}_{i_1}-\mathbf{z}_{j_1}}{\tilde h}\right)
		K\left(\frac{\mathbf{z}_{i_1}-\mathbf{z}_{k_1}}{\tilde h}\right)
		\right)\\
        &=\frac{1}{\tilde b^2}\sum_{\mathcal{M}} 
                E\left(\frac{E_{jl}E_{j_1l_1}
		}{N_jN_{j_1}}\right)
                E\left(
		\frac{I_i}{\tilde h^{2p}p(\mathbf{z}_i)}	\frac{I_{i_1}}{\tilde h^{2p}p(\mathbf{z}_{i_1})}
		K\left(\frac{\mathbf{z}_i-\mathbf{z}_j}{\tilde h}\right)
		K\left(\frac{\mathbf{z}_i-\mathbf{z}_k}{\tilde h}\right)
        \right.\\&\left.
K\left(\frac{\mathbf{z}_{i_1}-\mathbf{z}_{j_1}}{\tilde h}\right)
		K\left(\frac{\mathbf{z}_{i_1}-\mathbf{z}_{k_1}}{\tilde h}\right)
		\right).
    \end{align*}
    Case 1: $i,j,k,i_1,j_1,k_1$ are distinct. Since
    \begin{align}\label{eqn:129}
         &E\left(
		\frac{I_i}{\tilde h^{2p}p(\mathbf{z}_i)}	
		K\left(\frac{\mathbf{z}_i-\mathbf{z}_j}{\tilde h}\right)
		K\left(\frac{\mathbf{z}_i-\mathbf{z}_k}{\tilde h}\right)\right)\le 
        \int K(\mathbf{u}_1)K(\mathbf{u}_2)p(\mathbf{z}_i+\tilde h\mathbf{u}_1)p(\mathbf{z}_i+\tilde h\mathbf{u}_2)d\mathbf{u}_1d\mathbf{u}_2d\mathbf{z}_i
    \end{align}
    is bounded, we have
    \begin{align*}
        A\lesssim \frac{1}{\tilde b^2}\sum_{\mathcal{M}} 
                E\left(\frac{E_{jl}E_{j_1l_1}
		}{N_jN_{j_1}}\right).
    \end{align*}
   Case 2: At least one of \(i_1, j_1, k_1\) is equal to one of \(i, j, k\). Since $n\tilde h^{2p}\tilde b\to \infty$ by Assumption~\ref{ass:main3}, and $K(\cdot)$ is bounded, we have 
 \begin{align*}
      A&\lesssim  \frac{n^2}{\tilde h^{2p}\tilde b^3}\sum_{
				i,j,k,\text{ distinct}} 
                E\left(\frac{E_{jl}E_{j_1l_1}
		}{N_jN_{j_1}}\right)
                E\left(
		\frac{I_i}{\tilde h^{2p}p(\mathbf{z}_i)}	 
		K\left(\frac{\mathbf{z}_i-\mathbf{z}_j}{\tilde h}\right)
		K\left(\frac{\mathbf{z}_i-\mathbf{z}_k}{\tilde h}\right)\right)\\&
       \lesssim \frac{1}{\tilde b^2}\sum_{\mathcal{M}}  E\left(\frac{E_{jl}E_{j_1l_1}
		}{N_jN_{j_1}}\right),
    \end{align*}
    where in the last inequality we use (\ref{eqn:129}). The analysis for \(B\), \(C\), and \(D\) is similar. Finally, we conclude that
    \begin{align*}
        E\left(U_{22}-r_n\right)^2 \lesssim \frac{1}{n^6\tilde b^2}\sum_{\mathcal{M}} 	 E\left(
			\frac{
			  E_{jl}E_{j_1l_1}
			}{N_jN_{j_1}}
			\right).
    \end{align*}
 By Lemma 15 in \cite{li2022random},  \[ E(E_{ij})\lesssim \rho_n,  P(N_i<cn\rho_n)<\exp(-Cn\rho_n)\] for some constants $c,C>0$. 
	We then derive that when $n$ is sufficiently large, 
		\begin{align}
			\label{eqn:new3}
			&E\left(U_{22}-r_n\right)^2\notag\\
			&\lesssim \frac{1}{n^6\tilde b^2}\sum_{\mathcal{M}} 	 E\left(
			\frac{
			  E_{jl}E_{j_1l_1}
			}{N_jN_{j_1}}
			\right)\notag\\
			&\le \frac{1}{n^6\tilde b^2}\sum_{\mathcal{M}} 	 E\left(
			\frac{
				E_{jl}
			}{N_jN_{j_1}}
			\right)\notag\\
			&\lesssim \frac{1}{n^6\tilde b^2}\sum_{\mathcal{M}} 	
				[ \frac{\rho_n}{n^2\rho_n^2}+\rho_n\exp(-Cn\rho_n)]\notag\\
			&\lesssim 
			 \frac{1}{n^2\rho_n\tilde b^2}=o(n^{-1}),
		\end{align}  
		where in the last equality we use the fact that $n\rho_n\tilde b^2\to\infty$ which is implied by Assumption~\ref{assump:random_graph} and Assumption~\ref{ass:main3}.  Therefore, 
			\begin{align*}
			\sqrt{n}(U_{22}-EU_{22})-\sqrt{n}(r_n-Er_n)=o_p(1).
		\end{align*} 
		By following the procedure from (\ref{lem4th6:7:eqn2}) to (\ref{eqn:210}), we have that
		\begin{align*} 
				\sqrt{n}(r_n-Er_n)&=\frac{1}{\sqrt n}\sum_{j=1}^n
				\left(\frac{M_j}{N_j}-\pi\right)
				E\left(
				f_j^{(1)}\left(W_j,\pi\right)R_{jk}
				\right)+o_p(1)\\
				&=\frac{1}{\sqrt n}\sum_{j=1}^n
				\left(\frac{M_j}{N_j}-\pi\right)
				E\left(
				f_j^{(1)}\left(W_j,\pi\right)R_{j}
				\right)+o_p(1).
		\end{align*}
		Then the proof is complete. 
\end{proof}

\begin{lemma}
	\label{lem4th6:1}
	Under the assumptions for Theorem~\ref{th:5}, we have 
	\begin{align*}
		\sup_{\mathbf{z}\in\bb{R}^p}\left| 	\hat p(\mathbf{z})-p(\mathbf{z})\right|&=O_{a.s.}\left(\sqrt{\frac{\log(n)}{n \tilde h^p}}+\tilde h^q\right),\\
		\sup_{\mathbf{z}\in\bb{R}^p}\left| 	\hat p_1(\mathbf{z})-p(\mathbf{z})\right|&=O_{a.s.}\left(\sqrt{\frac{\log(n)}{n \tilde h^p}}+\tilde h^q\right),\\
		\sup_{\mathbf{z}\in\bb{R}^p}\left| 	\hat p_2(\mathbf{z})-p(\mathbf{z})\right|&=O_{a.s.}\left(\sqrt{\frac{\log(n)}{n \tilde h^p}}+\tilde h^q\right),
	\end{align*}
	where 
	\begin{align*}
		\hat p(\mathbf{z})&= \frac{1}{n   \tilde h^p}\sum_{j=1}^n	K\left(\frac{\mathbf{z}-\mathbf{z}_j}{\tilde h}\right),\\
		\hat p_1(\mathbf{z})&= \frac{1}{n   \tilde h^p\pi}\sum_{j=1}^n	K\left(\frac{\mathbf{z}-\mathbf{z}_j}{\tilde h}\right)W_j, \\
		\hat p_2(\mathbf{z})&= \frac{1}{n \tilde h^p(1-\pi)}\sum_{j=1}^n	K\left(\frac{\mathbf{z}-\mathbf{z}_j}{\tilde h}\right)(1-W_j).
	\end{align*}
\end{lemma}
\begin{proof}
	We only show that \[	\sup_{\mathbf{z}\in\bb{R}^p}\left| 	\hat p_1(\mathbf{z})-p(\mathbf{z})\right|=O_{a.s.}\left(\sqrt{\frac{\log(n)}{n \tilde h^p}}+\tilde h^q\right),\] as the other two results can be derived similarly. First, the conditions for Theorem 5 in \cite{hansen2008uniform} are satisfied under our Assumption~\ref{ass:main3}. Consequently, we have \[
    	\sup_{\mathbf{z}\in\bb{R}^p}\left| 	\hat p_1(\mathbf{z})-E\hat p_1(\mathbf{z})\right|=O_{a.s.}\left(\sqrt{\frac{\log(n)}{n \tilde h^p}}\right).
    \]
    Second, by integration by parts and a change of variables, we derive that
    \begin{align*}
        E\hat p_1(\mathbf{z})-p(\mathbf{z})&=\frac{1}{n   \tilde h^p\pi}\sum_{j=1}^n	E\left[K\left(\frac{\mathbf{z}-\mathbf{z}_j}{\tilde h}\right)W_j\right]\\
        &=\frac{1}{n   \tilde h^p}\sum_{j=1}^n	E\left[K\left(\frac{\mathbf{z}-\mathbf{z}_j}{\tilde h}\right)\right]\\
        &=\frac{1}{  \tilde h^p} 	\int K\left(\frac{\mathbf{z}-\mathbf{z}_j}{\tilde h}\right)p(\mathbf{z}_j)d\mathbf{z}_j\\
        &=\int K(\mathbf{u})p(\mathbf{z}+\tilde h\mathbf{u})d\mathbf{u}=p(\mathbf{z})+O(\tilde h^q),
    \end{align*}
    where the final equality is by a $q$-{th} order Taylor series expansion and using the assumed properties of the kernel and $p(\mathbf{z})$. 
	Then the results follows. 
\end{proof}

\begin{lemma}
	\label{lem:new3}
 {	Let \[
		\delta_{12}=\frac{2}{n^2 \tilde h^p}\sum_{i=1}^n\sum_{j=1}^n \frac{1}{p(\mathbf{z}_i)} K\left(\frac{\mathbf{z}_i-\mathbf{z}_j}{\tilde h}\right)f_j^{(1)}\left(W_j,\pi\right)\left(\frac{M_j}{N_j}-\pi\right)\left(\frac{W_j}{\pi}-\frac{1-W_j}{1-\pi}\right)	I\left(	p(\mathbf{z}_i)>\tilde b\right),\]
			\begin{align*}
			\delta_{22}&=\frac{1}{n^3\tilde h^{2p}}\sum_{i=1}^n\sum_{j=1}^n
			\sum_{k=1}^n 
			\frac{I(p(\mathbf{z}_i)>\tilde b)}{p^2(\mathbf{z}_i)} K\left(\frac{\mathbf{z}_i-\mathbf{z}_j}{\tilde h}\right)
			K\left(\frac{\mathbf{z}_i-\mathbf{z}_k}{\tilde h}\right)
			f_j^{(1)}\left(W_j,\pi\right)\left(\frac{M_j}{N_j}-\pi\right)\\&
			\left(\frac{W_jW_k}{\pi^2}-
			\frac{(1-W_j)(1-W_k)}{(1-\pi)^2}\right).
		\end{align*}
		 Under the assumptions for Theorem~\ref{th:5}, we have $E	\delta_{12}=0, E	\delta_{12}=o(n^{-1/2})$.
		}
\end{lemma}
\begin{proof}
	For $\delta_{12}$, note that we let $G_n$ denote the random graph, and $M_j-\pi N_j=\sum_{k:k\neq j}(W_k-\pi)E_{kj}$ given $G_n$. Since $W_k$'s are sampled independently, we have $E\delta_{12}=0$. 
	For $\delta_{22}$, by noting the above argument for $\delta_{12}$, we derive that 
	\begin{align*}
			E\delta_{22}&=\frac{1}{n^3\tilde h^{2p}}\sum_{i=1}^n\sum_{j=1}^n
		\sum_{k=1}^n E\left\{
		\frac{I(p(\mathbf{z}_i)>\tilde b)}{p^2(\mathbf{z}_i)} K\left(\frac{\mathbf{z}_i-\mathbf{z}_j}{\tilde h}\right)
		K\left(\frac{\mathbf{z}_i-\mathbf{z}_k}{\tilde h}\right)
		f_j^{(1)}\left(W_j,\pi\right)\left(\frac{M_j}{N_j}-\pi\right)\right.\\&\left.
		\left(\frac{W_jW_k}{\pi^2}-
		\frac{(1-W_j)(1-W_k)}{(1-\pi)^2}\right)\right\}\\
&=\frac{1}{n^3\tilde h^{2p}}\sum_{i=1}^n\sum_{j=1}^n
\sum_{k=1}^n E\left\{
\frac{I(p(\mathbf{z}_i)>\tilde b)}{p^2(\mathbf{z}_i)} K\left(\frac{\mathbf{z}_i-\mathbf{z}_j}{\tilde h}\right)
K\left(\frac{\mathbf{z}_i-\mathbf{z}_k}{\tilde h}\right)
f_j^{(1)}\left(W_j,\pi\right)E_{kj}\frac{W_k-\pi}{N_j}\right.\\&\left.
\left(\frac{W_jW_k}{\pi^2}-
\frac{(1-W_j)(1-W_k)}{(1-\pi)^2}\right) \right\}.
	\end{align*}
Noting that \(E(E_{kj}/N_j) = O(1/n)\)  which can be derived using similar arguments as those in equation (\ref{eqn:new3}), and that $E_{kj}/N_j$ is determined by graph $G_n$, we then have
\begin{align*}
		|E\delta_{22}|&\lesssim \frac{1}{n} 
		\frac{1}{n^3\tilde h^{2p}}\sum_{i=1}^n\sum_{j=1}^n
		\sum_{k=1}^n E\left|
		\frac{I(p(\mathbf{z}_i)>\tilde b)}{p^2(\mathbf{z}_i)} K\left(\frac{\mathbf{z}_i-\mathbf{z}_j}{\tilde h}\right)
		K\left(\frac{\mathbf{z}_i-\mathbf{z}_k}{\tilde h}\right)
		f_j^{(1)}\left(W_j,\pi\right)\right|\\
&\lesssim \frac{1}{n} 
\frac{1}{n^3\tilde h^{2p}}\sum_{i=1}^n\sum_{j=1}^n
\sum_{k=1}^n E\left|
\frac{I(p(\mathbf{z}_i)>\tilde b)}{p^2(\mathbf{z}_i)} K\left(\frac{\mathbf{z}_i-\mathbf{z}_j}{\tilde h}\right)
K\left(\frac{\mathbf{z}_i-\mathbf{z}_k}{\tilde h}\right)
\right|\\
	&	\lesssim \frac{1}{n} 
		E\left|
		\frac{1}{ \tilde h^{2p}}	\frac{1}{p^2(\mathbf{z}_1)} K\left(\frac{\mathbf{z}_1-\mathbf{z}_2}{\tilde h}\right)
		K\left(\frac{\mathbf{z}_1-\mathbf{z}_3}{\tilde h}\right)
	 \right|\\
	 &+
	 \frac{1}{n^2} 
	 E\left|
	 \frac{1}{ \tilde h^{2p}}	\frac{I(p(\mathbf{z}_i)>\tilde b)}{p^2(\mathbf{z}_1)} K^2\left(\frac{\mathbf{z}_1-\mathbf{z}_2}{\tilde h}\right) 
	 \right|
	 +
	 \frac{1}{n^2} 
	 E\left|
	 \frac{1}{ \tilde h^{2p}}	\frac{I(p(\mathbf{z}_i)>\tilde b)}{p^2(\mathbf{z}_1)} K\left(\frac{\mathbf{z}_1-\mathbf{z}_2}{\tilde h}\right)
	 K\left(0\right)
	 \right|\\
	& +
	 \frac{1}{n^3} 
	 \frac{1}{ \tilde h^{2p}}	\frac{1}{\tilde b^2}
	 K^2\left(0\right)
	\\
	 &=O\left(\frac{1}{n\tilde b}+\frac{1}{n^2\tilde h^p\tilde b}+\frac{1}{n^3\tilde h^{2p}\tilde b^2}\right)=O(1/\sqrt{n}).
\end{align*}
\end{proof}

\begin{lemma}
	\label{lem:new1}
	 {
	Let \begin{align*}
		\delta_{11}=\frac{2}{n^2 \tilde h^p}\sum_{i=1}^n\sum_{j=1}^n \frac{1}{p(\mathbf{z}_i)} K\left(\frac{\mathbf{z}_i-\mathbf{z}_j}{\tilde h}\right)f_j\left(W_j,\pi\right)\left(\frac{W_j}{\pi}-\frac{1-W_j}{1-\pi}\right)	I\left(	p(\mathbf{z}_i)>\tilde b\right).
	\end{align*}
	Under the assumptions of Theorem~\ref{th:5}, we have $E\delta_{11}=2\tau+o(n^{-\frac{1}{2}})$.
}
\end{lemma}
\begin{proof}
	Let $R_j=\frac{W_j}{\pi}-\frac{1-W_j}{1-\pi}, I_j=I\left(	p(\mathbf{z}_j)>\tilde b\right), S_2(u)=E(f_i(W_i,\pi)R_i|\mathbf{z}_i)\bigg|_{\mathbf{z}_i=u}=E(f_i(1,\pi)-f_i(0,\pi)|\mathbf{z}_i)\bigg|_{\mathbf{z}_i=u}$. 
	\begin{align*}
		E\delta_{11}&=\frac{2}{n \tilde h^p}E\left\{\frac{1}{p(\mathbf{z}_i)} K\left(0\right)f_i\left(W_i,\pi\right)\left(\frac{W_i}{\pi}-\frac{1-W_i}{1-\pi}\right)	I\left(	p(\mathbf{z}_i)>\tilde b\right)\right\}\\
		&+\frac{2n(n-1)}{n^2 \tilde h^p} E\left\{\frac{1}{p(\mathbf{z}_1)} K\left(\frac{\mathbf{z}_1-\mathbf{z}_2}{\tilde h}\right)f_2\left(W_2,\pi\right)\left(\frac{W_2}{\pi}-\frac{1-W_2}{1-\pi}\right)	I\left(	p(\mathbf{z}_1)>\tilde b\right)\right\}\\
		&=O\left(\frac{1}{n\tilde h^p\tilde b}\right)+
		\left(2-\frac{2}{n}\right)E\left[
		\frac{I_1}{\tilde h^p p(\mathbf{z}_1)} K\left(\frac{\mathbf{z}_1-\mathbf{z}_2}{\tilde h}\right) S_2(\mathbf{z}_2)
		\right]\\
	&=	O\left(\frac{1}{n\tilde h^p\tilde b}\right)
		+
		\left(2-\frac{2}{n}\right)\iint
	 {I(p(\mathbf{z}_2+\tilde h\mathbf{u})>\tilde{b})}  K\left(\mathbf{u}\right) S_2(\mathbf{z}_2)p(\mathbf{z}_2)
	 d\mathbf{z}_1d\mathbf{z}_2.
	\end{align*}
	By Taylor expansion, we have $p(\mathbf{z}_2+\tilde h\mathbf{u})=p(\mathbf{z}_2)+\tilde h\mathbf{u}^\top \nabla p(\xi)$, where $\xi$ is a vector in the neighbourhood of $\mathbf{z}_2$. By Assumption~\ref{ass:main3}, $\nabla p(\xi)$ is bounded and $K(\mathbf{u})=0$ for $u\in \{u: |u|_\infty\ge 1\}$,  therefore, $|\mathbf{u}^\top \nabla p(\xi)|$ is bounded. Moreover, $\tilde h=o(\tilde b)$, therefore, 
	\begin{align}
		\label{eqn:new2}
		&	\iint_{p(\mathbf{z}_2+\tilde h\mathbf{u})<\tilde{b}}  K\left(\mathbf{u}\right) S_2(\mathbf{z}_2)p(\mathbf{z}_2)
			d\mathbf{u}d\mathbf{z}_2  \lesssim
				\iint_{p(\mathbf{z}_2)<1.01\tilde{b}}  K\left(\mathbf{u}\right) |S_2(\mathbf{z}_2)|p(\mathbf{z}_2)
			d\mathbf{u}d\mathbf{z}_2\notag\\
&=	\int_{p(\mathbf{z}_2)<1.01\tilde{b}}   |S_2(\mathbf{z}_2)|p(\mathbf{z}_2) d\mathbf{z}_2
			=o(n^{-1/2}),
	\end{align}
	where in the last equality we use Assumption~\ref{ass:main3}. 
Note that $\int  S_2(\mathbf{z}_2)p(\mathbf{z}_2)d\mathbf{z}_2=\tau$. By combining all the equations above, we have  \begin{align*}
			E\delta_{11}&=O\left(\frac{1}{n\tilde h^p\tilde b}\right)
			+
			\left(2-\frac{2}{n}\right)\iint K\left(\mathbf{u}\right) S_2(\mathbf{z}_2)p(\mathbf{z}_2)
			d\mathbf{u}d\mathbf{z}_2+o(n^{-1/2})\\
			&=o(n^{-1/2})+2\tau.
	\end{align*}
	
	\end{proof}
	
	\begin{lemma}
		\label{lem:new2}
	 {
		Let 
		\begin{align*}
			\delta_{21}&=\frac{1}{n^3\tilde h^{2p}}\sum_{i=1}^n\sum_{j=1}^n
			\sum_{k=1}^n 
			\frac{I(p(\mathbf{z}_i)>\tilde b)}{p^2(\mathbf{z}_i)} K\left(\frac{\mathbf{z}_i-\mathbf{z}_j}{\tilde h}\right)
			K\left(\frac{\mathbf{z}_i-\mathbf{z}_k}{\tilde h}\right)
			f_j\left(W_j,\pi\right)
			\left(\frac{W_jW_k}{\pi^2}-
			\frac{(1-W_j)(1-W_k)}{(1-\pi)^2}\right).
		\end{align*} 
		Under the assumptions of Theorem~\ref{th:5}, we have $E\delta_{21}=\tau+o(n^{-\frac{1}{2}})$. 
	}
	\end{lemma}
	\begin{proof}
		Let $R_j=\frac{W_j}{\pi}-\frac{1-W_j}{1-\pi}, I_j=I\left(	p(\mathbf{z}_j)>\tilde b\right), S_1(u)=E(f_i(W_i,\pi)R_i|\mathbf{z}_i)\bigg|_{\mathbf{z}_i=u}=E(f_i(1,\pi)-f_i(0,\pi)|\mathbf{z}_i)\bigg|_{\mathbf{z}_i=u}$. We derive that 
		\begin{align*}
			 E\delta_{21}&=\frac{1}{n^3\tilde h^{2p}}\sum_{i=1}^n\sum_{j=1}^n
			 \sum_{k=1}^n 
			E\left\{ \frac{I_i}{p^2(\mathbf{z}_i)} K\left(\frac{\mathbf{z}_i-\mathbf{z}_j}{\tilde h}\right)
			 K\left(\frac{\mathbf{z}_i-\mathbf{z}_k}{\tilde h}\right)
			 f_j\left(W_j,\pi\right)
			R_j\right\}\\
			&= \frac{n(n-1)(n-2)}{n^3\tilde h^{2p}}E\left\{ 
			\frac{I_1}{p^2(\mathbf{z}_1)} K\left(\frac{\mathbf{z}_1-\mathbf{z}_2}{\tilde h}\right)
			K\left(\frac{\mathbf{z}_1-\mathbf{z}_3}{\tilde h}\right)
			f_2\left(W_2,\pi\right)
			R_2\right\}\\
			&+
			\frac{n(n-1)}{n^3\tilde h^{2p}}E
		\left\{ 	\frac{I_1}{p^2(\mathbf{z}_1)} K\left(0\right)
			K\left(\frac{\mathbf{z}_1-\mathbf{z}_3}{\tilde h}\right)
			f_1\left(W_1,\pi\right)
			R_1\right\}\\
			&+
			\frac{n(n-1)}{n^3\tilde h^{2p}}E
		\left\{	\frac{I_1}{p^2(\mathbf{z}_1)} K\left(\frac{\mathbf{z}_1-\mathbf{z}_2}{\tilde h}\right)
			K\left(0\right)
			f_2\left(W_2,\pi\right)
			R_2\right\}\\
			&+\frac{n(n-1)}{n^3\tilde h^{2p}}E\left\{ 
			\frac{I_1}{p^2(\mathbf{z}_1)} K^2\left(\frac{\mathbf{z}_1-\mathbf{z}_2}{\tilde h}\right)
			f_2\left(W_2,\pi\right)
			R_2\right\}\\
		&+	\frac{n}{n^3\tilde h^{2p}}E\left\{
			\frac{I_1}{p^2(\mathbf{z}_1)} K^2\left(0\right)
			f_1\left(W_1,\pi\right)
			R_1\right\}\\
			&=A+B+C+D+E.
		\end{align*}
		For $B,C,D,E$, similar to the arguments in Lemma~\ref{lem:new3}, we have \[B=O\left(\frac{1}{n\tilde h^p\tilde b}\right),
		C=O\left(\frac{1}{n\tilde h^p\tilde b}\right),
		D=O\left(\frac{1}{n\tilde h^p\tilde b}\right),
		E=O\left(\frac{1}{n^2\tilde h^{2p}\tilde b^2}\right).
		\] By Assumption~\ref{ass:main3}, $\sqrt{n}=o(n\tilde h^p\tilde b)$, therefore, \[B=o\left(\frac{1}{\sqrt{n}}\right),
		C=o\left(\frac{1}{\sqrt{n}}\right),
		D=o\left(\frac{1}{\sqrt{n}}\right),
E=o\left(\frac{1}{\sqrt{n}}\right).
		\]
		For $A$, we derive that 
		\begin{align*}
			A&=\left(1-O\left(\frac{1}{n}\right)\right)
			\frac{1}{ \tilde h^{2p}}E\left\{
			\frac{I_1}{p^2(\mathbf{z}_1)} K\left(\frac{\mathbf{z}_1-\mathbf{z}_2}{\tilde h}\right)
			K\left(\frac{\mathbf{z}_1-\mathbf{z}_3}{\tilde h}\right)
		S_1(\mathbf{z}_2)\right\}\\
		&=\left(1-O\left(\frac{1}{n}\right)\right)
		\iiint
		\frac{p(\mathbf{z}_2+\tilde h\mathbf{u}_1+\tilde h\mathbf{u}_2)}{p(\mathbf{z}_2+\tilde h\mathbf{u}_1)}I(p(\mathbf{z}_2+\tilde h\mathbf{u}_1)>\tilde b)K(\mathbf{u}_1)K(\mathbf{u}_2)p(\mathbf{z}_2)S_1(\mathbf{z}_2) d\mathbf{z}_2d\mathbf{u}_1d\mathbf{u}_2.
		\end{align*}
		By  (\ref{eqn:new2}), we have
		\begin{align}
			\label{eqn:new4}
				\iiint
		 I(p(\mathbf{z}_2+\tilde h\mathbf{u}_1)>\tilde b)K(\mathbf{u}_1)K(\mathbf{u}_2)p(\mathbf{z}_2)|S_1(\mathbf{z}_2)| d\mathbf{z}_2d\mathbf{u}_1d\mathbf{u}_2=o(1/\sqrt{n}).
				\end{align}
				By (\ref{eqn:new4}) and the fact that $\tilde h=o(\tilde b)$, we have
					\begin{align}
					\label{eqn:new5}
					\iiint
				\frac{\tilde h}{p(\mathbf{z}_2+\tilde h\mathbf{u}_1)}	I(p(\mathbf{z}_2+\tilde h\mathbf{u}_1)>\tilde b)K(\mathbf{u}_1)K(\mathbf{u}_2)p(\mathbf{z}_2)|S_1(\mathbf{z}_2)| d\mathbf{z}_2d\mathbf{u}_1d\mathbf{u}_2=o(1/\sqrt{n}).
				\end{align}
				Since $p(\mathbf{z}_2+\tilde h\mathbf{u}_1+\tilde h\mathbf{u}_2)-p(\mathbf{z}_2+\tilde h\mathbf{u}_1)=\tilde h\mathbf{u}_2^\top \nabla p(\xi)$, and $|\mathbf{u}^\top \nabla p(\xi)|$ is bounded, by (\ref{eqn:new5}), we have 
					\begin{align}
							\label{eqn:new6}
					\iiint
					\frac{|p(\mathbf{z}_2+\tilde h\mathbf{u}_1+h\mathbf{u}_2)-
						p(\mathbf{z}_2+\tilde h\mathbf{u}_1)|
						}{p(\mathbf{z}_2+\tilde h\mathbf{u}_1)}	I(p(\mathbf{z}_2+\tilde h\mathbf{u}_1)>\tilde b)K(\mathbf{u}_1)K(\mathbf{u}_2)p(\mathbf{z}_2)|S_1(\mathbf{z}_2)| d\mathbf{z}_2d\mathbf{u}_1d\mathbf{u}_2=o(1/\sqrt{n}).
				\end{align}
				By (\ref{eqn:new6}) and the Jensen's inequality, we have
				\begin{align*}
					A&=\left(1-O\left(\frac{1}{n}\right)\right)
				\left[	\iiint
				I(p(\mathbf{z}_2+\tilde h\mathbf{u}_1)>\tilde b)K(\mathbf{u}_1)K(\mathbf{u}_2)p(\mathbf{z}_2)S_1(\mathbf{z}_2) d\mathbf{z}_2d\mathbf{u}_1d\mathbf{u}_2+o(1/\sqrt{n})\right]\\
				&=
				\left(1-O\left(\frac{1}{n}\right)\right)
				\left[	\iint
				I(p(\mathbf{z}_2+\tilde h\mathbf{u}_1)>\tilde b)K(\mathbf{u}_1)p(\mathbf{z}_2)S_1(\mathbf{z}_2) d\mathbf{z}_2d\mathbf{u}_1+o(1/\sqrt{n})\right]\\
				&=
				\left(1-O\left(\frac{1}{n}\right)\right)
				\left[\tau-	\iint
				I(p(\mathbf{z}_2+\tilde h\mathbf{u}_1)\le \tilde b)K(\mathbf{u}_1)p(\mathbf{z}_2)S_1(\mathbf{z}_2) d\mathbf{z}_2d\mathbf{u}_1+o(1/\sqrt{n})\right]\\&=\tau+o(1/\sqrt{n}),
				\end{align*}
				where in the last equality we use (\ref{eqn:new2}). Then the result follows.
	\end{proof}
\end{document}